\documentclass[11pt]{article}
\usepackage{amsmath,amsthm,amsfonts,amssymb, graphicx,mathabx,epsfig}
\usepackage{bbm}
\usepackage{bm}
\usepackage{authblk}
\usepackage{setspace}
\usepackage[usenames]{color}
\usepackage[active]{srcltx}
\usepackage{hyperref}
\usepackage{graphicx,color}
\usepackage{footmisc}
\usepackage{amsmath}
\usepackage{amsfonts}
\usepackage{amssymb}
\usepackage{bbm}
\usepackage{epstopdf}
\usepackage{color}
\usepackage{dsfont}
\numberwithin{equation}{section}
\usepackage{tikz}
\usepackage{pgfplots}
\usepackage{subfig}
\usetikzlibrary{fit}
\usetikzlibrary{shapes.geometric}
\usetikzlibrary{decorations.pathmorphing}
\usepackage{slashed}
\usepackage{color}

\renewcommand{\d}{\mathrm{d}}

\newcommand{\e}{\varepsilon}
\newcommand{\dbar}{\kern-.1em{\raise.8ex\hbox{ -}}\kern-.6em{d}}

\def\?{\marginpar{not sure}}

\newcommand{\comment}[1]{}

\newtheorem{thm}{Theorem}[section]

\newtheorem{lemma}[thm]{Lemma}
\newtheorem{prop}[thm]{Proposition}

\newtheorem{rem}[thm]{Remark}

\def \be{\begin{equation}}
\def \ee{\end{equation}}
\def \ben{\begin{equation*}}
\def \een{\end{equation*}}
\def \bea{\begin{eqnarray}}
\def \eea{\end{eqnarray}}

\def\qed{\hfill\raise1pt\hbox{\vrule height5pt width5pt depth0pt}}
\def\nn{\nonumber}

\def\Tr{\mathrm{Tr}}

\def\L{\Lambda}
\def\l{\lambda}

\def\s{\sigma}
\def\a{\alpha}
\def\b{\beta}
\def\d{\delta}
\def\D{\Delta}
\def\m{\mu}

\def\G{\Gamma}
\def\e{\varepsilon}
\def\t{\tau}

\definecolor{light}{gray}{.75}

\let\a=\alpha \let\b=\beta    \let\d=\delta \let\e=\varepsilon
       \let\l=\lambda
\let\m=\mu    \let\n=\nu             
\let\s=\sigma \let\t=\tau    
   \let\o=\omega
\let\G=\Gamma \let\D=\Delta  \let\L=\Lambda 
         
\let\O=\Omega

\newcommand{\xx}{{\bf x}}
\newcommand{\yy}{{\bf y}}
\newcommand{\zz}{{\bf z}}
\newcommand{\kk}{{\bf k}}

\def\nn{\nonumber}

\def\\{\hfill\break}
\def\={:=}

\let\0=\noindent

\def\tende#1{\,\vtop{\ialign{##\crcr\rightarrowfill\crcr\noalign{\kern-1pt
    \nointerlineskip} \hskip3.pt${\scriptstyle #1}$\hskip3.pt\crcr}}\,}
\def\otto{\,{\kern-1.truept\leftarrow\kern-5.truept\to\kern-1.truept}\,}

\def\to{\rightarrow}

\def\qed{\hfill\raise1pt\hbox{\vrule height5pt width5pt depth0pt}}

\def\be{\begin{equation}}
\def\ee{\end{equation}}
\def\bea{\begin{eqnarray}}
\def\eea{\end{eqnarray}}
\def\nn{\nonumber}

\def\Tr{\mathrm{Tr}}

\theoremstyle{plain}

\theoremstyle{definition}

\begin{document}
\title{Universal edge transport in interacting Hall systems}

\author[1]{Giovanni Antinucci}
\affil[1]{University of Z\"urich, Institute for Mathematics, Winterthurerstrasse 190, 8057 Z\"urich, Switzerland}
\author[2]{Vieri Mastropietro}
\affil[2]{University of Milano, Department of Mathematics ``F. Enriquez'', Via C. Saldini 50, 20133 Milano, Italy}
\author[1]{Marcello Porta\footnote{New address: University of T\"ubingen, Department of Mathematics, Auf der Morgenstelle 10, 72076 T\"ubingen, Germany.}}

\maketitle

\begin{abstract}
We study the edge transport properties of $2d$ interacting Hall systems, displaying single-mode chiral edge currents. For this class of many-body lattice models, including for instance the interacting Haldane model, we prove the quantization of the edge charge conductance and the bulk-edge correspondence. Instead, the edge Drude weight and the edge susceptibility are interaction-dependent; nevertheless, they satisfy exact universal scaling relations, in agreement with the chiral Luttinger liquid theory. Moreover, charge and spin excitations differ in their velocities, giving rise to the spin-charge separation phenomenon. The analysis is based on exact renormalization group methods, and on a combination of lattice and emergent Ward identities. The invariance of the emergent chiral anomaly under the renormalization group flow plays a crucial role in the proof.
\end{abstract}

\section{Introduction}

It is a recent discovery that insulating systems, in the single-particle approximation, admit a topological classification. Hamiltonians labeled by different topological invariants cannot be continuously connected without destroying their insulating features or breaking their basic symmetry properties. The very first examples of {\it topological insulators} are the integer quantum Hall systems: there, the topological invariant is the bulk Hall conductivity, which can only take integer values \cite{Lau, TKNN}. The simplest example of Hall system is the {\it Haldane model} \cite{Hal}, a graphene-like system describing fermions on the honeycomb lattice exposed to a suitable zero-flux magnetic field.
 
While insulating in the bulk, Hall insulators support {\it edge currents}, \cite{Halp}; see \cite{Ch} for a comprehensive theoretical and experimental review. Bulk and edge transport properties are related by a remarkable duality, the {\it bulk-edge correspondence}, stating the equality of edge conductance and bulk conductivity. This duality is by now well understood from a mathematical viewpoint for singe-particle models, and has been extended to various classes of topological insulators \cite{Hat, SKR, EG, GP, GT, PSB}. 

Much less is known in the presence of many-body interactions, where the single particle description breaks down. If the strength of the interaction is much smaller than the gap of the single-particle Hamiltonian, one expects the bulk transport properties to be unaffected. This is in agreement with formal field-theoretic arguments \cite{CoH, I}, suggesting the universality of the bulk Hall conductivity. The situation becomes much less clear for the gapless edge excitations. Their behavior could be strongly affected by many-body interactions, no matter how small. In the case of single-mode edge currents, the {\it chiral Luttinger model} has been proposed as an effective field theory for the edge states of Hall systems \cite{Wen, F}. This effective continuum theory describes chiral $1+1$ dimensional massless relativistic fermions with short range interactions. Such approximate description suggests, in the weak coupling regime, the universality of edge conductance and the non-universality of other transport coefficients, like the susceptibility $\kappa$ and the Drude weight $D$. Remarkably, the nonuniversal transport coefficients of the chiral Luttinger model are connected by exact scaling relations, stating for instance that $D = \kappa v^2_{c}$, where $v_{c}$ is the dressed charge velocity of the chiral fermions.

The chiral Luttinger model can be solved using an exact mapping into free bosons, both in the spinless and spinful case. This effective theory, however, neglects important effects coming from the nonlinearity of the energy bands, from Umklapp scattering and from the bulk degrees of freedom, which {\it a priori} might affect the transport coefficients. Nevertheless, in analogy with the theory of Luttinger liquids \cite{HalLut}, whose predictions have been rigorously established in several non-solvable models \cite{BM,BM1,BFM, BMdrude, BMchains, BFMhub1, BFMhub2, GMT1, GMT2}, one is naturally led to conjecture that the scaling relations between transport coefficients found in the chiral Luttinger model hold true for a wide universality class, which includes quantum Hall edge states, at least if the interaction strength is not too large. Large interactions are expected to drastically modify the transport properties, see \cite{HK, QZ, HA} for reviews. The bosonic description of the chiral Luttinger model is the basis of a phenomenological theory for the fractional quantum Hall effect \cite{Wen, FK}, in which the edge modes are described by free chiral bosons parametrized by a rational number $\nu$ related to the Hall conductivity, see also \cite{KF, Ch, FS, F}.

Coming now to rigorous results, the geometrical approach of \cite{TKNN, AS2} has been generalized to many-body lattice Hamiltonians in \cite{HM} (see \cite{Ba} for a streamlined proof), where the exact quantization of the bulk Hall conductivity is proven under the assumption that the interacting spectrum is gapped. This assumption has been recently proven in \cite{Hgap}, for Hamiltonians obtained as perturbations of free Fermi gases. On a more field-theoretic side, the approach of \cite{CoH} can be implemented in a rigorous way \cite{GMP}, using constructive Quantum Field Theory methods to prove convergence of perturbation theory and lattice Ward identities to prove universality of the Hall conductivity in the interaction strength. The proof of \cite{GMP} works provided the interaction strength is smaller than the gap in the single-particle spectrum. More recently, the strategy of \cite{GMP} has been improved in \cite{GMPhald, GMPhald2}, to prove quantization of the Hall conductivity arbitrarily close to criticality, for models displaying conical intersections in the spectrum at the critical point, like the Haldane-Hubbard model. The same methods can be used to prove the universality of the longitudinal conductivity for interacting graphene-like models, \cite{GMPcond}, whose spectrum is gapless. Thus, the main advantage of the field-theoretic methods of \cite{GMPcond, GMPhald, GMPhald2} is that they allow to prove universality of transport coefficients {\it without assuming} the presence of a gap in the spectrum.

Concerning the edge transport properties of interacting topogical insulators, no rigorous results are available in the literature. Aim of this paper is to investigate the edge modes of a class of interacting Hall systems defined on a cylinder, supporting single-mode edge currents. That is, each edge of the cylinder supports one edge state, modulo the spin degeneracy. A famous example of such class of systems is the interacting Haldane model, in the nontrivial topological phases \cite{Hao}. We prove the exact quantization of the edge conductance, for weak interactions: all interaction corrections cancel out. Combined with \cite{GMP} and with the noninteracting bulk-edge correspondence \cite{Hat, SKR, EG}, this result provides the first proof of the bulk-edge correspondence for an interacting many-body quantum system. Moreover, we also consider the edge Drude weight and the edge susceptibility, both for charge and spin degrees of freedom; we find explicit expressions for these quantities, which turn out to be nonuniversal in the coupling strength. Nevertheless, the Drude weight $D$ and the susceptibility $\kappa$ satisfy the universal scaling relation $D = \kappa v_{c}^{2}$, as in the Luttinger model. Finally, we compute the two-point function, and we show that it exhibits spin-charge separation.

Notice that our analysis does not extend in a straightforward way to the case of multi-edge currents. The reason being the scattering between different edge modes. In the renormalization group terminology, the edge states scattering is a marginal process; our method allows to control the scattering between edge states with the same velocity, thanks to the comparison with the chiral Luttinger model (see below), but does not allow to control the scattering of edge states with different velocities. We leave the generalization to multi-edge channels Hall systems as a very interesting open problem, on which we plan to come back in the future.

The paper is organized as follows. In Section \ref{sec:setting} we introduce the class of interacting lattice models we will consider, and we define bulk and edge transport coefficients, in the linear response regime. In Section \ref{sec:nonintQH} we recall some known facts about noninteracting Hall systems. Then, in Section \ref{sec:main} we present our main result, Theorem \ref{thm:1}. In the rest of the paper, we discuss the proof of Theorem \ref{thm:1}. In Section \ref{sec:fint} we introduce a functional integral representation for fermionic lattice models, and in particular in Section \ref{sec:red1d} we derive a rigorous relationship between the model of interest and an interacting one-dimensional quantum field theory. This result actually applies to models with a general number of edge states. Starting from Section \ref{sec:ref}, we restrict the attention to single-channel Hall systems. In Section \ref{sec:ref} we introduce the chiral Luttinger model, rigorously constructed by renormalization group methods in \cite{FM, BFMhub2}, which plays the role of reference model in our proof, and in Proposition \ref{prp:relref} we state the precise connection between the correlation function of the reference model and those of the lattice model. This allows to use the reference model, whose correlations can be computed explicitly, to describe the large scale properties of the edge excitations. In Section \ref{sec:WI} we derive exact Ward identities for both the lattice and the reference model, following from $U(1)$ gauge symmetry. In Section \ref{sec:proof} we prove our main result, Theorem \ref{thm:1}, using the connection with the reference model, Proposition \ref{prp:relref}, and the Ward identities of Section \ref{sec:WI}. Finally, in Section \ref{sec:RG} we introduce an exact RG scheme, that allows to prove Proposition \ref{prp:relref} in Section \ref{sec:relref}. This is done by tuning the bare parameters of the reference model in order to match the asymptotic behavior of the lattice correlations, up to multiplicative and additive finite renormalization. All these finite renormalizations turn out to be completely determined by the lattice and emergent Ward identities. The advantage of the reference model with respect to the lattice one is the presence of extra, chiral, Ward identities. These relations are {\it anomalous}, and their anomaly satisfies a nonrenormalization property, analogous to the Adler-Bardeen theorem \cite{AB}, see \cite{M2} for a rigorous analysis in one-dimensional system. As a result, the anomaly is linear in the bare coupling, with an explicit prefactor. This, together with the computations of the finite renormalizations relating lattice and reference model, are the crucial facts behind the explicit expressions for edge transport coefficients of Theorem \ref{thm:1} and the validity of universal scaling relations. Finally, in Appendix \ref{app:1d} we collect some technical results needed for the RG analysis, and in Appendix \ref{app:wick} we prove a rigorous version of the Wick rotation for the edge transport coefficients.

\section{Setting}\label{sec:setting}

In this section we will introduce the class of interacting lattice models we will be interested in. In Section \ref{sec:1}, we will start by defining the models in first quantization. In Section \ref{sec:2ndqu}, we will switch to a grand-canonical Fock space description, and in Section \ref{sec:edgetrans} we will define bulk and edge transport coefficients.

\subsection{Many-body fermionic lattice models}\label{sec:1}

Given $L\in \mathbb{N}$, we consider a system of $N$ interacting fermions on a finite square lattice of side $L$:
\be
\L_{L} :=\{ \vec x \in \mathbb{Z}^{2} \mid \vec x = n_{1} \vec e_{1} + n_{2} \vec e_{2}\;,\quad 0\leq n_{i} \leq L\;,\quad i=1,2\}
\ee
where $\vec e_{1} = (1,\, 0)$ and $\vec e_{2} = (0,\, 1)$. For the moment, we leave the boundary conditions unspecified. At first, suppose that $N=1$. The single particle wave function is denoted by $\psi \equiv \psi(\vec x, r)$, with $\vec x\in \L_{L}$ and $r = 1,\ldots, M$ an internal degree of freedom. The spin of the particle, $\s = \uparrow\downarrow$, is an example of possible internal degree of freedom. The wave function satisfies the usual normalization condition $\|\psi\|_{2}^{2} = \sum_{\vec x, r} |\psi(\vec x, r)|^2 = 1$. The dynamics of the particle is generated by a Schr\"odinger operator $H: \frak{h}_{L}\to \frak{h}_{L}$, with $\frak{h}_{L} = \mathbb{C}^{L^2} \otimes \mathbb{C}^{M}$ the single particle Hilbert space. 

Suppose now that $N>1$. The state of the system is described by a normalized, antisymmetric wave function $\psi_{N}\equiv \psi_{N}(\vec x_{1}, r_{1}; \ldots; \vec x_{N}, r_{N})$, an element of $\frak{h}_{L}^{\wedge N}$ with $\wedge$ the antisymmetric tensor product. The dynamics of the $N$-particle system is generated by the many-body Hamiltonian $H_{N}: \frak{h}^{\wedge N}_{L}\to \frak{h}^{\wedge N}_{L}$:
\be\label{eq:HN}
H_{N} := \sum_{i=1}^{N} H^{(i)} + \lambda \sum_{i<j}^{N} w^{(ij)}
\ee
where $H^{(i)} := \mathbbm{1}^{\otimes (i-1)}\otimes H\otimes \mathbbm{1}^{\otimes(N-i)}$ with $H$ the single particle Hamiltonian and $\mathbbm{1}$ the identity on $\mathfrak{h}_{L}$; $w^{(ij)} := w(\hat x^{(i)}, \hat x^{(j)})$ is the pair interaction among the particles $i$, $j$, with $\hat x^{(i)} := \mathbbm{1}^{\otimes(i-1)} \otimes \hat x\otimes \mathbbm{1}^{\otimes(N-i)}$ where $\hat x$ is the position operator, and with $w(\vec x, \vec y)$ an $M\times M$ matrix; $\l\in \mathbb{R}$ is the coupling constant. In the following, we will assume that both the single particle Hamiltonian and the interaction potential are finite range. More precisely, we shall consider Hamiltonians allowing hoppings between nearest and next-to-nearest neighbours: $H_{rr'}(\vec x, \vec y) = w_{rr'}(\vec x, \vec y) = 0$ whenever $\| \vec x - \vec y \| >\sqrt{2}$, with $\|\cdot\|$ the Euclidean distance on the lattice. This is not a loss of generality: one can always enlarge the number of internal degrees of freedom $M$ in such a way that this holds true.

We assume periodic boundary conditions and translation invariance in the $\vec e_{1}$ direction. That is, the kernels of the single particle Schr\"odinger operator and of the interaction potential have the form:
\be
H_{rr'}(\vec x,\vec y) \equiv H_{rr'}(x_{1} - y_{1}; x_{2}, y_{2})\;,\qquad w_{rr'}(\vec x, \vec y) \equiv w_{rr'}(x_{1} - y_{1}; x_{2}, y_{2})\;,
\ee
and $H_{rr'}(\vec x, \vec y) \equiv H_{rr'}(\vec x + n L \vec e_{1}, \vec y + m L \vec e_{1})$, $w_{rr'}(\vec x, \vec y) \equiv w_{rr'}(\vec x + nL \vec e_{1}, \vec y+ mL \vec e_{1})$ for all $n, m\in \mathbb{Z}$. Let $S^{1}$ the circle of length $2\pi$, and let $S^{1}_{L} = S^{1}\cap \frac{2\pi}{L} \mathbb{Z}$ its discretization of mesh $2\pi/L$. Let $k_{1}\in S^{1}_{L}$. We define the Fourier transforms:
\be
\hat H_{rr'}(k_{1}; x_{2}, y_{2}) := \sum_{z_{1}=0}^{L} e^{i  k_{1} z_{1}} H_{rr'}(z_{1}; x_{2}, y_{2})\;,\qquad \hat w_{rr'}(k_{1}; x_{2}, y_{2}) := \sum_{z_{1}=0}^{L} e^{i k_{1} z_{1}} w_{rr'}(z_{1}; x_{2}, y_{2})\;.
\ee
For fixed $k_{1}$, the operator $\hat H(k_{1})$ defines an {\it effective one-dimensional Schr\"odinger operator}. It acts as, for all $\varphi\in \mathbb{C}^{L}\otimes \mathbb{C}^{M}$:
\be\label{eq:Hk1} (\hat H(k_{1})\varphi)_{x_{2}} = A(k_{1})\varphi_{x_{2}+1} + A(k_{1})^{*}\varphi_{x_{2} - 1} + V(k_{1})\varphi_{x_{2}}\;,\ee 
where $A(k_{1})$ and $V(k_{1})$ are $M\times M$ matrices, with $V(k_{1}) = V(k_{1})^{*}$, and $\det A(k_{1}) \neq 0$. We shall impose {\it Dirichlet boundary conditions} at $x_{2} = 0, L$ by setting $\hat H_{rr'}(k_{1}; x_{2}, y_{2}) = 0$ whenever $x_{2}, y_{2} \notin [1, L-1]$. Concerning the many-body Hamiltonian (\ref{eq:HN}), the Dirichlet boundary condition enters by setting $w_{rr'}(\vec x, \vec y) = 0$ whenever $x_{2}, y_{2} \notin [1, L-1]$. That is, this choice of boundary conditions reduces $\L_{L}$ to a {\it cylinder}.

In the following, we will also be interested in the case of periodic boundary conditions in both $\vec e_{1}$ and $\vec e_{2}$ directions. In this case, $\L_{L}$ will be equivalent to a {\it torus}. We shall denote by $H^{\text{(per)}}$ the corresponding one-particle Schr\"odinger operator, $H^{\text{(per)}}_{rr'}(\vec x, \vec y) = H^{\text{(per)}}_{rr'}(\vec x + \vec n L, \vec y + \vec m L)$ for all $\vec n,\vec m\in \mathbb{Z}^{2}$. We shall assume that $H^{\text{(per)}}$ is translation invariant, $H^{\text{(per)}}_{rr'}(\vec x, \vec y)\equiv H^{\text{(per)}}_{rr'}(\vec x - \vec y)$. As before, we can introduce the effective one-dimensional Sch\"odinger operator $\hat H^{\text{(per)}}(k_{1})$. Being the system translation invariant in the $\vec e_{2}$ direction as well, we can define the {\it Bloch Hamiltonian} as,  for $\vec k = (k_{1}, k_{2}) \in \mathbb{T}_{L}^{2} \equiv S^{1}_{L}\times S^{1}_{L}$:
\be
\hat H^{(\text{per})}(\vec k) := \sum_{\vec z\in \L_{L}} e^{i\vec z\cdot \vec k} H^{\text{(per)}}(\vec z)\;,\qquad \hat{H}^{(\text{per})}(\vec k): \mathbb{C}^{M}\to \mathbb{C}^{M}\;.
\ee
We will denote by $w^{\text{(per)}}$ the interaction potential in the presence of periodic boundary conditions, $w^{\text{(per)}}_{rr'}(\vec x,\vec y) \equiv w^{\text{(per)}}_{rr'}(\vec x + \vec n L, \vec y + \vec m L)$, $w^{\text{(per)}}_{rr'}(\vec x, \vec y) \equiv w^{\text{(per)}}_{rr'}(\vec x - \vec y)$, and by $H_{N}^{\text{(per)}} $the corresponding many-body Hamiltonian.

For finite $L$, the spectra of $\hat H(k_{1})$, $\hat H^{\text{(per)}}(k_{1})$ are discrete. The eigenfunctions of $\hat H^{\text{(per)}}(k_{1})$ are extended, and have the form $(1/\sqrt{L}) e^{-i k_{2} x_{2}} u_{\alpha}(\vec k)$, with $u_{\alpha}(\vec k)$ the {\it Bloch functions} of $\hat H^{\text{(per)}}(\vec k)$, with $\alpha = 1,2, \ldots, M$ the {\it band label}. The spectrum of $\hat H(k_{1})$ might be qualitatively different from the one of $\hat H^{(\text{per})}(k_{1})$: {\it edge states} might appear. These are solutions of the Schr\"odinger equation decaying in the bulk of the system as $e^{-c|x_{2}|}$ or as $e^{-c|x_{2} - L|}$. The edge modes are essential for the edge transport properties of topological insulators, which will be the focus of the present paper.

\medskip

\noindent{\it Example: the Haldane model.}\label{sec:Hal} The {\it Haldane model} \cite{Hal} is a paradigmatic example of topological insulator. Here, $M=2$: neglecting the spin for simplicity, the two internal degrees of freedom correspond to the two triangular sublattices forming the honeycomb lattice.
\begin{figure}[hbtp]
\centering
\includegraphics[width=.45
\textwidth]{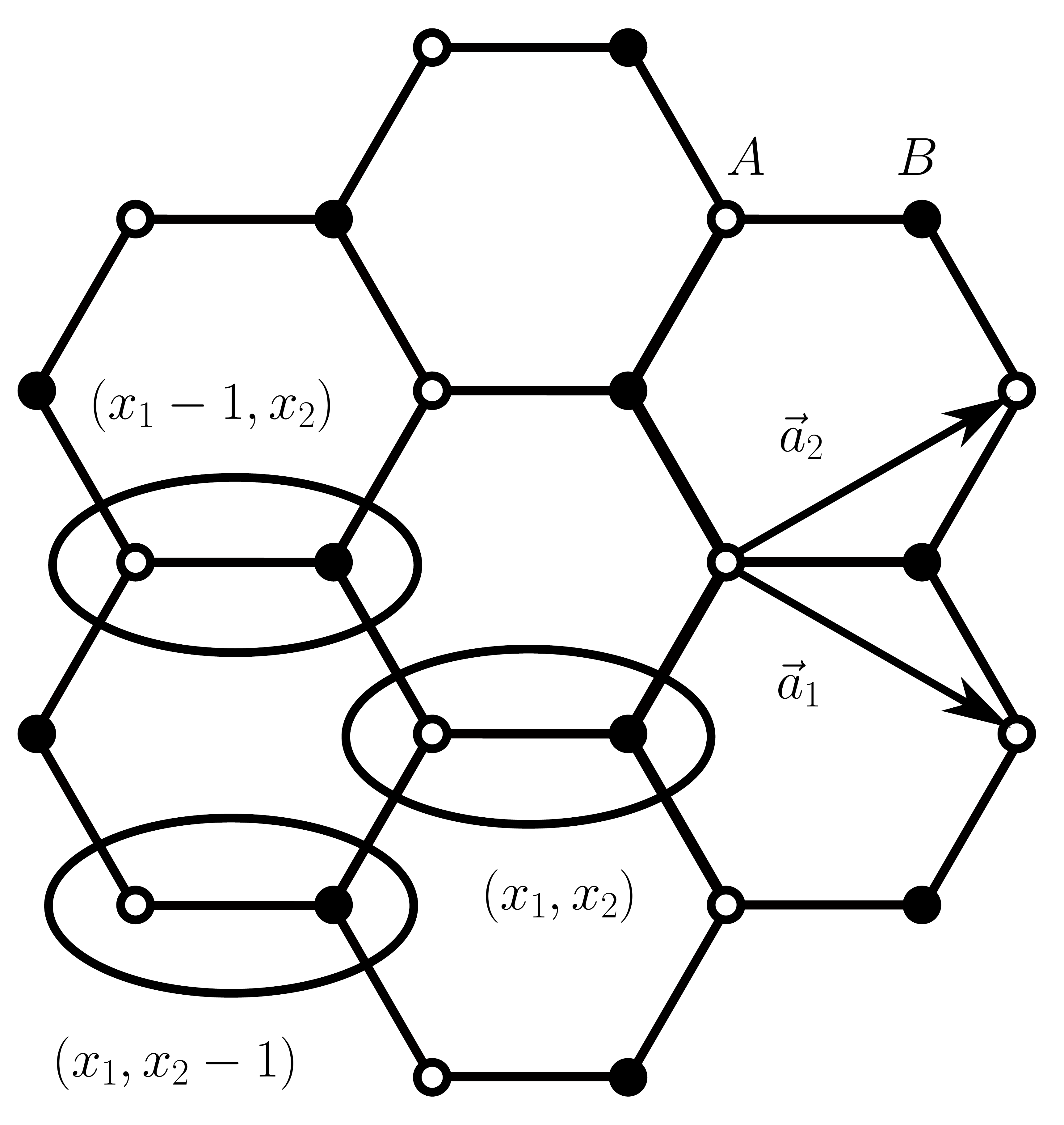}
\caption{The honeycomb lattice of the Haldane model. The empty dots belong to the $A$-sublattice, while the black dots belong to the $B$-sublattice. The ovals encircle the two sites of the fundamental cell, which are labeled by the coordinates $(x_{1}, x_{2})$, corresponding to the basis $\vec a_{1},\, \vec a_{2}$ of the $A$ triangular lattice.
}\label{fig:haldane}
\end{figure}
%
%
Let $\varphi^{A}_{\vec x},\, \varphi^{B}_{\vec x} \in \mathbb{C}$ be the values of the wave function corresponding to the two lattice sites enclosed by the oval in Fig. \ref{fig:haldane}. We define $\varphi^{T}_{\vec x} = (\varphi_{\vec x}^{A},\, \varphi_{\vec x}^{B})$. Let $(x_{1}, x_{2})$ be the coordinates on the $A$-triangular lattice in the basis $\vec a_{1}$, $\vec a_{2}$, see Figure \ref{fig:haldane}. The Haldane model describes fermions on the honeycomb lattice, hopping between nearest and next-to-nearest neighbours, in the presence of a zero-flux magnetic field and of a staggered chemical potential. Consider the model in the presence of periodic boundary conditions. The corresponding effective one-dimensional Schr\"odinger operator $\hat H^{\text{(per)}}(k_{1})$ has the form (\ref{eq:Hk1}), with:
\bea\label{eq:AV}
A^{\text{(per)}}(k_{1}) &=& \begin{pmatrix} -t_{2} e^{i\phi} e^{-ik_{1}} - t_{2} e^{-i\phi} & 0 \\ -t_{1} & -t_{2} e^{-i\phi} e^{-ik_{1}} - t_{2} e^{i\phi} \end{pmatrix} \nn\\
V^{(\text{per})}(k_{1}) &=& \begin{pmatrix} W - t_{2}e^{i\phi}e^{ik_{1}} - t_{2}e^{-i\phi}e^{-ik_{1}} & - t_{1} e^{-ik_{1}} - t_{1} \\ - t_{1} e^{ik_{1}} - t_{1} & -W - t_{2}e^{-i\phi} e^{ik_{1}} - t_{2} e^{i\phi}e^{-ik_{1}} \end{pmatrix}\;.
\eea
The energy bands and the eigenfunctions of $\hat H^{(\text{per})}(k_{1})$ can be computed explicitly. It turns out that its spectrum is {\it gapped}, for a generic choice of the parameters. The eigenvalues of $\hat H^{\text{(per)}}(k_{1})$ are labeled by $k_{2}\in S^{1}_{L}$ and by the band label $\alpha = \pm$. They are:
\be\label{eq:bandsHal}
e_{\pm}(\vec k) = -2 t_{2}\cos\phi [\cos(k_{1} - k_{2}) + \cos k_{2} + \cos k_{1}] \pm \sqrt{ m(\vec k)^{2} + t_{1}^{2} |\O(\vec k)|^{2} }\;,
\ee
where $\O(\vec k) = t_{1}(1 + e^{-ik_{1}} + e^{-ik_{2}})$ and $m(\vec k) = W - 2t_{2}\sin\phi [ \sin(k_{1} - k_{2}) + \sin k_{2} - \sin k_{1} ]$. The corresponding eigenfunctions are:
\be\label{eq:blochHal}
\phi_{x_{2}}^{\pm}(\vec k) = \frac{e^{-ik_{2} x_{2}}}{\sqrt{L}} u_{\pm}(\vec k)\;,\quad u_{\pm}(\vec k) = \frac{e^{i\gamma_{\pm}(\vec k)}}{N_{\pm}(\vec k)} \begin{pmatrix} t_{1}\O^{*}(\vec k) \\ \pm (m(\vec k)^{2} + t_{1}^{2}|\O(\vec k)|^{2})^{1/2} + m(\vec k)  \end{pmatrix}\;,
\ee
where $N_{\pm}(\vec k)$ ensures the normalization $\| u_{\pm} \|_{2} = 1$, and $\gamma_{\pm}(\vec k)$ is a phase. In the $L\to\infty$ limit, it is well known that it might be impossible to choose $\gamma_{+}(\vec k),\, \gamma_{-}(\vec k)$ as single-valued, continuous functions on the Brillouin zone $\mathbb{T}^{2} = S^{1}\times S^{1}$, \cite{TKNN, AS2}. This is the manifestation of a nontrivial topology of the {\it Bloch bundle} $\mathcal{E}^{\alpha}_{\text{B}} = \{ (\vec k, v_{\alpha}(\vec k)) \in \mathbb{T}^{2}\times \mathbb{C}^{2} \mid v_{\alpha}(\vec k) \in \text{Ran} P_{\alpha}(\vec k)\}$, with $P_{\alpha}(\vec k) = |u_{\alpha}(\vec k)\rangle \langle u_{\alpha}(\vec k)|$ the projector over the $\alpha$-th energy band. This fact is deeply related with the presence of a nonzero Hall conductivity, \cite{Hal}. Concerning the model on the cylinder, the eigenfunctions and the eigenvalues do not have simple expressions as on the torus. Nevertheless, choosing the parameters of the model in the topologically nontrivial phases, corresponding to a nonzero Hall conductivity, one can show that the Schr\"odinger equation admits two solutions, {\it edge modes}, decaying in the bulk respectively as $e^{-cx_{2}}$, $e^{-c(L - x_{2})}$ \cite{Hao}.

\subsection{Second quantization}\label{sec:2ndqu}

In order to study the system for $\l\neq 0$, it is convenient to switch to a {\it Fock space description}. Given $\frak{h}_{L} = \mathbb{C}^{L^2} \otimes \mathbb{C}^{M}$, the fermionic Fock space is defined as $\mathcal{F}_{L} := \bigoplus_{N\geq 0} \frak{h}_{L}^{\wedge N}$. Notice that, for $L<\infty$, $\text{dim}\,\mathcal{F}_{L}<\infty$ (by Pauli principle). It is convenient to introduce fermionic creation and annihilation operators, $a^{+}_{\vec x,r}$ resp. $a^{-}_{\vec x,r}$, as the operators that allow to move between different sectors of the fermionic Fock space. For any $\varphi = (\varphi^{(0)}, \varphi^{(1)}, \ldots, \varphi^{(n)}, \ldots, ) \in \mathcal{F}_{L}$, we define:
\bea
&&(a^{+}_{\vec y, r} \varphi)^{(n)}(\vec x_{1},r_{1}, \ldots, \vec x_{n}, r_{n}) \nn\\&&\qquad := \frac{1}{\sqrt{n}} \sum_{j=1}^{n} (-1)^{j} \delta_{\vec y, \vec x_{j}}\d_{r, r_{j}} \varphi^{(n-1)}(\vec x_{1},r_{1}, \ldots, \vec x_{j-1},r_{j-1}, \vec x_{j+1},r_{j+1},\ldots, \vec x_{n}, r_{n})\;,\nn\\
&&(a^{-}_{\vec y, r} \varphi )^{(n)}(\vec x_{1}, r_{1}, \ldots, \vec x_{n}, r_{n}) := (\sqrt{n+1})\varphi^{(n+1)}(\vec y, r, \vec x_{1}, r_{1}, \ldots, \vec {x}_{n}, r_{n})\;,
\eea
with $\d_{\cdot, \cdot}$ the Kronecker delta. It follows that $a^{+} = (a^{-})^{*}$. We will encode the boundary conditions in the fermionic operators by setting $a_{\vec x, r}^{\pm} = 0$ whenever $x_{2} = 0, L$, and $a^{\pm}_{\vec x,r} \equiv a^{\pm}_{x_{1} + n L, x_{2}, r}$ for all $n\in \mathbb{Z}$. Starting from these definitions, one can easily check the canonical anticommutation relations: 
\be
\{ a^{+}_{\vec x, r}\,, a^{-}_{\vec y, r'} \} = \delta_{\vec x,\vec y} \d_{r,r'}\;,\qquad \{ a^{+}_{\vec x, r}\, , a^{+}_{y, r'} \} = \{ a^{-}_{\vec x, r}\, , a^{-}_{\vec y, r'} \} = 0\;. 
\ee
We let the partial Fourier transforms of the fermionic creation/annihilation be:
\be
\hat a^{\pm}_{k_{1}, x_{2}} = \sum_{x_{1} = 0}^{L} e^{\mp i k_{1} x_{1}}\, a^{\pm}_{\vec x}\quad \forall k_{1}\in S^{1}_{L} \iff a^{\pm}_{\vec x} = \frac{1}{L} \sum_{k_{1} \in S^{1}_{L}} e^{\pm i k_{1} x_{1}} \hat a^{\pm}_{k_{1}, x_{2}}\quad \forall x_{1}\in \mathbb{Z}\;.
\ee
Let $H_{N}$ be the Hamiltonian acting on the $N$-particle sector of the Fock space, $\mathcal{F}_{L}^{(N)} \equiv \frak{h}_{L}^{\wedge N}$, defined in Eq. (\ref{eq:HN}). The second quantized Hamiltonian is defined as $\mathcal{H}_{L} := \bigoplus^{ML^{2}}_{N\geq 0} [H_{N} - \m N]$. In terms of the fermionic creation and annihilation operators,
\be
\mathcal{H}_{L} = \sum_{\vec x, \vec y}\sum_{r,r'} a^{+}_{\vec x,r} H_{rr'}(\vec x, \vec y)a^{-}_{\vec y, r'} + \lambda \sum_{\vec x, \vec y}\sum_{r, r'} \Big[\rho_{\vec x,r} - \frac{1}{2}\Big] w_{rr'}(\vec x, \vec y) \Big[\rho_{\vec y,r'} - \frac{1}{2}\Big] - \mu \mathcal{N}_{L}\;,
\ee
where $\rho_{\vec x,r} = a^{+}_{\vec x, r} a^{-}_{\vec x, r}$ and $\mathcal{N}_{L} = \sum_{\vec x, r} \rho_{\vec x, r}$. The factors $-1/2$ have been introduced in order to simplify the Grassmann representation of the model, see Section \ref{sec:fint}; these factors simply amount to a redefinition of the parameter $\m$. The Gibbs state associated with $\mathcal{H}_{L}$ is:
\be
\langle \cdot \rangle_{\beta,L} = \frac{\Tr_{\mathcal{F}_{L}}\cdot e^{-\beta \mathcal{H}_{L}}}{\mathcal{Z}_{\beta,L}}\;,\qquad \mathcal{Z}_{\beta,L} = \Tr_{\mathcal{F}_{L}}e^{-\beta\mathcal{H}_{L}}\;.
\ee
The parameter $\m$ plays the role of chemical potential, and it is chosen so to fix the average density of the system. We shall also use the notation $\langle \cdot \rangle_{\infty} := \lim_{\beta\to\infty}\lim_{L\to\infty} \langle \cdot \rangle_{\beta,L}$. 

We define the imaginary-time (or Euclidean) evolution of the fermionic operators as:
\be\label{eq:axx}
a^{\pm}_{\xx,r} := e^{x_{0} \mathcal{H}_{L}} a^{\pm}_{\vec x, r} e^{-x_{0} \mathcal{H}_{L}}\;,\qquad \xx = (x_{0}, \vec x)\quad \text{with}\quad x_{0}\in [0, \beta)\;.
\ee
In the following, we shall also denote by $a^{\pm}_{\xx, r}$ the antiperiodic extension of (\ref{eq:axx}) to all $x_{0}\in \mathbb{R}$. Let $\mathbb{M}^{\text{F}}_{\beta} = \frac{2\pi}{\beta}(\mathbb{Z} + \frac{1}{2})$ be the set of fermionic Matsubara frequencies. The time Fourier transform of $a^{\pm}_{\xx}$ is defined as, for $k_{0}\in \mathbb{M}^{\text{F}}_{\beta}$, $\hat a^{\pm}_{k_{0}, \vec x} = \int_{0}^{\beta} dx_{0}\, e^{\mp i k_{0} x_{0}}\, a^{\pm}_{\xx}$. Also, let $\mathbb{M}_{\b}^{\text{B}} = \frac{2\pi}{\beta} \mathbb{Z}$ be the set of bosonic Matsubara frequencies, and let $O$ be an even observable on the Fock space, {\it i.e.} a linear combination of monomials of even degree in the fermionic creation/annihilation operators. We define its time Fourier transform as, for $p_{0}\in \mathbb{M}_{\b}^{\text{B}}$, $\hat O_{p_{0}} = \int_{0}^{\beta} dx_{0}\, e^{i p_{0} x_{0}} O_{x_{0}}$.

Given any set of fermionic operators $a_{\xx_{i}, r_{i}}^{\e_{i}}$, $i=1,\ldots, n$, we define the fermionic time-ordered product as:
\be
{\bf T} a^{\e_{1}}_{\xx_{1}, r_{1}} a^{\e_{2}}_{\xx_{2}, r_{2}}\cdots a^{\e_{n}}_{\xx_{n}, r_{n}} = \text{sgn}(\pi) a^{\e_{\pi(1)}}_{\xx_{\pi(1)}, r_{\pi(1)}}\cdots\, a^{\e_{\pi(n)}}_{\xx_{\pi(n)}, r_{\pi(n)}}\;,
\ee
where the permutation $\pi$ is such that $x_{0,\pi(1)}\geq x_{0,\pi(2)}\geq \cdots \geq x_{0,\pi(n)}$; in case some times are equal the ambiguity is solved by normal ordering. The $n$-point {\it Schwinger function} is defined as:
\be
S^{\beta,L}_{n}(\xx_{1}, \e_{1}, r_{1}; \ldots; \xx_{n}, \e_{n}, r_{n}) := \langle {\bf T} a^{\e_{1}}_{\xx_{1}, r_{1}} a^{\e_{2}}_{\xx_{2}, r_{2}}\cdots\, a^{\e_{n}}_{\xx_{n}, r_{n}} \rangle_{\beta,L}\;.
\ee
Information on the ground state of the system can be obtained by studying the zero temperature, infinite volume limit of the Schwinger functions,
\be
S_{n}(\xx_{1}, \e_{1}, r_{1}; \ldots; \xx_{n}, \e_{n}, r_{n}) := \lim_{\beta\to\infty} \lim_{L\to\infty} S^{\beta,L}_{n}(\xx_{1}, \e_{1}, r_{1}; \ldots; \xx_{n}, \e_{n}, r_{n})\;.
\ee
%
%
%
%

In the absence of interactions, the Gibbs state is quasi-free, and all the Schwinger functions can be computed starting from the two-point function via the fermionic Wick rule.  The noninteracting two-point function $S^{\beta,L,(0)}_{2}$ can be computed explicitly as follows. Let $\underline{x} = (x_{0}, x_{1})$ and $\underline{k} = (k_{0}, k_{1})$. Suppose that $x_{0} - y_{0} \neq n\beta$. Then, setting $S^{\beta,L,(0)}_{2}(\xx,r; \yy, r') \equiv S^{\beta,L,(0)}_{2}(\xx,-,r; \yy, +,r')$:
\bea\label{eq:2pt}
S^{\beta,L,(0)}_{2}(\xx,r; \yy, r')  &=& \frac{1}{\beta L} \sum_{\underline{k} \in \mathbb{M}^{\text{F}}_{\beta}\times S^{1}_{L}} \frac{e^{-i\underline{k}\cdot (\underline{x} - \underline{y})}}{-ik_{0} + \hat H(k_{1}) - \m}(x_{2}, r; y_{2}, r')\nn\\
&\equiv& \frac{1}{\beta L} \sum_{\underline{k}\in \mathbb{M}^{\text{F}}_{\beta}\times S^{1}_{L}}e^{-i\underline{k}\cdot (\underline{x} - \underline{y})}\sum_{q=1}^{ML} \, \frac{\varphi_{x_{2}}^{q}(k_{1}; r) \overline{\varphi_{y_{2}}^{q}(k_{1}; r')}}{-ik_{0} + e_{q}(k_{1}) - \m}\;,
\eea
where $\varphi^{q}(k_{1}), e_{q}(k_{1})$ are the eigenfunctions and the eigenvalues of $\hat H(k_{1})$. If instead $x_{0} - y_{0} = n\beta$, then $S^{\beta,L,(0)}_{2}((y_{0}+n\beta, \vec x),-,r; \yy, +,r') = (-1)^{n} \lim_{x_{0} - y_{0}\to 0^-} S^{\beta,L,(0)}_{2}(\xx,-,r; \yy, +,r')$. We will also consider the analogous quantities in the presence of periodic boundary conditions. In this case, the Hamiltonian $H$ is replaced by $H^{(\text{per})}$, and the fermionic operators are compatible with the periodic boundary conditions: $a^{\pm}_{\vec x,r} = a^{\pm}_{\vec x + \vec n L, r}$ for all $\vec n \in \mathbb{Z}^2$. We will denote by $\langle \cdot \rangle^{(\text{per})}_{\beta,L}$ the corresponding Gibbs state. The two-point function can be obtained from Eq. (\ref{eq:2pt}) replacing $\varphi^{q}_{x_{2}}(k_{1}; r)$ with $(1/\sqrt{L}) e^{-ik_{2}x_{2}}u_{\alpha}(\vec k; r) $, with $u_{\alpha}(\vec k)$, $\alpha=1,\ldots, M$, the Bloch functions of $\hat H^{(\text{per})}(\vec k)$:
\be\label{eq:2ptper}
S^{\beta,L,(0)}_{2,\text{per}}(\xx, r; \yy, r') = \frac{1}{\beta L^{2}} \sum_{\kk\in \mathbb{M}^{\text{F}}_{\beta}\times \mathbb{T}_{L}^2}\sum_{\alpha=1}^{M}  \frac{e^{-i\kk\cdot (\xx - \yy)}}{-ik_{0} + e_{\alpha}(\vec k) - \m} u_{\alpha}(\vec k; r)\overline{u_{\alpha}(\vec k; r')}\;.
\ee
As an example, consider the Haldane model. The periodic two-point function can be computed explicitly starting from Eq. (\ref{eq:2ptper}) and the explicit expressions of the energy bands and the Bloch functions, given respectively in Eqs. (\ref{eq:bandsHal}), (\ref{eq:blochHal}). One gets:
\bea
&&S^{\beta,L,(0)}_{2,\text{per}}(\xx, r; \yy,r') = \frac{1}{\beta L^{2}} \sum_{\kk\in \mathbb{M}^{\text{F}}_{\beta}\times \mathbb{T}_{L}^2} e^{-i\kk\cdot (\xx - \yy)} \nn\\
&&\quad\cdot \begin{pmatrix} -ik_{0} + m(k) - \alpha_{1}( k) - \mu  & -t_{1} \Omega^*( k) \\ - t_{1}\Omega( k) & -ik_{0} - m(k) - \alpha_{1}( k) - \mu  \end{pmatrix}^{-1}_{rr'}
\eea
where $\alpha_{1}(\vec k) = 2t_{2} \cos\phi [\cos(k_{1} - k_{2}) + \cos k_{2} + \cos k_{1}]$.

\subsection{Linear response theory}\label{sec:edgetrans}

The Schwinger functions can be used to investigate the transport properties of the system in the framework of linear response theory. To this end, we define the {\it current operator} as follows. The time-derivative of the density operator $\rho_{\xx} = \sum_{r} \rho_{\xx, r}$ satisfies the continuity equation:
\be\label{eq:Jcons}
i\partial_{x_{0}} \rho_{\xx}+ \text{div}_{\vec x}\, \vec j_{\xx} = 0\;,
\ee
where $\text{div}_{\vec x}$ is the lattice divergence, $\text{div}_{\vec x} f_{\xx} = \text{d}_{1} f_{1,\xx} + \text{d}_{2} f_{2,\xx} = f_{1,\xx} - f_{1, \xx - {\bf e}_{1}} + f_{2,\xx} - f_{2,\xx - {\bf e}_{2}}$, with ${\bf e}_{1} = (0,\vec e_{1})$, ${\bf e}_{2} = (0, \vec e_{2})$, and the current density operator $\vec j_{\xx}$ is given by:
\bea\label{eq:curr}
&&\vec j_{\xx} = e^{x_{0}\mathcal{H}} \vec j_{\vec x} e^{-x_{0}\mathcal{H}}\;,\qquad \vec j_{\vec x} := \vec e_{1} j_{1,\vec x} + \vec e_{2} j_{2, \vec x}\\
&& j_{1,\vec x} := j_{\vec x, \vec x+\vec e_{1}} + \frac{1}{2}( j_{\vec x, \vec x + \vec e_{1} - \vec e_{2}} + j_{\vec x, \vec x + \vec e_{1} + \vec e_{2}} ) + \frac{1}{2}( j_{\vec x - \vec e_{2}, \vec x + \vec e_{1}} + j_{\vec x + \vec e_{2}, \vec x + \vec e_{1}} ) \nn\\
&&j_{2,\vec x} := j_{\vec x, \vec x+\vec e_{2}} + \frac{1}{2}( j_{\vec x, \vec x - \vec e_{1} + \vec e_{2}} + j_{\vec x, \vec x + \vec e_{1} + \vec e_{2}} ) + \frac{1}{2}( j_{\vec x - \vec e_{1}, \vec x +\vec e_{2}} + j_{\vec x + \vec e_{1}, \vec x + \vec e_{2}} )\nn
\eea
with the {\it bond current}:
\be
j_{\vec x, \vec y} := \sum_{r,r'} (i a^{+}_{\vec x,r} H_{rr'}(\vec x, \vec y) a^{-}_{\vec y,  r'} - i a^{+}_{\vec y,r} H_{rr'}(\vec y, \vec x) a^{-}_{\vec x,  r'})\;. 
\ee
%
%
%

For convenience, we collect the density and the current density operators in single three-dimensional Euclidean current $j_{\m, \xx}$, with $\mu = 0,1,2$, by setting $j_{0,\xx} := \rho_{\xx}$. With these notations, the continuity equation (\ref{eq:Jcons}) reads: $\text{d}_{\m} j_{\m,\xx} = 0$, with $d_{0} \equiv i\partial_{0}$.

Let us first discuss the {\it bulk} transport properties. The response of the system to bulk perturbations is expected to be insensitive to boundary effects. {\it Kubo formula} provides an expression for the bulk conductivity matrix; it is given by, choosing for convenience periodic boundary conditions:
\be\label{eq:sij}
\s_{ij} := \lim_{\eta \to 0^{+}} \frac{i}{\eta} \Big[\int_{-\infty}^{0} dt\, e^{t\eta }\, \pmb{\langle} [ j_{i}(t)\,, j_{j} ] \pmb{\rangle}^{(\text{per})}_{\infty} - \pmb{\langle} [ X_{i}\,, j_{j} ]  \pmb{\rangle}^{(\text{per})}_{\infty} \Big]\;,\qquad i, j = 1,2\;,
\ee
where $\vec j(t) = e^{i\mathcal{H} t} \vec j e^{-i\mathcal{H}t}$ is the real-time evolution of the total current $\vec j = \sum_{\vec x\in \L_{L}} \vec j_{\vec x}$ , $\vec X$ is the second quantization of the position operator, $\pmb{\langle} \cdot \pmb{\rangle}_{\infty}^{(\text{per})} := \lim_{\beta \to \infty} \lim_{L\to\infty} L^{-2} \langle \cdot \rangle_{\beta,L}^{(\text{per})}$, and, with a slight abuse of notations, the factor $i$ in front of the integral is $\sqrt{-1}$. The quantity $\s_{ii}$ is called the longitudinal conductivity, while $\s_{12} = -\s_{21}$ is called the transverse, of Hall, conductivity. For general insulating systems $\s_{ii} = 0$, while $\s_{12}$ might be different from zero. If so, we shall say that the system is a {\it Hall insulator}. Eq. (\ref{eq:sij}) describes the linear response of the system at the time $t=0$ after introducing an external perturbation $-e^{\eta t} \vec E\cdot \vec X$, for $t\leq 0$ (see \cite{Giu} for a formal derivation).

Let us now discuss the {\it edge} transport properties. These are very relevant for the physics of topological insulators. Here one is interested in the response of physical observables, such as the charge or the current, supported in the proximity of one edge, after exposing the system to a perturbation localized around the same edge. Let us suppose that the system is equipped with cylindric boundary conditions, as described in Section \ref{sec:1}. We will focus on the transport of charge or of spin in the vicinity of the $x_{2} = 0$ edge. Given an operator $O_{\vec x}$ on $\mathcal{F}_{L}$, compatible with the cylindric boundary conditions, let us introduce the notation $\hat O^{\leq a}_{p_{1}} := \sum_{x_{1} = 0}^{L}\sum_{x_{2} = 0}^{a} e^{ip_{1} x_{1}} O_{1,\vec x}$, for $p_{1} \in \frac{2\pi}{L}\mathbb{Z}$. Then, we define, for $a> a'$ and $\underline{a} = (a, a')$:
\be\label{eq:G}
G^{\underline{a}}(\eta, p_{1}) :=  -i \int_{-\infty}^{0} dt\, e^{t \eta}\, \pmb{\langle} [ \hat \rho^{\leq a}_{p_{1}}(t)\,, \hat j^{\leq a'}_{1,-p_{1}} ] \pmb{\rangle}_{\infty}
\ee
where now $\pmb{\langle} \cdot \pmb{\rangle}_{\infty} := \lim_{\beta, L\to \infty} L^{-1} \langle \cdot \rangle_{\beta,L}$ (notice the change of normalization with respect to Eq. (\ref{eq:sij})). The {\it edge charge conductance} is defined as:
\be\label{eq:Glim}
G := \lim_{a'\to\infty}\lim_{a\to\infty} \lim_{p_{1}\to 0}\lim_{\eta\to 0^{+}} G^{\underline{a}}(\eta, p_{1}) \;.
\ee
The above quantity measures the response at $t=0$ of the current along the $\vec e_{1}$ direction, supported in a strip of width $a'$, after introducing at $t=-\infty$ a perturbation of the form $e^{\eta t} \d \m \hat \rho_{p_{1}}^{\leq a}$, for $L\gg a\gg a' \gg 1$. The motivation for this range of parameters is as follows: we are interested in the effect of macroscopic perturbations, on observables that capture the presence of edge states on just one boundary of the system, corresponding to $x_{2} = 0$. The decay of the edge modes takes place on a microscopic scale, which explains the choice $a\gg a'$. As we shall see, the dependence of $G^{\underline{a}}$ on $a, a'$ is exponentially small for $a, a'$ large, which allows to take the $a',a\to \infty$ limit. 

The order of the $p_{1}, \eta \to 0$ limit in Eq. (\ref{eq:Glim}) corresponds to a static situation, in which the edge current is driven by the space modulation of a local  chemical potential. Also, notice that the limit $\beta\to\infty$ is performed before $\eta \to 0^+$: for this reason, $G^{\underline{a}}$ defines a zero temperature transport coefficient. Positive temperature transport could be studied by taking the limit $\eta \to 0^+$ for $\beta <\infty$; this choice of limits poses additional difficulties, \cite{SPA}, which will not be addressed in this paper. 

One can also study the variation of the density as a result of the application of an electric field in a strip of width $a$ from the edge $x_{2} = 0$:
\be\label{eq:tildeG}
\widetilde{G}^{\underline{a}}(\eta, p_{1}) := i \int_{-\infty}^{0} dt\,e^{t \eta}\, \pmb{\langle} [ \hat j^{\leq a}_{1, p_{1}}(t)\,, \hat \rho^{\leq a'}_{-p_{1}} ] \pmb{\rangle}_{\infty}\;,\quad \widetilde{G} := \lim_{a'\to\infty} \lim_{a\to \infty} \lim_{\eta\to 0^{+}} \lim_{p_{1}\to 0} \widetilde{G}^{\underline{a}}(\eta, p_{1})\;.
\ee
In this case, the order of the limits $\eta, p_{1}\to 0$ is reversed with respect to Eq. (\ref{eq:Glim}). This order of limits corresponds to the situation in which the variation of the density is driven by a dynamical perturbation. Similarly, we define:
\bea\label{eq:kD}
\kappa^{\underline{a}}(\eta, p_{1}) &:=& -i \int_{-\infty}^{0}dt\, e^{t \eta}\, \pmb{\langle} [ \hat \rho^{\leq a}_{p_{1}}(t)\,,  \hat \rho^{\leq a'}_{-p_{1}} ] \pmb{\rangle}_{\infty} \;,\\
D^{\underline{a}}(\eta, p_{1}) &:=& i \Big[ \int_{-\infty}^{0} dt\,e^{t \eta}\, \pmb{\langle} [ \hat j^{\leq a}_{1,p_{1}}(t)\,, \hat j^{\leq a'}_{1, -p_{1}} ] \pmb{\rangle}_{\infty} - \pmb{\langle} [X_{1}^{\leq a}, j^{\leq a'}_{1}] \pmb{\rangle}_{\infty}\Big]\;.\nn
\eea
Notice that, due to the presence of the periodic boundary condition in the $\vec e_{1}$ direction, the position operator $X_{1}$ is not well defined; however, the commutator in Eq. (\ref{eq:kD}) is a well defined object. More precisely, the commutator in Eq. (\ref{eq:kD}) has to be understood as the periodization on $\L_{L}$ of the analogous quantity on $\mathbb{Z}\times [0, L]$. This gives rise to a well defined operator, provided $a> a'$. To see this, we compute (using Eq. (\ref{eq:schwinger})):
\bea\label{eq:XJ}
&&[X_{1}^{\leq a}, j^{\leq a'}_{1}] = \sum_{y_{1}}\sum_{y_{2}\leq a'} \sum_{x_{1}} \sum_{x_{2}\leq a} x_{1} [\rho_{\vec x}, j_{1,\vec y}]  \nn\\
&&= \sum_{y_{1}, y_{2}\leq a'} \Big[ -i \mathbbm{1}(y_{2}\leq a) \tau_{\vec y, \vec y + \vec e_{1}} + \frac{i}{2}\big( y_{1}\mathbbm{1}(y_{2}\leq a) - (y_{1} + 1) \mathbbm{1}(y_{2}\leq a+1) \big)\tau_{\vec y, \vec y + \vec e_{1} - \vec e_{2}}\nn\\
&& \qquad\qquad +\frac{i}{2} \big( y_{1}\mathbbm{1}(y_{2}\leq a) - (y_{1} + 1) \mathbbm{1}(y_{2}\leq a - 1) \big) \tau_{\vec y, \vec y + \vec e_{1} + \vec e_{2}}\nn\\
&&\qquad\qquad +\frac{i}{2}\big( y_{1}\mathbbm{1}(y_{2}\leq a+1) - (y_{1} + 1) \mathbbm{1}(y_{2}\leq a) \big)\tau_{\vec y - \vec e_{2}, \vec y + \vec e_{1}} \nn\\
&&\qquad\qquad + \frac{i}{2} \big( y_{1}\mathbbm{1}(y_{2}\leq a-1) - (y_{1} + 1)\mathbbm{1}(y_{2}\leq a) \big)\tau_{\vec y + \vec e_{2}, \vec y + \vec e_{1}}\Big]\;,
\eea
with $\tau_{\vec x, \vec y}$ the kinetic energy associated to the hopping between $\vec x$ and $\vec y$:
\be\label{eq:tauxy}
\tau_{\vec x, \vec y} = \sum_{r,r'} ( a^{+}_{\vec x, r} H_{rr'}(\vec x, \vec y) a^{-}_{\vec y,  r'} + a^{+}_{\vec y, r} H_{rr'}(\vec y, \vec x) a^{-}_{\vec x,  r'} )\;.
\ee
In Eq. (\ref{eq:XJ}), the $y_{1}$ factors cancel out for $a>a'$. One gets the well defined expression $\pmb{\langle} [X_{1}^{\leq a}, \hat j^{\leq a'}_{1}] \pmb{\rangle}_{\infty} = -i \sum_{y_{2}\leq a'} \Delta_{1, y_{2}}$, with:
\be\label{eq:Delta}
\Delta_{1,y_{2}} = \lim_{\b, L\to \infty} \Delta^{\b, L}_{1,y_{2}}\;,\quad \D^{\b, L}_{1, y_{2}} = \langle \tau_{\vec y, \vec y + \vec e_{1}} + \frac{1}{2}\sum_{\vec z: \| \vec y - \vec z \| = \sqrt{2}} \tau_{\vec y, \vec z}\rangle_{\b, L}\;.
\ee
Notice that $\D^{\b, L}_{1, y_{2}}$ does not depend on $y_{1}$ by translation invariance in the $\vec e_{1}$ direction. The {\it edge charge susceptivity} and the {\it Drude weight} are defined as:
\be\label{eq:kappa}
\kappa := \lim_{a'\to\infty} \lim_{a\to \infty}\lim_{p_{1}\to 0} \lim_{\eta \to 0^{+}} \kappa^{\underline{a}}(\eta, p_{1})\;,\qquad D := \lim_{a'\to\infty} \lim_{a\to \infty}\lim_{\eta\to 0^{+}} \lim_{p_{1}\to 0} D^{\underline{a}}(\eta, p_{1})\;.
\ee
In the same way, one can also define {\it edge spin transport coefficients}. Suppose that $H$ commutes with the spin operator, namely that $H = H^{\uparrow} \oplus H^{\downarrow}$, where $H^{\s}$ acts on the spin sector labeled by $\s = \uparrow,\downarrow$. In the following, it will also be convenient to label the internal degrees of freedom as $r = (\bar r, \s)$, with $\s = \uparrow\downarrow$ the spin degree of freedom, and $\bar r = 1,\ldots, M/2$ the internal degree of freedom of $H^{\s}$ (notice that, by spin symmetry, $M$ is even). More precisely, we shall identify the internal degree of freedom $r = 1,\ldots, M/2$ with $(1,\uparrow), \ldots, (M/2, \uparrow)$, and the edge states $r = M/2+1, \ldots, M$ with $(1,\downarrow), \ldots, (M/2, \downarrow)$. Let $\rho_{\xx, \s} = \sum_{\bar r} \rho_{\xx, (\bar r, \s)}$. We define the spin density operator as $\rho^{s}_{\xx} := \rho_{\xx, \uparrow} - \rho_{\xx, \downarrow}$. The continuity equation for the spin density is:
\be
i\partial_{x_{0}} \rho^{s}_{\xx} + \text{div}_{\vec x} \vec j^{s}_{\xx} = 0\;,
\ee
where $\vec j^{s}_{\vec x} := \vec j_{\vec x,\uparrow} - \vec j_{\vec x,\downarrow}$ and $\vec j_{\vec x, \s}$ is the current density associated to the Hamiltonian $H^{\s}$. 
%
%
We define the edge spin transport coefficients $G^{\underline{a}}_{s}$, $\widetilde{G}^{\underline{a}}_{s}$, $D^{\underline{a}}_{s}$, $\kappa^{\underline{a}}_{s}$, and their limits $G_{s}$, $\widetilde{G}_{s}$, $D_{s}$, $\kappa_{s}$, by simply replacing in Eqs. (\ref{eq:G})--(\ref{eq:kappa}) the charge density $\rho_{\vec x}$ and the density of charge current $\vec j_{\vec x}$ by the spin density $\rho^{s}_{\vec x}$ and the density of spin current $\vec j^{s}_{\vec x}$. For later convenience, we shall collect all edge transport coefficients in a matrix: 
\bea\label{eq:defGG}
&&G_{01} \equiv G\;,\qquad G_{10}\equiv \widetilde{G}\;,\qquad G_{00} \equiv \kappa\qquad G_{11}\equiv D\;,\\
&&G^{s}_{01} \equiv G_{s}\;,\qquad G^{s}_{10}\equiv \widetilde{G}_{s}\;,\qquad G^{s}_{00} \equiv \kappa_{s}\qquad G^{s}_{11}\equiv D_{s}\;.\nn
\eea
Also, we shall set $G^{c}_{\m\n}\equiv G_{\m\n}$, and we will collect both charge and spin transport coefficients in $G^{\sharp}_{\m\n}$, $\sharp  = c,s$.
\begin{rem}
In general, one could also be interested in mixed transport coefficients, describing for instance the response of the spin current after introducing a voltage drop at the edges of the system. This is precisely the relevant setting for {\em spin Hall systems}, such as the Kane-Mele model \cite{KM}. We refer the reader to \cite{MP2}, where the interacting, spin-conserving Kane-Mele model has recently been studied, via an extension of the methods introduced in this paper. Here we shall focus on a class of models for which all these transport coefficients are zero by spin symmetry.
\end{rem}





\section{Noninteracting Hall systems}\label{sec:nonintQH}

Here we recall some known facts about noninteracting Hall insulators. Let us consider the Hamiltonian $H^{\text{(per)}}$ for a system on the torus. Suppose that $H^{\text{(per)}}$ is gapped, and let us place the Fermi level $\m$ in the gap: $\mu \notin \s(\hat H^{(\text{per})}(k_{1}))$ for all $k_{1}\in S^{1}_{L}$. Thanks to the gap condition, in the absence of interactions all correlation functions decay exponentially:
\be\label{eq:decS0}
\big | S^{\beta,L,(0)}_{2,\text{per}}(\xx,r;\yy,r') \big| \leq Ce^{-c\| \xx - \yy \|_{\beta, L}}
\ee
for some $C,c>0$ independent of $\beta, L$, and where $\| \cdot \|_{\beta, L}$ is the distance on the torus of sides $\beta, L$: $\| {\bf a} \|_{\beta,L} := \min_{{\bf n}\in \mathbb{Z}^{3}} \sqrt{|a_{0} - n_{0}\beta|^{2} + |a_{1} - n_{1} L|^{2} + |a_{2} - n_{2}L|^{2}}$. The constant $c$ is proportional to the distance of $\m$ to the spectrum of $H^{(\text{per})}$.

Consider the noninteracting conductivity matrix $\s^{(0)}_{ij} \equiv \s_{ij}|_{\l=0}$, as given by Eq. (\ref{eq:sij}). Let $P(\vec k) = \sum_{\alpha:\, e_{\a}(\vec k) < \m} | u_{\alpha}(\vec k)\rangle \langle u_{\alpha}(\vec k)|$ be the Fermi projector associated with $\hat H^{(\text{per})}(\vec k)$. Then, \cite{AS2}:
\be\label{eq:pdpdp}
\s^{(0)}_{ij} = i \int_{\mathbb{T}^{2}} \frac{d^{2} \vec k}{(2\pi)^{2}}\, \Tr\, P(\vec k) [ \partial_{i} P(\vec k)\,, \partial_{j} P(\vec k) ]\;.
\ee
This shows immediately that $\s^{(0)}_{ii} = 0$. Moreover, it turns out that $\s^{(0)}_{12} \in  \frac{1}{2\pi}\mathbb{Z}$, \cite{TKNN, AS2}. The quantization of the transverse, or Hall, conductivity $\s^{(0)}_{12}$ has a deep topological interpretation: $\s^{(0)}_{12}$ is the {\it Chern number} of the Bloch bundle.

Now, consider the Hamiltonian $H$ for the system on the cylinder. The spectrum of $\hat H(k_{1})$ might differ from the one of $\hat H^{(\text{per})}(k_{1})$ by the appearance of isolated eigenvalues in the gap of $\hat H^{(\text{per})}(k_{1})$. As a function of $k_{1}$, they describe discrete eigenvalue branches $k_{1}\mapsto \e(k_{1})$. For a fixed $k_{1}$, the eigenvalue branch corresponds to a solution of the Schr\"odinger equation $\hat H(k_{1}) \xi(k_{1}) = \e(k_{1}) \xi(k_{1})$, with $\| \xi (k_{1})\|_{2}^{2} = \sum_{x_{2}, r} |\xi_{x_{2}}(k_{1}; r)|^{2} = 1$, satisfying the Dirichlet boundary condition, $\xi_{0}(k_{1}; r) = \xi_{L}(k_{1}; r) = 0$. We will be interested in the eigenvalue branches that cross the Fermi level $\m$. The corresponding solutions of the Schr\"odinger equation are called the {\it edge states} of $H$. The value of $k_{1}$ for which the eigenvalue branch crosses the Fermi level is called the {\it Fermi point} $k_{F}$ of the edge state. More precisely, $k_{F}$ is defined as the element $k_{1}\in S^{1}_{L}$ that minimizes $|\e(k_{1}) - \m|$ (which tends to zero as $L\to\infty$). 

In order to deal with the presence of multiple edge states, we introduce an edge state label $e = 1,\ldots, n_{\text{edge}}$, that counts the number of eigenvalue branches intersecting the Fermi level. The value $k_{F}^{e}$ is the Fermi point of the edge state labeled by $e$. We shall consider the class of Hamiltonians specified by the following assumption.

\medskip

\noindent{\bf Assumption 1.} {\it There exists $\tilde \d, \d, L_{0} > 0$, with $\tilde \d \geq \d$ and $\tilde \d$, $\d$ independent of $L$, such that, for $L\geq L_{0}$, the following is true. The Hamiltonian $H$ has the form $H = H^{\uparrow}\oplus H^{\downarrow}$, with $H^{\s}$ the Hamiltonian acting on the $\s \in \{\uparrow,\, \downarrow\}$ spin sector, and $H^{\uparrow} = H^{\downarrow}$. The number of edge states $n_{\text{edge}}$ is $L$-independent. For every $e = 1,\ldots, n_{\text{edge}}$, the set $\{ k_{1} \in S^{1}_{L}\mid |\e_{e}(k_{1}) - \m|\leq \tilde\d \}$ supports two and only two edge states of $H$, spin degenerate. Moreover, let $\partial_{k_{1}}$ be the discrete derivative, $\partial_{k_{1}} f(k_{1}) = (L/2\pi)(f(k_{1} + \frac{2\pi}{L}) - f(k_{1}))$. The edge states of $H$ satisfy, for all $n\in \mathbb{N}$, $e = 1,\ldots, n_{\text{edge}}$ and $k_{1}\in S^{1}_{L}$ such that $|\e_{e}(k_{1}) - \m| < \d$:
\bea\label{eq:dec1}
&&\qquad\qquad\qquad\qquad\qquad\qquad\qquad |\partial^{n}_{k_{1}} \e_{e}(k_{1})|\leq C_{n}\;,\\
&&\text{either}\qquad |\partial^{n}_{k_{1}}\xi^{e}_{x_{2}}(k_{1};r)| \leq C_{n}e^{-cx_{2}}\qquad \text{or}\qquad | \partial^{n}_{k_{1}}\xi^{e}_{x_{2}}(k_{1}; r)|\leq C_{n}e^{-c(L - x_{2})}\;,\nn
\eea
with $C_{n}, C$ independent of $L$. Moreover, $v_{e} := \partial_{k_{1}} \e_{e}(k_{F}^{e}) \neq 0$ for all $e = 1,\ldots, n_{\text{edge}}$.}
\begin{rem}
\begin{itemize}
\item The bounds (\ref{eq:dec1}) are generically true for the class of Hamiltonians introduced in Section \ref{sec:1}. The proof is straightforward, and will be omitted.
\item In the following, it will be convenient to label the edge states by $e = (\bar e, \s)$, with $\s = \uparrow\downarrow$ the spin degree of freedom, and $\bar e = 1,\ldots, n_{\text{edge}}/2$ the edge state label for $H^{\s}$ (notice that, by spin symmetry, $n_{\text{edge}}$ is even). More precisely, we shall identify the edge states $e = 1,\ldots, n_{\text{edge}}/2$ with $(1,\uparrow), \ldots, (n_{\text{edge}}/2, \uparrow)$, and the edge states $e = n_{\text{edge}}/2+1, \ldots, n_{\text{edge}}$ with $(1,\downarrow), \ldots, (n_{\text{edge}}/2, \downarrow)$. 
\end{itemize}
\end{rem}
In this setting, the two-point function, Eq. (\ref{eq:2pt}), can be written as:
\be\label{eq:2ptmassless}
S^{\beta,L,(0)}_{2}(\xx,r;\yy,r') = \sum_{e=1}^{n_{\text{edge}}}  e^{-i k_{F}^{e}(x_{1} - y_{1})}\frac{\xi^{e}_{x_{2}}(k_{F}^{e}; r) \overline{\xi^{e}_{y_{2}}(k_{F}^{e}; r')} }{-i(x_{0} - y_{0}) + v_{e} (x_{1} - y_{1})} + R(\xx,r;\yy,r')\;,
\ee
where $v_{e} = \partial_{1} \e_{e}(k_{F}^{e})$ is the velocity of the edge state, and $|R(\xx,r;\yy,r')| \leq C\| \underline{x} - \underline{y} \|_{\b, L}^{-1 -\theta}e^{-c|x_{2} - y_{2}|}$ for some $\theta >0$. Due to the presence of the edge states, the Hamiltonian is {\it gapless} at the Fermi level, and the correlations decay algebraically. From the transport viewpoint, the system supports metallic currents on its boundaries. The edge transport properties of the system can be described by the transport coefficients introduced in Section \ref{sec:edgetrans}. For $\l=0$, they can be explicitly computed.
\begin{prop}{\bf (Noninteracting edge transport coefficient.)}\label{prop:edgenonint}
Let $H$ be a Hamiltonian satisfying Assumption 1. Then, the following is true.
\begin{enumerate}
\item \underline{{\it Noninteracting edge transport coefficients.}} Let $\o_{e} = \text{sgn}(v_{e})$, $\underline{p} = (\eta, p_{1})$, $a > a'$. The noninteracting edge charge and spin transport coefficients are given by, for $\sharp = c, s$:
\bea\label{eq:G0kk'}
\kappa^{\underline{a},(0)}_{\sharp}(\underline{p}) &=& \sum^{*}_{e} \frac{1}{\pi |v_{e}|} \frac{v_{e} p_{1}}{-i\eta + v_{e} p_{1}}+R^{\underline{a},(0)}_{\kappa,\sharp}(\underline{p})\;,\nn\\
G^{\underline{a},(0)}_{\sharp}(\underline{p}) &=& -\sum^{*}_{e}\frac{\o_{e}}{\pi}  \frac{v_{e} p_{1}}{-i\eta + v_{e} p_{1}}+R^{\underline{a},(0)}_{G,\sharp}(\underline{p})\;,\nn\\
\widetilde{G}^{\underline{a},(0)}_{\sharp}(\underline{p}) &=& -\sum^{*}_{e}\frac{\o_{e}}{\pi} \frac{-i \eta}{-i\eta + v_{e} p_{1}}+R^{\underline{a},(0)}_{\widetilde{G},\sharp}(\underline{p})\;,\nn\\
D^{\underline{a},(0)}_{\sharp}(\underline{p}) &=& \sum^{*}_{e}\frac{|v_{e}|}{\pi} \frac{-i \eta}{-i\eta + v_e p_{1}}+R^{\underline{a},(0)}_{D,\sharp}(\underline{p})\;,
\eea
where the asterisk restricts the sums to the edge states satisfying the first bound in the second line of Eq. (\ref{eq:dec1}), and the error terms are bounded as $|R^{\underline{a},(0)}(\underline{p})| \leq C_{n}a'(|\underline{p}|^{\theta} + |{a'} - a|^{-n}) + C_{n}a'^{-n}$ for some $\theta >0$ and all $n\in \mathbb{N}$. In particular:
\be\label{eq:Gresnonint}
\kappa^{(0)}_{\sharp} = \sum_{e}^{*} \frac{1}{2 \pi |v_{e}|}\;,\qquad D^{(0)}_{\sharp} = \sum^{*}_{e}\frac{|v_{e}|}{2\pi}\;,\qquad G^{(0)}_{\sharp} = \widetilde{G}^{(0)}_{\sharp} = -\sum_{e}^{*}\frac{\o_{e}}{2\pi}\;.
\ee
\item \underline{{\it Noninteracting bulk-edge correspondence.}} The following relation holds:
\be\label{eq:benonint}
G^{(0)} =  \s^{(0)}_{21}\;.
\ee
\end{enumerate}
\end{prop}
The proof of item $i)$ of Proposition \ref{prop:edgenonint} will be a corollary of the proof of item $i)$ of our main result, Theorem \ref{thm:1}. The proof of item $ii)$ is well known \cite{SKR, EG}. Finally, notice that, for the class of models specified by Assumption 1, all mixed transport coefficients (for instance, involving the spin current - charge density correlation function) are zero; this follows from the fact that they are odd under the symmetry transformation $\s\to -\s$.

\section{Edge transport in interacting Hall systems}\label{sec:main}

In this section we consider interacting systems, $\l\neq 0$, and we present our main result, Theorem \ref{thm:1}. We restrict the attention to the class of models specified by the following assumption, stronger than Assumption 1. 

\medskip
\noindent{\bf Assumption 2.} {\it Let $H$ be a Hamiltonian satisfying Assumption 1. We have $n_{\text{edge}} = 4$: the edge states $e = (1,\s)$ are localized on the $x_{2} = 0$ edge, while the edge states $e = (2,\s)$ are localized on the $x_{2} = L$ edge,
\be\label{eq:dec22}
|\partial^{n}_{k_{1}}\xi^{(1,\s)}_{x_{2}}(k_{1};r)| \leq C_{n}e^{-cx_{2}}\;,\qquad |\partial^{n}_{k_{1}}\xi^{(2,\s)}_{x_{2}}(k_{1}; r)|\leq C_{n}e^{-c(L - x_{2})}\;,\qquad \forall n\in \mathbb{N}\;.
\ee
Morever, the many-body interaction is spin independent: $v_{rr'}(\vec x,\vec y) \equiv v_{\bar r\bar r'}(\vec x,\vec y)$, if $r = (\bar r, \s)$, $ r' = (\bar r', \s')$, real valued and symmetric: $v_{rr'}(\vec x, \vec y) = v_{ r'r}(\vec y,\vec x)$.}
\medskip

A well known example of model satisfying these assumptions is the spinful Haldane model in the topologically nontrivial phase \cite{Hao}, see Section \ref{sec:Hal}. We shall say that these models support {\it single-channel} edge currents. We are now ready to state our main result.
\begin{thm}\label{thm:1}{\bf (Main result.)} Let $H$ be a Hamiltonian satisfying Assumption 2. Then, there exists $\bar \l>0$ such that the Schwinger functions are analytic in $\l$, for $|\l|<\bar \l$. Moreover, the following is true.
\begin{enumerate}
\item \underline{Interacting edge transport coefficients.} Let $\o = \text{sgn}(v_{(1,\s)})$, $\underline{p} = (\eta, p_{1})$, $a > a'$. The edge charge and spin transport coefficients are given by, for $\sharp = c,s$:
\bea\label{eq:edgeint}
&&\kappa^{\underline{a}}_{\sharp}(\underline{p}) = \frac{1}{\pi v_{\sharp}} \frac{\o v_{\sharp} p_{1}}{-i\eta + \o v_{\sharp} p_{1}}+R^{\underline{a}}_{\kappa,\sharp}(\underline{p})\;,\qquad G^{\underline{a}}_{\sharp}(\underline{p}) = -\frac{\o}{\pi}  \frac{\o v_{\sharp} p_{1}}{-i\eta + \o v_{\sharp} p_{1}}+R^{\underline{a}}_{G,\sharp}(\underline{p}) \;,\nn\\
&&\widetilde{G}^{\underline{a}}_{\sharp}(\underline{p}) = -\frac{\o}{\pi} \frac{-i \eta}{-i\eta + \o v_{\sharp} p_{1}}+R^{\underline{a}}_{\widetilde{G},\sharp}(\underline{p})\;,\qquad 
D^{\underline{a}}_{\sharp}(\underline{p}) = \frac{v_{\sharp}}{\pi} \frac{-i \eta}{-i\eta + \o v_\sharp p_{1}}+R^{\underline{a}}_{D,\sharp}(\underline{p})
\eea
where: $v_{c} \equiv v_{c}(\l) = |v_{(1,\s)}| + O(\l)$ is the renormalized charge velocity and $v_{s}\equiv v_{s}(\l) = v_{c}(\l) - (A/\pi) \l + O(\l^2)$ is the renormalized spin velocity, with 
\be\label{eq:A}
A = \sum_{\substack{x_{2}, y_{2} \\ r,r'}} \hat w_{rr'}(0; x_{2}, y_{2}) \overline{\xi^{(1,\s)}_{x_{2}}(k_{F}; r)} \xi^{(1,\s)}_{x_{2}}(k_{F}; r) \overline{\xi^{(1,\s)}_{y_{2}}(k_{F};  r')} \xi^{(1,\s)}_{y_{2}}(k_{F};  r')\;,
\ee
and $k_{F} \equiv k_{F}(\l) = k_{F}^{(1,\s)} + O(\l)$; the error terms are bounded as $|R^{a}(\underline{p})| \leq C_{n}a(|\underline{p}|^{\theta} + ({a'} - a)^{-n}) + C_{n}a'^{-n}$, for some $\theta >0$ and for all $n\in\mathbb{N}$. In particular:
\be\label{eq:Gres}
\kappa_{\sharp} = \frac{1}{\pi v_{\sharp}}\;,\qquad D_{\sharp} = \frac{v_{\sharp}}{\pi}\;,\qquad G_{\sharp} = \widetilde{G}_{\sharp} = -\frac{\o}{\pi}\;.
\ee

\item \underline{Spin-charge separation.} The $2$-point Schwinger function is given by, for $\underline{x}\neq \underline{y}$:
\bea\label{eq:SC}
S_{2}(\xx,r; \yy, r') &=& \frac{1}{Z} \frac{e^{-i k_{F} (x_{1} - y_{1}) } \xi^{(1,\s)}_{x_{2}} (k_{F}; r) \overline{\xi^{(1,\s)}_{y_{2}}(k_{F}; r')}}{\sqrt{(v_{s} (x_{0} - y_{0}) + i\o(x_{1} - y_{1}))(v_{c} (x_{0} - y_{0}) + i\o (x_{1} - y_{1}))}}\nn\\&& + R(\xx,r;\yy, r')\;,
\eea
where $|R(\xx,r;\yy,r') | \leq C_{n} \| \underline{x} - \underline{y} \|^{-1-\theta} |x_{2} - y_{2}|^{-n}$ for all $n\in \mathbb{N}$ and some $\theta > 0$, $C>0$, and $Z \equiv Z(\l) = 1 + O(\l)$.
\item \underline{Bulk-edge correspondence.} The following relation holds true:
\be\label{eq:beint}
G = \s_{21}\;.
\ee
\end{enumerate}
\end{thm}

The proof of item $i), ii)$ will be discussed below, starting from Section \ref{sec:fint}. In particular, the proof of item $ii)$ follows from the analogous result of \cite{FM}, for the chiral Luttinger model. The relativistic chiral Luttinger model of \cite{FM} will play the role of {\it reference model} for our analysis, and it will be discussed in Section \ref{sec:ref}. The comparison with this effective relativistic model will allow to control the renormalization group flow of the edge states scattering, which is marginal in the renormalization group sense. The proof of item $iii)$ immediately follows from the combination of item $i)$ together with the bulk-edge correspondence for noninteracting systems, Eq. (\ref{eq:benonint}), and the universality of the Hall conductivity \cite{GMP}.

A similar result can be proven for spinless fermions, after multiplying all (charge) transport coefficients by $1/2$. In this case, the proof turns out to be much simpler. In fact, one of the main technical challenges we have to face is to control the renormalization group flow of the edge states scattering, which is marginal in the RG sense; the infrared behavior of the theory is governed by a {\it nontrivial} RG fixed point. In the spinless case, instead, the quartic interaction turns out to be {\it irrelevant} in the RG terminology, thanks to Pauli exclusion principle: local quartic interactions are zero.

Our result proves the quantization of the edge conductance and provides the first proof of the bulk-edge correspondence for an interacting system. This is particularly remarkable, in view of the fact that the interaction changes the value of the other transport coefficients, and produces a drastic modification of the 2-point Schwinger function. In particular, the interaction gives rise to {\it two} different velocities for charge and spin excitations, $v_{s} = v_{c} - A\l + O(\l^2)$, thus resolving the degeneracy between charge and spin transport coefficients found in the noninteracting case. Despite all these renormalizations, the following exact relations, analogous to the Luttinger liquid relations of \cite{HalLut}, hold true:
\be
D = \kappa v_{c}^{2}\;,\qquad D_{s} = \kappa_{s} v_{s}^{2}\;.
\ee
%
%
Finally, note also that  Eqs. (\ref{eq:edgeint}) are the correlation functions of the space-time derivatives of a massless Gaussian free field, thus justifying the chiral Luttinger liquid theory of quantum Hall edge states of \cite{Wen}.

\section{Functional integral representation}\label{sec:fint}

In this section we will introduce a Grassmannian functional integral representation for the lattice model, which will allow to rewrite all Euclidean correlation functions as the derivatives of a suitable generating functional. This generating functional will then be studied via cluster expansion techniques and renormalization group. The connection with the real-time transport coefficients will be provided by Proposition \ref{prp:wick}, via a rigorous version of the Wick rotation. Finally, in Section \ref{sec:red1d} we will discuss how to integrate the bulk degrees of freedom. The outcome of the integration will be the generating functional of a suitable one-dimensional Grassmann field. 

We will set up the strategy by considering the general class of models specified by Assumption 1. Later, we will restrict the attention to the single-channel case, Assumption 2.

\subsection{Euclidean quantum field theory}

Let $\chi(s)$ be a smooth, even, compactly supported function, such that $\chi(s) = 0$ for $|s|>2$ and $\chi(s) = 1$ for $|s|<1$. Let $N\in \mathbb{N}$, and $\mathbb{M}^{\text{F}*}_{\beta} = \{ k_{0}\in \mathbb{M}^{\text{F}}_{\beta} \mid \chi(2^{-N}k_{0}) > 0 \}$. Let $\mathbb{D}^{*}_{\beta,L} = \mathbb{M}^{\text{F}*}_{\beta}\times S^{1}_{L}$. We consider the finite Grassmann algebra generated by the the Grassmann variables $\{\hat\Psi^{\pm}_{\underline{k}, q}\}$ with $\underline{k}\in \mathbb{D}^{*}_{\beta,L}$, $q=1,\ldots, ML$. The Grassmann Gaussian integration $\int P_{ N}(d\Psi)$ is a linear functional acting on the Grassmann algebra as follows. Its action on a given monomial $\prod_{j=1}^{n} \hat\Psi^{\e_{j}}_{\underline{k}_{j}, q_{j}}$ is zero unless $| \{j: \e_{j} =+ \} | = |\{ j: \e_{j} = - \}|$, in which case:
\be \int P_{N}(d\Psi) \hat \Psi^-_{\underline{k}_{1},q_{1}}\hat \Psi^+_{\underline{p}_1,q'_{1}}\cdots
\hat \Psi^-_{\underline{k}_n,q_{n}}\hat \Psi^+_{\underline{p}_n,q'_{n}}=\det[C(\underline{k}_j,q_{j};\underline{p}_k,q'_{k})]_{j,k=1,\ldots,n},
\ee
where $C(\underline{k},q;\underline{p},q')=\b L \d_{\underline{k},\underline{p}}\delta_{q,q'} \hat g_{\b,L,N}(\underline{k},q)$ and
\bea\label{eq:prop}
&&\hat g_{\beta,L,N}(\underline{k},q) := \frac{\chi_{N}(k_{0})}{-ik_{0} + e_{q}(k_{1}) - \m}\;,\qquad \chi_{N}(k_{0}) \equiv \chi_{0}(2^{-N}k_{0})\;,
\eea
with $\{e_{q}(k_{1})\}_{q=1}^{ML}$ the eigenvalues of $\hat H(k_{1})$. Eq. (\ref{eq:prop}) is the definition of the {\it Green's function}, or {\it free propagator} of the noninteracting lattice model. We define the configuration space Grassmann fields as:
\be
\Psi^{+}_{\xx,r} := \frac{1}{\beta L}\sum_{\underline{k}\in \mathbb{D}^{*}_{\beta,L}} \sum_{q=1}^{ML} e^{i\underline{k}\cdot \underline{x}}\, \overline{\varphi^{q}_{x_{2}}(k_{1}; r)}\hat\Psi^{+}_{\underline{k}, q}\;,\qquad \Psi^{-}_{\xx,r} := \frac{1}{\beta L}\sum_{\underline{k}\in \mathbb{D}^{*}_{\beta,L}} \sum_{q=1}^{ML} e^{-i\underline{k}\cdot \underline{x}}\, \varphi^{q}_{x_{2}}(k_{1}; r)\hat\Psi^{-}_{\underline{k}, q}\;,
\ee
where $\varphi^{q}(k_{1})$, $q = 1, \ldots , ML$, are the normalized eigenfunctions of $\hat H(k_{1})$. We have:
\be \int P_{N}(d\Psi)\Psi_{\xx,r}^-\Psi^+_{\yy,r'}= g_{\b,L,N}(\xx, r; \yy, r'),\label{eq:g14.b}\ee
where
\be
g_{\beta,L,N}(\xx, r; \yy, r') = \frac{1}{\beta L}\sum_{\underline{k} \in \mathbb{D}^{*}_{\beta,L}}\sum_{q=1}^{ML} e^{-i\underline{k}\cdot (\underline{x} - \underline{y})} \varphi^{q}_{x_{2}}(k_{1}; r) \overline{\varphi^{q}_{y_{2}}(k_{1}; r')} \hat g_{\beta,L,N}(\underline{k}, q).
\ee
As $N\to \infty$ and for $x_{0} \neq y_{0}$, the propagator converges pointwise to the two-point Schwinger function of the noninteracting lattice model, Eq. (\ref{eq:2pt}). 

If needed, $\int P_{N}(d\Psi)$ can be written explicitly in terms of the usual Berezin integral $\int d\Psi$,
which is the linear functional on the Grassmann algebra acting non trivially on a monomial 
only if the monomial is of maximal degree, in which case 
$$\int d\Psi \prod_{\underline{k}\in\mathbb{D}_{\b,L}^*}\prod_{q=1}^{ML}
\hat \Psi^-_{\underline{k},q}\hat \Psi^+_{\underline{k},q}=1.$$ The explicit expression of 
$\int P_{N}(d\Psi)$ in terms of $\int d\Psi$ is
\bea\label{eq:P} &&\int P_{N}(d\Psi)\big(\cdot\big) = \frac1{\mathcal{N}_{\b,L,N}}\int d\Psi \exp\Big\{-\frac1{\b L}
\sum_{\underline{k}\in\mathbb{D}_{\b,L}^*}\sum_{q=1}^{ML}
\hat\Psi^{+}_{\underline{k},q}\,\big[{\hat g}_{\beta,L,N}(\underline{k}, q)\big]^{-1}\hat\Psi^{-}_{\underline{k},q}\Big\}\big(\cdot\big),
\nonumber\\
&& \quad {\rm with}\qquad \mathcal{N}_{\b,L,N}=\prod_{\underline{k}\in\mathbb{D}_{\b,L}^*}\prod_{q=1}^{ML} [\b L]
{\hat g}_{\beta,L,N}(\underline{k},q)\;.
\label{2.3}\eea
The Grassmann counterpart of the many-body interaction is:
\be V_{\b,L}(\Psi) := \l\int_{0}^\b dx_{0}\sum_{\vec x, \vec y \in\L_L} \sum_{r, r' =1}^{M} n_{\xx,r} w_{rr'}(\vec x, \vec y) n_{(x_{0}, \vec y),r'}\;,
\ee
where $n_{\xx,r} = \Psi^{+}_{\xx,r} \Psi^{-}_{\xx,r}$ is the Grassmann counterpart of the density operator. We define $n_{\xx,\s} := \sum_{\bar r} n_{\xx, (\bar r, \s)}$ and $n^{c}_{\xx} := \sum_{\s} n_{\xx,\s}$, $n^{s}_{\xx} := \sum_{\s} \s n_{\xx,\s}$ as the Grassmann counterparts of the charge and spin densities, respectively (performing the identification $\s = \uparrow \equiv +$, $\s = \downarrow \equiv -$). We also introduce the Grassmann counterpart of the charge and spin currents as $\vec J^{c}_{\xx} := \sum_{\s} \vec J_{\xx,\s}$, $\vec J^{s}_{\xx} := \sum_{\s} \s \vec J_{\xx,\s}$, with, recalling that ${\bf e}_{i} = (0, \vec e_{i})$, $i=1,2$:
\bea
&& J_{1,\xx,\s} := J_{\xx,\xx+{\bf e}_{1},\s} + \frac{1}{2}( J_{\xx, \xx+{\bf e}_{1}-{\bf e}_{2},\s} + J_{\xx,\xx+{\bf e}_{1}+{\bf e}_{2},\s} ) + \frac{1}{2}( J_{\xx-{\bf e}_{2}, \xx+{\bf e}_{1},\s} + J_{\xx+{\bf e}_{2}, \xx+{\bf e}_{1},\s} ) \nn\\
&&J_{2,\xx,\s} := J_{\xx,\xx+{\bf e}_{2},\s} + \frac{1}{2}( J_{\xx,\xx-{\bf e}_{1}+{\bf e}_{2},\s} + J_{\xx, \xx+{\bf e}_{1}+{\bf e}_{2},\s} ) + \frac{1}{2}( J_{\xx-{\bf e}_{1}, \xx+{\bf e}_{2},\s} + J_{\xx+{\bf e}_{1}, \xx+{\bf e}_{2},\s} )\;,\nn
\eea
with the Grassmann bond current:
\be
J_{\xx,\yy} := \sum_{r,r'} i \Psi^{+}_{\xx,r} H_{rr'}(\vec x, \vec y) \Psi^{-}_{\yy,  r'} - i\Psi^{+}_{\yy,r} H_{rr'}(\vec y, \vec x) \Psi^{-}_{\xx,  r'}\;.
\ee
For convenience, we set $J^{c}_{0,\xx} \equiv n^{c}_{\xx}$, $J^{s}_{0,\xx} \equiv n^{s}_{\xx}$, and we collect Grassmann densities and Grassmann current densities in $J^{\sharp}_{\m,\xx}$, $\m = 0,1,2$. We then define the {\it source terms} as:
\be
B(\Psi; \phi) := \int_{0}^{\beta} dx_{0}\sum_{\vec x\in \L_{L}} \sum_{r=1}^{M} \phi^{+}_{\xx,r} \Psi^{-}_{\xx,r} + \Psi^{+}_{\xx,r} \phi^{-}_{\xx,r}\;,\quad \G(\Psi; A) := \int_{0}^{\beta} dx_{0}\sum_{\vec x\in \L_{L}}  \sum_{\m = 0,1,2}\sum_{\sharp=c,s} A^{\sharp}_{\m,\xx} J^{\sharp}_{\m,\xx}\;,
\ee
where $\phi^{\pm}_{\xx},\, A^{\sharp}_{\m,\xx}$  are, respectively, Grassmann and a complex valued external fields. The {\it generating functional} of the correlation functions $\mathcal{W}_{\beta,L}(A,\phi)$ is:
\be\label{eq:genfcn}
\mathcal{W}_{\beta,L}(A,\phi) := \lim_{N\to\infty} \mathcal{W}_{\beta,L,N}(A,\phi) := \lim_{N\to \infty}\log \int P_{ N} (d\Psi) e^{-V_{\beta,L}(\Psi) + \G(\Psi; A) + B(\Psi; \phi)}\;,
\ee
provided the limit exists. Thanks to Eq. (\ref{eq:P}), we can also rewrite the functional integral in the right-hand side as:
\be
\int P_{ N} (d\Psi) e^{-V_{\beta,L}(\Psi) + \G(\Psi; A) + B(\Psi; \phi)} = \frac{1}{\mathcal{N}_{\beta,L,N}}\int d\Psi\, e^{-S_{\beta,L}(\Psi) + \G(\Psi; A) + B(\Psi; \phi)}
\ee
where 
\be\label{eq:action}
S_{\beta,L}(\Psi) := \frac{1}{\beta L} \sum_{\underline{k} \in \mathbb{D}^{*}_{\beta,L}} \sum_{x_{2}, y_{2}} \Psi^{+}_{\underline{k}, x_{2}, r} \big[ (-ik_{0} - \m)\delta_{r,r'}\delta_{x_{2}, y_{2}} + \hat H_{rr'}(k_{1}; x_{2}, y_{2}) \big]\Psi^{-}_{\underline{k}, y_{2},r'} + V_{\beta,L}(\Psi)\;.
\ee
It is a well-known fact in quantum statistical mechanics that the Schwinger functions of the Gibbs state defined in Section \ref{sec:2ndqu} can be obtained as functional derivatives of the generating functional. Setting $j^{c}_{\m} \equiv j_{\m}$, we have:
\be
\langle {\bf T} a^{\e_{1}}_{\xx_{1}, r_{1}}\,; \cdots \,; a^{\e_{n}}_{\xx_{n}, r_{n}}\,; j^{\sharp_{1}}_{\m_{1}, \yy_{1}}\,; \cdots \,; j^{\sharp_{m}}_{\m_{m}, \yy_{m}} \rangle_{\beta,L}  = \frac{\partial^{n+m}\mathcal{W}_{\beta,L}(A,\phi)}{\partial \phi^{\e_{1}}_{\xx_{1}, r_{1}} \cdots\, \partial \phi^{\e_{n}}_{\xx_{n}, r_{n}} \partial A^{\sharp_{1}}_{\m_{1}, \yy_{1}} \cdots\, \partial A^{\sharp_{m}}_{\m_{m}, \yy_{m}}} \Big|_{\substack{A=0 \\ \phi=0}}\;.
\ee
We refer the reader to, e.g., Section 5.1 of \cite{GMP} for a proof of this statement. The usefulness of this result is that the Grassmann representation of the model can be investigated using cluster expansion techniques and rigorous renormalization group. Moreover, one can recover the real-time transport coefficients starting from the Euclidean correlation functions thanks to the following proposition.

\begin{prop}{\bf (Wick rotation.)}\label{prp:wick} Under the same assumption of Theorem \ref{thm:1}, the following is true. Let $\eta \neq 0$, $p_{1}\neq 0$. Let $\eta_{\beta}\in \frac{2\pi}{\beta} \mathbb{Z}$ such that it minimizes $|\eta_{\beta} - \eta |$. The charge transport coefficients can be rewritten as, for $\m,\n = 0,1$ (recall Eq. (\ref{eq:defGG})):
\be\label{eq:Dkk'0}
G^{\underline{a}}_{\m\n}(\eta, p_{1}) = \sum_{y_{2}\leq a'}\Big[ \sum_{x_{2}\leq a} \lim_{\beta, L\to \infty} \pmb{\langle} {\bf T} \hat j_{\m,\underline{p}, x_{2}}\,; \hat j_{\n,-\underline{p}, y_{2}} \pmb{\rangle}_{\beta, L} + \D_{\m,y_{2}} \d_{\m\n}  \Big](-1)^{\d_{\m,1}}\;,
\ee
where $\D_{0,y_{2}} = 0$, $\D_{1,y_{2}}$ is given by Eq. (\ref{eq:Delta}), $\pmb{\langle} \cdot \pmb{\rangle}_{\beta, L} = (\beta L)^{-1} \langle \cdot \rangle_{\beta, L}$ and $\underline{p} = (\eta_{\beta}, p_{1})$.
%
%
Similarly, the spin transport coefficients are obtained from Eq. (\ref{eq:Dkk'0}) replacing $j_{\m,\xx}$ with $j_{\m,\xx}^{s}$, for $\m=0,1$.
\end{prop}
Proposition \ref{prp:wick} is proven in Appendix \ref{app:wick}. The proof is a slight modification of the Wick rotation of \cite{GMP}, inspired by Theorem 5.4.12 of \cite{BR}. Proposition \ref{prp:wick} will be the starting for the evaluation of the edge transport coefficients.

\subsection{Reduction to a one-dimensional theory}\label{sec:red1d}
In this section we start discussing the evaluation of the generating functional of the correlation functions. In particular, we shall describe the integration of the bulk degrees of freedom, that allows to rewrite the generating functional in terms of a suitable $1+1$-dimensional Grassmann field, associated with the edge modes. This is the content of the next proposition.
\begin{prop}\label{prp:red}{\bf (Reduction to a one-dimensional theory.)}\label{prp:1d}  
Let $\m_{e}$, $e= 1, \ldots, n_{\text{edge}}$, be a sequence of real numbers such that $|\m_{e} - \m|\leq C|\l|$. Let, for $\underline{x} = (x_{0}, x_{1})$, $\underline{k} = (k_{0}, k_{1})$:
\be\label{eq:prop1d}
g^{\text{(1d)}}_{e, N}(\underline{x} - \underline{y}) := \frac{1}{\beta L} \sum_{\underline{k} \in \mathbb{D}_{\beta,L}^{*}}e^{-i\underline{k}\cdot(\underline{x} - \underline{y})} \frac{\chi_{N}(k_{0})\chi_{e}(k_{1})}{-ik_{0} + \e_{e}(k_{1}) - \m}
\ee
where $\chi_{e}(k_{1}) \equiv \chi((2/\d) |\e_{e}(k_{1}) - \m_{e}|)$ with $\d$ as in Assumption 1. Then, there exists $\bar \l>0$ such that for $|\l|< \bar \l$:
\be\label{eq:1d}
\mathcal{W}_{\beta,L}(A,\phi) = \mathcal{W}^{\text{(bulk)}}_{\beta,L}(A,\phi) + \lim_{N\to\infty} \log \int P_{N}(d\psi) e^{V^{\text{(1d)}}(\psi; A, \phi)}\;,
\ee
where $\psi\equiv \psi^{\pm}_{\underline{x}, e}$ is a $1+1$-dimensional Grassmann field and $P_{N}(d\psi)$ is a Gaussian Grassmann integration, with propagator:
\be\label{eq:cov1d}
\int P_{N}(d \psi) \psi^{-}_{\underline{x}, e} \psi^{+}_{\underline{y}, e'} = \delta_{e,e'}\, g^{(1d)}_{e,N}(\underline{x} - \underline{y})\;.
\ee
The effective one dimensional action is:
\be\label{eq:V1d}
V^{\text{(1d)}}(\psi; A, \phi) = \sum_{\G} \int_{\beta, L} D \underline{X}  D {\bf Y}  D {\bf Z}\,\psi_{\G}(\underline{X}) \phi_{\G}({\bf Y}) A_{\G}({\bf Z})  W_{\G}^{\text{(1d)}}(\underline{X},{\bf Y},{\bf Z})
\ee
where: $\G = \{ \G_{\psi}, \G_{\phi}, \G_{A}\}$, with $\G_{\psi} = \{(e_{1}, \e_{1}), \ldots, (e_{n}, \e_{n})\}$, $\G_{\phi} = \{ (r_{1}, \kappa_{1}), \ldots, (r_{s}, \kappa_{s}) \}$, $\G_{A} =  \{ (\m_{1},\sharp_{1}), \ldots, (\m_{m},\sharp_{m}) \}$, collects all field indices; $\sum_{\G}$ sums over all field configurations, for any $n\geq 1$, $m, s\geq 0$; $\int_{\beta, L} d\xx \equiv \int_{0}^{\beta} dx_{0} \sum_{\vec x \in \L_{L}}$; $D\underline{X} = \prod_{i=1}^{n} d\underline{x}_{i}$, $D{\bf Y} = \prod_{j=1}^{s} d \yy_{j}$, $D{\bf Z} = \prod_{k=1}^{m} d\zz_{k}$; $\psi_{\G}(\underline{X}) = \prod_{i=1}^{n} \psi^{\e_{i}}_{\underline{x}_{i}, e_{i}}$, $\phi_{\G}({\bf Y})  = \prod_{j=1}^{s} \phi^{\kappa_{j}}_{\yy_{j}, r_{j}}$, with $\e_{i}, \kappa_{j} = \pm$, and $A_{\G}({\bf Z}) = \prod_{k=1}^{m} A^{\sharp_{k}}_{\zz_{k}, \m_{k}}$; the kernels $W_{\G}^{\text{(1d)}}(X,{\bf Y},{\bf Z}) \equiv W_{\G}^{\text{(1d)}}(\{ \underline{x}_{i} \},\{ \yy_{j} \}, \{\zz_{k}\})$ are analytic in $|\l| < \bar \l$, are translation invariant in the ${\bf e}_0 = (1,0,0)$ and ${\bf e}_{1} = (0, \vec e_{1})$ directions and satisfy the bounds, for some constant $C>0$ depending only $\{n_{ij},\, m_{ij}\}$:
\be\label{eq:decW1d}
\int_{\beta, L} D \underline{Q} D\widetilde{Q}_{2}\, \prod_{i<j} \| \underline{q}_{i} - \underline{q}_{j} \|_{\b, L}^{n_{ij}}\prod_{l<k} |\tilde q_{l,2} - \tilde q_{k,2}|^{m_{lk}} |W^{(1d)}_{\G}(\underline{Q}, \widetilde{Q}_{2})| \leq \beta L C\;,\quad \forall n_{ij},\, m_{ij} \in \mathbb{N}
\ee
where $\| \cdot \|_{\beta, L}$ is the distance on the $2$-torus of sides $\beta , L$, and $\underline{Q} = \{ \underline{x}_{1},\ldots, \underline{x}_{n}, \underline{y}_{1}, \ldots, \underline{y}_{s}, \underline{z}_{1}, \ldots, \underline{z}_{m}\}$, $\widetilde{Q}_{2} = \{ y_{1,2},\ldots, y_{s,2}, z_{1,2},\ldots, z_{m,2} \}$. The formula (\ref{eq:V1d}) has to be understood with the conventions $\psi_{\emptyset} = A_{\emptyset} = \phi_{\emptyset} = 1$, $W_{\emptyset} = 0$. Finally, the generating functional $\mathcal{W}^{\text{(bulk)}}_{\beta,L}(A,\phi)$ is given by an expression of the form (\ref{eq:V1d}), with $n=0$; its kernels satisfy a bound like (\ref{eq:decW1d}) with $\beta L$ replaced by $\beta L^{2}$. 
\end{prop}
\begin{rem}\label{rem:red}
The choice of replacing $\m$ with $\m_{e}$ at the argument of the cutoff functions is motivated by the renormalization of the chemical potential, that will be introduced later.
\end{rem}
\begin{proof}
We rewrite the propagator as:
\be
g_{\beta, L, N}(\xx,r; \yy, r') = g^{\text{(edge)}}_{\beta, L,N}(\xx, r; \yy, r') + g^{\text{(bulk)}}_{\beta, L,N}(\xx,r; \yy, r')
\ee
where, for $\delta$ as in Assumption 1,
\be\label{eq:gedge}
g^{\text{(edge)}}_{\beta, L,N}(\xx, r; \yy, r') := \frac{1}{\beta L}\sum_{\underline{k} \in \mathbb{D}^{*}_{\beta,L}} \sum_{e=1}^{n_{\text{edge}}} e^{-i\underline{k}\cdot (\underline{x} - \underline{y})} \xi^{e}_{x_{2}}(k_{1}; r) \overline{\xi^{e}_{y_{2}}(k_{1}; r')} \frac{\chi_{N}(k_{0})\chi_{e}(k_{1})}{-ik_{0} + \e_{e}(k_{1}) - \m}\;,
\ee
with $\xi^{e}(k_{1})$ the edge state with dispersion relation $\e_{e}(k_{1})$ and with Fermi point $k_{F}^{e}$ (recall the discussion after Eq. (\ref{eq:pdpdp})). Since $|\m_{e} - \m|\leq C|\l|$, for $|\l|$ small enough the support of $\chi_{e}(k_{1}) \equiv \chi((2/\d)|\e_{e}(k_{1}) - \m_{e}|)$ is contained in the support of $\chi((1/\d)|\e_{e}(k_{1}) - \m|)$. Notice that, by Assumption 1, for fixed $k_{1}$ the sum over $e$ involves one edge state, whose spin is determined by the internal degrees of freedom $r,  r'$. Instead, the bulk propagator is given by:
\bea\label{eq:gbulk}
g^{\text{(bulk)}}_{\beta, L,N}(\xx, r; \yy, r') &=& \frac{1}{\beta L}\sum_{\underline{k} \in \mathbb{D}^{*}_{\beta,L}} \sum_{e=1}^{n_{\text{edge}}} e^{-i\underline{k}\cdot (\underline{x} - \underline{y})}  \Big(\frac{\chi_{N}(k_{0})\chi_{e}(k_{1})}{-ik_{0} + \hat H(k_{1}) - \m}P_{\perp}^{e}(k_{1})\Big)(x_{2},r; y_{2},  r')\nn\\
&& + \frac{1}{\beta L}\sum_{\underline{k} \in \mathbb{D}^{*}_{\beta,L}} e^{-i\underline{k}\cdot (\underline{x} - \underline{y})}  \Big(\frac{\chi_{N}(k_{0}) \chi_{\geq}(k_{1})}{-ik_{0} + \hat H(k_{1}) - \m}\Big)(x_{2},r; y_{2},  r')\;,
\eea
with $P^{e}_{\perp}(k_{1}) = 1 - P^{e}(k_{1})$, $P^{e}(k_{1}) = |\xi^{e}(k_{1})\rangle \langle \xi^{e}(k_{1})|$, and $\chi_{\geq}(k_{1}) = 1 - \sum_{\bar e} \chi_{(\bar e, \s)}(k_{1})$. Thus, by construction, the bulk propagator is massive, and satisfies the bound, see Appendix \ref{app:bulk}:
\be\label{eq:gbound}
| g^{\text{(bulk)}}_{\beta, L,N}(\xx, r; \yy, r') | \leq \frac{C_{n} \d^{-2}}{1 + (\d \| \underline{x} - \underline{y} \|_{\beta, L})^{n}} e^{- c \d |x_{2} - y_{2}| }\;,\qquad \forall n\in \mathbb{N}\;,
\ee
for some constants $c,\,C_{n}$ independent of $\beta, L, N, \d$. We now use the addition principle of Grassmann Gaussian variables to write the Grassmann field as sum of two independent fields, $\Psi^{\pm}_{\xx} = \Psi^{(\text{edge})\pm}_{\xx} + \Psi^{(\text{bulk})\pm}_{\xx}$, having propagators $g^{\text{(edge)}}$, $g^{\text{(bulk)}}$, respectively. The integration of the ``massive'' field $\Psi^{(\text{bulk})}$ can be performed in a standard way, following for example Section 5.2 of \cite{GMP} with some minor changes discussed in  Appendix \ref{app:bulk}. The result is:
\be\label{eq:Weff}
\mathcal{W}_{\beta,L,N}(A,\phi) = \mathcal{W}^{(\text{bulk})}_{\beta,L,N}(A,\phi) + \log \int P_{N}(d\Psi^{\text{(edge)}}) e^{V^{\text{(edge)}}(\Psi^{\text{(edge)}}; A, \phi)}
\ee
with:
\be\label{eq:Vedge}
V^{\text{(edge)}}(\Psi; A, \phi) =\sum_{\G'} \int_{\beta, L} D {\bf X}  D {\bf Y}  D {\bf Z}\,\Psi_{\G'}({\bf X}) \phi_{\G'}({\bf Y}) A_{\G'}({\bf Z}) W^{(\text{edge})}_{\G'}({\bf X},{\bf Y},{\bf Z})
\ee
where $\G' = \{\G'_{\psi}, \G'_{\phi}, \G'_{A}\}$, with $\G'_{\psi} = \{(r_{1}, \e_{1}), \ldots, (r_{n}, \e_{n})\}$, $\G'_{\phi} = \{ ( r'_{1}, \kappa_{1}), \ldots, ( r'_{s}, \kappa_{s}) \}$, $\G'_{A} = \{ (\m_{1},\sharp_{1}), \ldots, (\m_{m},\sharp_{m}) \}$, and $n\geq 1$. The kernels $W^{(\text{edge})}_{\G'}({\bf X},{\bf Y},{\bf Z})$ are explicit and analytic in $|\l|< \bar \l$, for some $\bar \l>0$ independent of $\beta, L$. They satisfy the bound:
\be\label{eq:wedge}
\int_{\beta, L} D \underline{X} D {\bf Y} D{\bf Z}\, \prod_{i<j} \| \underline{q}_{i} - \underline{q}_{j} \|_{\beta, L}^{n_{ij}} |q_{i,2} - q_{j,2}|^{m_{ij}} |W^{(\text{edge})}_{\G'}({\bf Q})| \leq \beta L C\;,\qquad \forall n_{ij}, m_{ij} \in \mathbb{N}
\ee
with ${\bf Q} = \{ {\bf X }, {\bf Y}, {\bf Z}\}$, and for some $C$ only dependent on $\{n_{ij}, m_{ij}\}$. The presence of the $\beta L$ factor follows from translation invariance in the ${\bf e}_{0}$, ${\bf e}_{1}$ directions\footnote{Notice that we are not integrating in the $x_{i,2}$ variables: this, together with the fast decay in $x_{2} - y_{2}$ of (\ref{eq:gbound}), is the reason why the right-hand side of (\ref{eq:wedge}) is bounded proportionally to $\beta L$ and not $\beta L^{2}$.}. The bulk effective action $\mathcal{W}^{(\text{bulk})}_{\beta,L,N}(A,\phi)$ is given by an expression like (\ref{eq:Vedge}), with $n=0$; the kernels satisfy the bounds (\ref{eq:wedge}), with $\beta L$ replaced by $\beta L^{2}$. Now we use that, as proven in Appendix \ref{app:2d1d}:
\be\label{eq:2d1d}
\int P_{N}(d\Psi^{\text{(edge)}}) e^{V^{\text{(edge)}}(\Psi^{\text{(edge)}}; A, \phi)} = \int P'_{N}(d\psi) e^{V^{\text{(edge)}}(\psi*\check\xi; A, \phi)}\;,
\ee
where $P'_{N}(d\psi)$ is a Grassmann Gaussian integration for some new $1+1$-dimensional Gaussian Grassmann fields $\psi^{\pm}_{\underline{x}, e}$ with covariance (\ref{eq:cov1d}), and
\be\label{eq:conv}
(\psi^{-} * \check{\xi})_{\xx, r} :=  \sum_{e}\sum_{h} \psi^{-}_{(x_{0}, h),e}\, \overline{\check{\xi}^{e}_{x_{2}}(x_{1} - h; r)}\;,\qquad (\psi^{+} * \check{\xi})_{\xx, r} :=  \sum_{e}\sum_{h} \psi^{+}_{(x_{0}, h),e}\, \check{\xi}^{e}_{x_{2}}(x_{1} - h; r)\;,
\ee
where we introduced:
\be\label{eq:checkxi}
\check{\xi}^{e}_{x_{2}}(x_{1};r) := \frac{1}{L}\sum_{k_{1} \in S^{1}_{L}} e^{-ik_{1} x_{1}}\, \xi^{e}_{x_{2}}(k_{1}; r) \chi((1/\d) |\e_{e}(k_{1}) - \m_{e}|)\;.
\ee
By the exponential decay of $\xi^{e}_{x_2}$, Eq. (\ref{eq:dec1}), and by the compact support of $\chi$ we easily get that, using the regularity in $k_{1}$ of both $\xi^{e}$ and $\chi$: 
\be\label{eq:checkbd}
|\check{\xi}^e_{x_2}(x_{1};r)|\leq C_{n} \frac{e^{-cx_{2}}}{1 + (\d |x_{1}|_{L})^{n}}\qquad \text{or}\qquad |\check{\xi}^e_{x_2}(x_{1};r)|\leq C_{n} \frac{e^{-c(L-x_{2})}}{1 + (\d |x_{1}|_{L})^{n}}\;,\qquad \forall n\in \mathbb{N}\;,
\ee
with $|\cdot|_{L}$ the distance on the circle of length $L$.  Finally, defining $V^{\text{(1d)}}(\psi; A, \phi) := V^{\text{(edge)}}(\psi*\check\xi; A, \phi)$ we get an expression of the form (\ref{eq:V1d}), with kernels:
\be\label{eq:edgeto1d}
W_{\G}^{\text{(1d)}}(\underline{X}, {\bf Y}, {\bf Z}) = \sum_{\{x_{k, 2},\, h_{k},\,  r'_{k}\}} \Big[\prod_{k=1}^{n} \tilde{\xi}^{e_{k}}_{x_{k,2}}(h_{k}; r'_{k})\Big]W^{(\text{edge})}_{\G'}( {\bf X} + H, {\bf Y}, {\bf Z})
\ee
with $\G_{\phi} = \G'_{\phi}$, $\G_{A} = \G'_{A}$ and $\G_{\psi} = \{ (e_{1}, \e_{1}), \ldots, (e_{n}, \e_{n}) \}$; $\tilde{\xi}^{e_{k}} = \check{\xi}^{e_{k}}$ if $\e_{k} = -$ or $\tilde{\xi}^{e_{k}} = \overline{\check{\xi}^{e_{k}}}$ if $\e_k = +$; ${\bf X} + H$ is a shorthand notation for $\{x_{i,0}, x_{i,1} + h_{i}, x_{i,2}\}_{i=1}^{n}$. As a consequence of the translation invariance of $W^{(\text{edge})}_{\G'}$, the kernels $W_{\G}^{\text{(1d)}}$ are also translation invariant in the ${\bf e}_{0}$ and ${\bf e}_{1}$ directions. To prove the bound (\ref{eq:decW1d}) we proceed as follows. We write (omitting the $\b, L$ labels in the distances):
\bea
&&\Big[ \prod_{i<j} \| \underline{q}_{i} - \underline{q}_{j} \|^{n_{ij}}\Big]\Big[ \prod_{l<k} |\tilde q_{l,2} - \tilde q_{k,2}|^{m_{l,k}}\Big] |W_{\G}^{\text{(1d)}}(\underline{Q}, \widetilde{Q}_{2})| \leq \sum_{\{ h_{k} \}} \Big[\prod_{i<j} \frac{\| \underline{q}_{i} - \underline{q}_{j} \|^{n_{ij}}}{\| \underline{q}_{i} - \underline{q}_{j} + \underline{h}_{ij} \|^{n_{ij}}} \frac{1}{1 + |h_{i}|^{m}}\Big] \nn\\&&\sum_{\substack{\{x_{k,2}\} \\  \{ r'_{k}\}}}\Big[ \prod_{k=1}^{n} |\tilde{\xi}^{e_{k}}_{x_{k,2}}(h_{k}; r'_{k})| |h_{k}|^{K} \Big] \big[ \prod_{i<j} \| \underline{q}_{i} - \underline{q}_{j} + \underline{h}_{ij} \|^{n_{ij}}\big] \big[ \prod_{k<l} | \tilde q_{k,2} - \tilde q_{l,2} |^{m_{kl}} \big] |W^{(\text{edge})}_{\G'}( {\bf Q} + H)|\nn
\eea
with ${\bf Q} + H_{1} = ({\bf X} + H_{1}, {\bf Y}, {\bf Z})$, $\underline{h}_{ij} = (0, h_{i} - h_{j})$, for $K\equiv K(m)$ and $m$ suitably large. Using that $\| \underline{q}_{i} - \underline{q}_{j} \|^{n_{ij}}\leq C( \| \underline{q}_{i} - \underline{q}_{j} + \underline{h}_{ij}\|^{n_{ij}} + \| \underline{h}_{ij} \|^{n_{ij}})$, the first product of the right-hand side can be bounded proportionally to $\prod_{i<j} \| \underline{h}_{ij} \|^{n_{ij}}/(1 + |h_{i}|^{m})$, which is summable in $\{h_{k}\}$ for $m$ large enough. Finally, using the bounds (\ref{eq:checkbd}) together with the bound on the edge kernels (\ref{eq:wedge}), Eq. (\ref{eq:decW1d}) follows.
\end{proof}
As mentioned in Remark \ref{rem:red}, the parameters $\m_{e}$ appearing in Proposition \ref{prp:1d} play the role of renormalized chemical potentials. We rewrite the Gaussian integration as:
\be
\int P'_{ N}(d\psi) e^{V^{(\text{1d})}(\psi; A, \phi)} = \int \widetilde{P}_{N}(d\psi) e^{V^{\text{(1d)}}(\psi; A, \phi) + \sum_{e} \n_{e} (\psi^{+}_{e}, \psi^{-}_{e})}\;,
\ee
where $(\psi^{+}_{e}, \psi^{-}_{e}) = \int_{0}^{\beta} dx_{0} \sum_{x_{1}} \psi^{+}_{\underline{x}, e} \psi^{-}_{\underline{x}, e}$, and the new Grassmann Gaussian integration $\widetilde{P}_{N}(d\psi)$ has covariance:
\be
\int \widetilde{P}_{N}(d\psi) \hat \psi^{-}_{\underline{k},e} \hat \psi^{+}_{\underline{q},e'} = \beta L \delta_{e,e'} \delta_{\underline{k},\underline{q}} \frac{\chi_{N}(k_{0}) \chi_{e}(k_{1})}{-ik_{0} + \e_{e}(k_{1}) - \m + \n_{e}(\underline{k})}
\ee
with $\n_{e}(\underline{k}) = \n_{e}\chi_{N}(k_{0}) \chi_{e}(k_{1})$. The parameters $\n_{e}\equiv \n_{e}(0, k_{F}^{e})$ will be chosen later on; for the moment, we shall only suppose that $|\n_{e}|\leq C\l$, and we shall set $\m_{e} := \m - \nu_{e}$. From now on, we shall denote by $k_{F}^{e} \equiv k_{F}^{e}(\l)$ the {\it interacting Fermi momentum}, that is the solution of $\e_{e}(k_{F}^{e}) = \m - \n_{e} \equiv \m_{e}$. 

We are now ready to integrate the field $\psi$. Being the covariance massless, we cannot integrate the field in a single step. Instead, we will proceed in a multiscale fashion, by decomposing the field over single scale fields, living on well defined energy scales. We start by writing $\psi = \psi^{\text{(u.v.)}} + \psi^{(\text{i.r.})}$, where the $\psi^{\text{(u.v.)}}$ has propagator given by, setting $\underline{k} = \underline{k}_{F}^{e} + \underline{k}'$:
\bea\label{eq:1N}
\int P_{[1,N]}(d\psi^{(\text{u.v.})}) \hat \psi^{\text{(u.v.)}-}_{\underline{k}, e} \hat \psi^{\text{(u.v.)}+}_{\underline{q}, e'} &=&  \beta L \delta_{\underline{k}, \underline{q}} \d_{e,e'} \hat g^{\text{(u.v.)}}_{e}(\underline{k})\\
\hat g^{\text{(u.v.)}}_{e}(\underline{k}) &=& \frac{[\chi_{N}(k_{0})\chi_{e}(k_{1}) - \chi_{0,e}(\underline{k}')]}{-ik_{0} + \e_{e}(k_{1}) - \m + \n_{e}(k_{1})}\;,\nn
\eea
where $\chi_{0,e}(\underline{k}') := \chi\big((1/\d'_{e})\sqrt{k^2_{0} + v^2_{e} {{k}_{1}^{'2}}}\big)$ and $\d_{e}'>0$ small enough such that the support of $\chi_{0,e}(\underline{k}')$ is contained in the support of $\chi_{N}(k_{0})\chi_{e}(k_{1})$, and
\be\label{eq:gir}
\int P_{\text{i.r.}}(d\psi^{(\text{i.r.})}) \hat \psi^{\text{(i.r.)}-}_{\underline{k}, e} \hat \psi^{\text{(i.r.)}+}_{\underline{q}, e'} =  \beta L \delta_{\underline{k}, \underline{q}}\d_{e,e'} \hat g^{\text{(i.r.)}}_{e}(\underline{k})\;,\quad \hat g^{\text{(i.r.)}}_{e}(\underline{k}) = \frac{\chi_{0,e}(\underline{k}')}{-ik_{0} + \e_{e}(k_{1}) - \m_{e}}\;.
\ee
Here we used that, by possibly taking a smaller $\d'_{e}$, $\n_{e}(\underline{k}) = \n_{e}(0, k_{F}^{e})\equiv \nu_{e}$ for $\underline{k}'$ in the support of $\chi_{0,e}(\underline{k}')$. Now, the propagator $g^{\text{(u.v.)}}_{e}$ satisfies the bound
\be
| g^{\text{(u.v.)}}_{e}(\underline{x} - \underline{y}) | \leq \frac{C_{n} \d^{-1}}{1 + (\d \| \underline{x} - \underline{y} \|_{\b, L})^{n}}\;,\qquad \forall n\in \mathbb{N}\;,
\ee
which can be used to integrate the $\psi^{\text{(u.v.)}}$ in a single step, proceeding as for $\Psi^{(\text{bulk})}$. The integration of the $\psi^{(\text{u.v.)}}$ field can be performed uniformly in the ultraviolet cutoff $N$; we refer to, e.g., Section 5.3 of \cite{GMP} for the details. The result is:
\be\label{eq:RGstart}
\mathcal{W}_{\beta,L}(A,\phi) =  \mathcal{W}^{(\text{bulk})}_{\beta,L}(A,\phi) + \mathcal{W}^{(\text{u.v.})}_{\beta,L}(A,\phi) + \log \int P_{\text{i.r.}}(d\psi^{\text{(i.r.)}}) e^{V^{\text{(i.r.)}}(\psi^{\text{(i.r.)}}; A, \phi)}
\ee
where $V^{\text{(i.r.)}}(\psi^{\text{(i.r.)}}; A, \phi), \mathcal{W}^{(\text{u.v.})}_{\beta,L}(A,\phi)$ have the form (\ref{eq:V1d}), with kernels $W^{\text{(i.r.)}}_{\G}$ satisfying the bounds (\ref{eq:decW1d}).

\section{The reference model} \label{sec:ref}

The result of the previous section is that the generating functional of the correlations can be expressed in terms of the generating functional of a suitable one-dimensional quantum field theory, with propagator given by $g^{(\text{i.r.})}_{e}$, Eq. (\ref{eq:gir}), and effective action $V^{\text{(i.r.)}}(\psi^{\text{(i.r.)}}; A, \phi)$, involving arbitrarily high powers of the fields monomials. Being $V^{\text{(i.r.)}}(\psi^{\text{(i.r.)}}; A, \phi)$ a very involved object, computing directly the correlation functions of this effective one-dimensional theory is, in general, an impossible task.

In this section we shall introduce a {\it reference model} for the effective one-dimensional theory obtained in Section \ref{sec:red1d}, which captures the asymptotic behavior of the correlations of the class of models specified by Assumption 2, and for which the current-current correlation functions and the two-point Schwinger function can be computed explicitly. This model is the {\it chiral Luttinger model}, investigated in \cite{FM, BFMhub2} via rigorous RG methods. Proposition \ref{prp:relref} will then allow to write the correlation functions of the full lattice model in terms of those of the reference model, after a suitable fine tuning of its bare parameters. The proof of Proposition \ref{prp:relref} will be postponed to Section \ref{sec:relref}.

\subsection{Chiral relativistic fermions}

Let $\L_{\beta, L,\frak{M}} \subset \frac{\beta}{\frak{M}} \mathbb{Z} \times \frac{L}{\frak{M}} \mathbb{Z}$ be the lattice made of the points $\underline{x} = (n_{0} \frac{\beta}{\frak{M}}\, , n_{1} \frac{L}{\frak{M}})$, $n_{i} = 0,1, \ldots, \frak{M}$. With each lattice site $\underline{x} \in \L_{\beta, L, \frak{M}}$ we associate a Grassmann field $\psi^{\pm}_{\underline{x}, \o, \s}$, with $\o = \pm$ and $\s = \uparrow\downarrow$. Let $\chi_{N}\big(\underline{k}) := \chi(2^{-N}\sqrt{k_{0}^{2} + v^{\text{(ref)}2} k_{1}^{2}}\big)$, $v^{\text{(ref)}}\in \mathbb{R}$. We introduce the Grassmann Gaussian measure $P^{\text{(ref)}}_{\leq N}(d\psi)$ as:
\bea
P^{\text{(ref)}}_{\leq N}(d\psi) &:=& \mathcal{N}_{\beta,L,N,\frak{M}}^{-1} \big[ \prod_{\o =  \pm}\prod_{\s = \uparrow\downarrow} \prod_{\underline{x}\in \L_{\beta, L, \frak{M}}} d\psi^{+}_{\underline{x}, \o, \s} d\psi^{+}_{\underline{x}, \o, \s}\big]\, e^{-(\psi^{+}, g^{(\leq N)-1} \psi^{-})}\nn\\
(\psi^{+}, g^{(\leq N)-1} \psi^{-}) &:=& \sum_{\o, \s} \int_{\beta, L,\frak{M}} \frac{d\underline{k}}{(2\pi)^2}\,  \hat\psi^{+}_{\underline{k}, \o, \s} \hat g^{(\leq N)}_{\o}(\underline{k})^{-1} \hat\psi^{-}_{\underline{k}, \o, \s}\;,
\eea
where: $\mathcal{N}_{\beta,L,N,\frak{M}}$ is a suitable normalization factor; $\int_{\beta, L, \frak{M}} \frac{d\underline{k}}{(2\pi)^2}$ is a short-hand notation for $\frac{1}{\beta L} \sum_{\underline{k} \in \mathbb{D}^{*}_{\beta, L, \frak{M}}}$, with $\mathbb{D}^{*}_{\beta, L, \frak{M}} := \{ \underline{k} = (\frac{2\pi}{\beta}(n_{0} + \frac{1}{2}), \frac{2\pi}{L}(n_{1} + \frac{1}{2}) ) \mid n_{i} = 1,\ldots, \frak{M}\;,\; \chi_{N}(\underline{k})>0\}$; the Fourier transform of the field is $\hat\psi^{\pm}_{\underline{k},\o, \s} = (\beta L/\frak{M}^2) \sum_{\underline{x}\in \L_{\b, L, \frak{M}}} e^{\mp i \underline{k}\cdot \underline{x}} \psi^{\pm}_{\underline{x}, \o, \s}$ for $\underline{k}\in \mathbb{D}^{*}_{\beta, L, \frak{M}}$ and zero otherwise; and, finally, the propagator is:
\be
\hat g^{(\leq N)}_{\o}(\underline{k}) := \frac{1}{Z^{\text{(ref)}}} \frac{\chi_{N}(\underline{k})}{-ik_{0} + \o v^{\text{(ref)}}k_{1}}\;.
\ee
We define the many-body interaction of the reference model as:
\be
V(\sqrt{Z^{\text{(ref)}}}\psi) := \l^{\text{(ref)}} {Z^{\text{(ref)}}}^{2} \sum_{\s,\s'} \sum_{\o} \int_{\beta, L, \frak{M}} d\underline{x}d\underline{y}\, w(\underline{x} - \underline{y})\, n_{\underline{x}, \s, \o} n_{\underline{y}, \s', \o}\;,
\ee
where: $\int_{\beta, L, \frak{M}}\, d\underline{x}$ is a shorthand notation for $(\beta L/\frak{M}^2)\sum_{\underline{x} \in \L_{\beta,L, \frak{M}}}$, $w(\underline{x} - \underline{y})$ is the restriction to $\L_{\b, L, \frak{M}}$ of a smooth, short-ranged, rotation invariant interaction potential on $\mathbb{R}^{2}$, and $n_{\underline{x}, \s, \o} = \psi^{+}_{\underline{x}, \s, \o}\psi^{-}_{\underline{x}, \s, \o}$. The real parameters $\l^{\text{(ref)}}$, $Z^{\text{(ref)}}>0$, $v^{\text{(ref)}}>0$, so far unknown, will be chosen later on. For later convenience, we will assume that $\hat w(\underline{0}) = 1$. The generating functional of the Schwinger functions of the reference model is
\be
\mathcal{W}_{\beta, L, N, \frak{M}}^{(\text{ref})}(\psi; A, \phi) := \log \int P_{\leq N}^{(\text{ref})}(d\psi) e^{-V(\sqrt{Z^{\text{(ref)}}}\psi) + \G(\sqrt{Z^{\text{(ref)}}}\psi; A) + B(\psi; \phi)}
\ee
where, for $\hat n_{\underline{p},\o,\s} := \int_{\beta, L, \frak{M}} \frac{d\underline{k}}{(2\pi)^2}\, \hat\psi^{+}_{\underline{k}+\underline{p},\o,\s}\hat\psi^{-}_{\underline{k}, \o, \s}$, $\hat n^{c}_{\underline{p}, \o} = \hat n_{\underline{p}, \o, \uparrow} + \hat n_{\underline{p}, \o, \downarrow}$, $\hat n^{s}_{\underline{p}, \o} = \hat n_{\underline{p}, \o, \uparrow} - \hat n_{\underline{p}, \o, \downarrow}$:
\bea\label{eq:source}
\G(\psi; A) &:=& \sum_{\o,\sharp}\sum_{\m=0,1} \sum_{x_{2} = 0}^{L}\int_{\beta, L, \frak{M}} \frac{d\underline{p}}{(2\pi)^{2}} \, \hat A^{\sharp}_{\underline{p}, x_{2}, \m, \o} \hat n^{\sharp}_{\underline{p}, \o} Z_{\sharp,\m}^{(\text{ref})}(x_{2}) \\
B(\psi; \phi) &:=& \sum_{\substack{r = (\bar r, \s) \\ \o}} \sum_{x_{2} = 0}^{L}\int_{\beta, L, \frak{M}} \frac{d\underline{k}}{(2\pi)^2}\, \big[ \hat\psi^{+}_{\underline{k}, \o, \s} Q^{(\text{ref})-}_{r}(x_{2})\hat \phi^{-}_{\underline{k},x_{2}, r,\o} + \hat \phi^{+}_{\underline{k}, x_{2},r,\o} Q^{\text{(ref)}+}_{r}(x_{2}) \hat \psi^{-}_{\underline{k}, \o, \s}\big]\nn
\eea
with $\int_{\beta, L, \frak{M}} \frac{d\underline{p}}{(2\pi)^{2}}$ a short-hand notation for $\frac{1}{\beta L} \sum_{\underline{k} \in \widetilde{\mathbb{D}}_{\beta, L, \frak{M}}}$ where $\widetilde{\mathbb{D}}_{\beta, L, \frak{M}} := \{ \underline{p} = (\frac{2\pi}{\beta}n_{0}, \frac{2\pi}{L}n_{1}) \mid n_{i} = 1,\ldots, \frak{M}\}$. As for the other parameters, the functions $Z^{\text{(ref)}}_{\sharp,\m}(x_{2})$, $Q^{\text{(ref)}\pm}_{r}(x_{2})$ will be chosen later on. We define the correlation functions of the reference model as derivatives with respect to the external fields of $\mathcal{W}_{\beta, L, N}^{(\text{ref})}(\psi; A, \phi) := \lim_{\frak{M}\to\infty} \mathcal{W}_{\beta, L, N, \frak{M}}^{(\text{ref})}(\psi; A, \phi)$. For instance:
\bea
&&Q^{(\text{ref})-}_{ r'}(y_{2}) Q^{(\text{ref})+}_{r}(x_{2}) \pmb{\langle} \hat \psi^{-}_{\underline{k},\o,\s} \hat \psi^{+}_{\underline{k}, \o', \s'} \pmb{\rangle}^{\text{(ref)}}_{\beta, L, N} := \frac{\partial^{2} \mathcal{W}_{\beta, L, N}^{(\text{ref})}(\psi; A, \phi)}{\partial \hat \phi^{+}_{\underline{k}, x_{2}, r, \o} \partial \hat \phi^{-}_{\underline{k}, y_{2},  r', \o'}}\Big|_{A= \phi = 0}\nn\\
&&Z^{(\text{ref})}_{\sharp,\m}(z_{2}) Q^{(\text{ref})+}_{ r'}(y_{2})Q^{(\text{ref})-}_{ r''}(x_{2}) \pmb{\langle} \hat n^{\sharp}_{\underline{p}, \o}\,; \hat \psi^{-}_{\underline{k}+\underline{p},\o',\s'} \hat \psi^{+}_{\underline{k}, \o'', \s''} \pmb{\rangle}^{\text{(ref)}}_{\beta, L, N} \nn\\&& \qquad\qquad\qquad\qquad\qquad\qquad\qquad\qquad := \frac{\partial^{3} \mathcal{W}_{\beta, L, N}^{(\text{ref})}(\psi; A, \phi)}{\partial \hat A^{\sharp}_{\underline{p},z_{2}, \m,\o}\partial \hat \phi^{+}_{\underline{k}+\underline{p}, x_{2},  r',\o'} \partial \hat \phi^{-}_{\underline{k},y_{2},  r'',\o''}}\Big|_{A= \phi = 0} \nn\\
&&Z^{(\text{ref})}_{\sharp,\m}(z_{2}) Z^{(\text{ref})}_{\sharp',\n}(w_{2}) \pmb{\langle} n^{\sharp}_{\underline{p}, \o}\,; n^{\sharp'}_{-\underline{p}, \o'} \pmb{\rangle}_{\beta, L, N}^{\text{(ref)}} :=   \frac{\partial^{2} \mathcal{W}_{\beta, L, N}^{(\text{ref})}(\psi; A, \phi)}{\partial \hat A^{\sharp}_{\underline{p}, z_{2},\m, \o} \partial A^{\sharp'}_{-\underline{p},w_{2},\n, \o'}}\Big|_{A= \phi = 0}\;.
\eea
Also, we shall set $\pmb{\langle}\cdots \pmb{\rangle}_{\beta, L} := \lim_{N\to\infty} \pmb{\langle}\cdots \pmb{\rangle}_{\beta, L, N}$. The generating functional of the correlation functions can be constructed using rigorous renormalization group methods, \cite{FM}. As discussed in the next proposition, the reference model can be used to capture the asymptotic behaviour of the correlation funciton of interacting Hall insulators exhibiting single-channel edge models, by properly tuning the bare parameters $\l^{\text{(ref)}}$, $Z^{\text{(ref)}}$, $v^{\text{(ref)}}$, $Z^{\text{(ref)}}_{\sharp,\m}(z_{2})$, $Q^{(\text{ref})\pm}(x_{2})$. 

Recall that $e = (1,\s)$ are the edge states localized around $x_{2} = 0$, Eq. (\ref{eq:dec22}). In the following, we shall set $\o := \text{sgn}(v_{(1,\s)})$, and $v_{\o} := v_{(1,\s)}$. Similarly, we shall set $k_{F}^{\o} \equiv k_{F}^{(1,\s)}$ and $\xi^{\o}_{x_{2},r} \equiv \xi^{(1,\s)}_{x_2}(k_{F}^{(1,\s)}; r)$.

\begin{prop}{\bf (Relation with the lattice model.)}\label{prp:relref} Let $\l$ be the bare coupling constant of the lattice model. There exists $\bar \l>0$ independent of $\beta ,L$ such that for $|\l| <\bar \l$ there exists choices of the bare parameters of the reference model satisfying, for all $n\in \mathbb{N}$:
\bea
&&\l^{(\text{ref})} = A \l + O(\l^2)\;,\quad v^{(\text{ref})} = |v_{\o}| + O(\l)\;,\quad Z^{(\text{ref})} = 1+O(\l)\;,\quad |Z^{(\text{ref})}_{\sharp,\m}(x_{2})|\leq \frac{C_{n}}{1 + x_{2}^{n}}\;,\nn\\
&& |Q^{(\text{ref})+}_{r}(x_{2})| \leq \frac{C_{n}}{1 + x_{2}^{n}}\;,\quad \| Q^{(\text{ref})+}_{r} \|_{2} = 1+O(\l)\;,\quad Q^{(\text{ref})-}_{r}(\underline{k}, x_{2}) = \overline{Q^{(\text{ref})+}_{r}(\underline{k}, x_{2})}\;,
\eea
with $A$ given by Eq. (\ref{eq:A}), such that the following is true. For $\|\underline{k}'\| = \kappa$, $\| \underline{p} \|\ll \kappa$, $\| \underline{k}' + \underline{p} \| \leq \kappa$ and $\kappa$ small enough,
\bea\label{eq:latWI}
&&\pmb{\langle} {\bf T} \hat a^{-}_{\underline{k}' + \underline{k}_{F}^{\o},x_{2}, r} \hat a^{+}_{\underline{k}' + \underline{k}_{F}^{\o},y_{2}, r'} \pmb{\rangle}_{\beta, L} \\&& = \pmb{\langle} \hat\psi^{-}_{\underline{k}',\o,\s} \hat\psi^{+}_{\underline{k}', \o,\s'} \pmb{\rangle}_{\beta, L}^{(\text{ref})} \Big[Q^{(\text{ref})+}_{r}(x_{2}) \overline{Q^{(\text{ref})-}_{r'}(y_{2})} + R^{\beta, L}_{\o,r,r'}(\underline{k}'; x_{2}, y_{2})\Big]\nn\\
&&\pmb{\langle} {\bf T} \hat j^{\sharp}_{\m, \underline{p}, z_{2}}\,; \hat a^{-}_{\underline{k}' + \underline{p}+\underline{k}_{F}^{\o},x_{2}, r} \hat a^{+}_{\underline{k}' + \underline{k}_{F}^{\o},y_{2}, r'} \pmb{\rangle}_{\beta, L} \nn\\&& = \pmb{\langle} \hat n^{\sharp}_{\underline{p}, \o}\,; \hat \psi^{-}_{\underline{k}'+\underline{p} ,\o,\s} \hat \psi^{+}_{\underline{k}' ,\o,\s'}  \pmb{\rangle}^{\text{(ref)}}_{\beta, L}\Big[ Z^{(\text{ref})}_{\sharp,\m}(z_{2}) Q^{(\text{ref})+}_{r}(x_{2}) \overline{Q^{(\text{ref})-}_{r'}(y_{2})} + R^{\beta, L}_{\m, \o, r,  r'}(\underline{k}', \underline{p}; x_{2}, y_{2}, z_{2}) \Big] \nn\\
&&\pmb{\langle} {\bf T} \hat j^{\sharp}_{\m, \underline{p}, x_{2}}\,; \hat j^{\sharp'}_{\n, -\underline{p}, y_{2}} \pmb{\rangle}_{\beta, L} = Z^{(\text{ref})}_{\sharp,\m}(x_{2})Z^{(\text{ref})}_{\sharp',\n}(y_{2}) \pmb{\langle} \hat n^{\sharp}_{\underline{p}, \o}\,; \hat n^{\sharp'}_{-\underline{p},\o} \pmb{\rangle}_{\beta, L}^{(\text{ref})} + R^{\beta, L}_{\m,\n,\underline{\sharp}}(\underline{p}; x_{2}, y_{2})\nn
\eea
%
%
%
where $\m,\n =0,1$; the error terms are bounded as, for all $n\in \mathbb{N}$, and for some $\theta,\, C_{n},\, C>0$ independent of $\beta, L$:
\bea\label{eq:bdR}
&& |R^{\infty}_{\o, r,  r'}(\underline{k}'; x_{2}, y_{2})| \leq \frac{C_{n} \kappa^{\theta}}{1+|x_{2} - y_{2}|^{n}}\;,\quad |R^{\infty}_{\m, \o, r,  r'}(\underline{k}', \underline{p}; x_{2}, y_{2}, z_{2})|\leq \frac{C_{n} \kappa^{\theta}}{1 + d(x_{2}, y_{2}, z_{2})^{n}} \nn\\
&&|R^{\infty}_{\m,\n,\underline{\sharp}}(\underline{p}; x_{2}, y_{2})|\leq \frac{C_{n}}{1 + |x_{2} - y_{2}|^{n}} \;,\quad  \sum_{y_{2}} | R^{\infty}_{\m,\n,\underline{\sharp}}(\underline{p}; x_{2}, y_{2}) - R^{\infty}_{\m,\n,\underline{\sharp}}(\underline{0}; x_{2}, y_{2}) | \leq C|\underline{p}|^{\theta}
\eea
with $d(x_{2}, y_{2}, x_{2}) = \min\{ |x_{2} - y_{2}| + |y_{2} - z_{2}|\;, |x_{2} - z_{2}|+ |z_{2} - y_{2}| \}$.
\end{prop}
The faster than any power decays could be improved to exponential decay, but it will not be needed here. This proposition is the main technical ingredient of the paper, and will be proved in Section \ref{sec:relref}. Before discussing this, let us show how to use Proposition \ref{prp:relref}
 to prove Theorem \ref{thm:1}. To do this, we will need to compute explicitly the coefficients $Z^{\text{(ref)}}_{\sharp,\m}(z_{2})$, together with the correlation functions of the reference model. Both problems will be solved using {\it Ward identities}.
\section{Ward identities}\label{sec:WI}
{\it Ward identities} are nonperturbative relations among correlation functions, implied by conservation laws. In Section \ref{sec:WIlat} we will derive Ward identities for the correlation functions of the lattice model, implied by the conservation of current. In Section \ref{sec:WIref} we shall derive Ward identities for the reference model, as a consequence of $U(1)$ gauge symmetry. As we shall see, in this last case the presence of the momentum regularization gives rise to {\it anomalies}.
\subsection{Ward identities for the lattice model}\label{sec:WIlat}
Let us consider the continuity equation, Eq. (\ref{eq:Jcons}). Writing ${\bf T} \rho_{\xx} j_{\n,\yy} = \theta(x_{0} - y_{0}) \rho_{\xx} j_{\n,\yy} + \theta(y_{0} - x_{0}) j_{\n,\yy} \rho_{\xx}$, it implies:
\be
i\partial_{x_{0}} \langle {\bf T} \rho_{\xx}\,; j_{\n,\yy} \rangle_{\beta, L} = -\text{div}_{x} \langle {\bf T}j_{\xx} \,; j_{\n,\yy}\rangle_{\beta,L} + i\langle [\rho_{\xx}\,, j_{\n,\yy}]\rangle_{\beta,L}\delta(x_{0}-y_{0})\;,
\ee
where the last term follows from the derivative of the $\theta$ functions, and $\d(x_{0}- y_{0})$ is the Dirac delta. The commutator in the right-hand side can be computed explicitly (using the notation $\vec e_{3} \equiv \vec e_{1}$):
\bea\label{eq:schwinger}
\langle [\rho_{\vec x}\, , j_{i,\vec y}] \rangle_{\beta,L} &=& i(\delta_{\vec x, \vec y} - \delta_{\vec x, \vec y + \vec e_{i}}) \tau_{\vec y, \vec y + \vec e_{i}} + \frac{i}{2}( \delta_{\vec x, \vec y} - \delta_{\vec x,\vec y + \vec e_{i} - \vec e_{i+1}}) \tau_{\vec y, \vec y +\vec e_{i} - \vec e_{i+1}}\nn\\
&& + \frac{i}{2}( \delta_{\vec x, \vec y} - \delta_{\vec x, \vec y + \vec e_{i} + \vec e_{i+1}} )\tau_{\vec y, \vec y + \vec e_{i} + \vec e_{i+1}} + \frac{i}{2}( \delta_{\vec x,\vec y - \vec e_{i+1}} - \delta_{\vec x, \vec y + \vec e_{i}} )\tau_{\vec y - \vec e_{i+1}, \vec y + \vec e_{i}}\nn\\&& + \frac{i}{2}( \delta_{\vec x, \vec y + \vec e_{i+1}} - \delta_{\vec x, \vec y + \vec e_{i}} )\tau_{\vec y + \vec e_{i+1}, \vec y + \vec e_{i}}\;,\qquad i =1,2\;,
\nn\\
\langle [\rho_{\vec x}\,, j_{0,\vec y}] \rangle_{\beta,L} &=& 0\;,
\eea
with $\delta_{\vec x, \vec y}$ the Kronecker delta and $\tau_{\vec x, \vec y}$ given by Eq. (\ref{eq:tauxy}).
%
%
Therefore, after summing over $x_{2}$ one gets, recalling the Dirichlet boundary conditions $\sum_{x_{2} = 0}^{L} d_{2} j_{2,\xx} = -j_{2,(\underline{x}, 0)} + j_{2,(\underline{x}, L)} = 0$:
\bea
\text{d}_{x_{0}} \sum_{x_{2} = 0}^{L} \langle {\bf T} \rho_{\xx}\,; j_{0,\yy} \rangle_{\beta, L} + \text{d}_{x_1} \sum_{x_{2} = 0}^{L} \langle {\bf T} j_{1,\xx}\,; j_{0,\yy} \rangle_{\beta, L} &=& 0\\
\text{d}_{x_{0}} \sum_{x_{2} = 0}^{L} \langle {\bf T} \rho_{\xx}\,; j_{1,\yy} \rangle_{\beta, L} + \text{d}_{x_1} \sum_{x_{2} = 0}^{L} \langle {\bf T} j_{1,\xx}\,; j_{1,\yy} \rangle_{\beta, L} &=& -\delta(x_{0} - y_{0}) ( \delta_{x_{1}, y_{1}} - \delta_{x_{1}, y_{1} + 1} ) \D^{\beta, L}_{1, y_{2}}\;,\nn
\eea
where $\D^{\b,L}_{1, y_{2}}$ has been defined in Eq. (\ref{eq:Delta}). For short, setting $\D_{0, y_{2}}^{\beta, L} \equiv 0$, we have:
\be\label{eq:WI}
\sum_{\m = 0,1}\text{d}_{x_\m} \sum_{x_{2} = 0}^{L}\langle{\bf T} j_{\m,\xx}\,; j_{\n,\yy} \rangle_{\beta,L} = -\delta(x_{0} - y_{0}) (\delta_{x_{1}, y_{1}} - \delta_{x_{1}, y_{1} + 1}) \D^{\beta,L}_{\n,y_{2}}\;,\qquad \n = 0,1\;. 
\ee
The $\D^{\beta,L}_{\n, y_{2}}$ term is called the {\it Schwinger term}. Eq. (\ref{eq:WI}) is the Ward identity for the current-current correlation functions. 
%
%
One can also derive a Ward identity relating the vertex function of the lattice model to the two point correlation functions. Setting $\d_{\xx,\yy} \equiv \d_{\vec x, \vec y}\d(x_{0} - y_{0})$:
\bea\label{eq:WIvertex}
&&\sum_{\m=0,1,2}d_{z_\m} \langle {\bf T} j_{\m, \zz}\,; a^{-}_{\yy, r'} a^{+}_{\xx, r} \rangle_{\beta,L} = i\big[\langle {\bf T} a^{-}_{\yy ,r'} a^{+}_{\xx, r}\rangle_{\beta,L}\delta_{\xx, \zz} - \langle {\bf T} a^{-}_{\yy , r'} a^{+}_{\xx, r}\rangle_{\beta,L}\delta_{\yy,\zz}\big]\;.
\eea
Taking the Fourier transform, Eqs. (\ref{eq:WI}), (\ref{eq:WIvertex}) imply, for $p_{0} \in \frac{2\pi}{\beta}\mathbb{Z}$ and $p_{1}\in \frac{2\pi}{L}\mathbb{Z}$, introducing $\eta_{0}(\underline{p}) := i$, $\eta_{1}(\underline{p}) := (1 - e^{ip_{1}})/(-ip_{1})$:
\bea\label{eq:sumWI}
&& \sum_{\m = 0,1}\eta_{\m}(\underline{p}) p_{\m}\sum_{x_{2} = 0}^{L} \pmb{\langle} {\bf T} \hat j_{\m,(\underline{p}, x_{2})}\,; \hat j_{\n,(-\underline{p},y_{2})} \pmb{\rangle}_{\beta,L}  = -\eta_{\n}(p) p_{\n} \D^{\beta,L}_{\n, y_{2}}\;,\qquad \n=0,1\;,\nn\\
&& -\sum_{\m=0,1} \eta_{\m}(\underline{p}) p_{\m} \sum_{z_{2} = 0}^{L} \pmb{\langle} {\bf T} \hat j_{\m, \underline{p}, z_{2}}\,;  \hat a^{-}_{\underline{k}+\underline{p}, y_{2}, r'} \hat a^{+}_{\underline{k}, x_{2},r} \pmb{\rangle}_{\beta,L} \nn\\&&\qquad\qquad\qquad\qquad = \pmb{\langle} {\bf T} \hat a^{-}_{\underline{k}, y_{2}, r'} \hat a^{+}_{\underline{k}, x_{2}, r}\pmb{\rangle}_{\beta,L} -  \pmb{\langle} {\bf T} \hat a^{-}_{\underline{k}+\underline{p}, y_{2}, r'} \hat a^{+}_{\underline{k} + \underline{p}, x_{2}, r} \pmb{\rangle}_{\beta,L}\;
\eea
with $\pmb{\langle} \cdot \pmb{\rangle}_{\beta, L} = (\beta L)^{-1}\Tr \cdot e^{-\beta (\mathcal{H} - \mu \mathcal{N})}/\mathcal{Z}_{\beta, L}$. One can derive similar identities for the spin densities and spin currents:
\bea\label{eq:WIspinlat}
&& \sum_{\m=0,1}\eta_{\m}(\underline{p}) p_{\m}\sum_{x_{2} = 0}^{L} \pmb{\langle} {\bf T} \hat j^{s}_{\m,(\underline{p}, x_{2})}\,; \hat j^{s}_{\n,(-\underline{p},y_{2})} \pmb{\rangle}_{\beta,L}  = -\eta_{\n}(p) p_{\n} \D^{\beta,L}_{\n, y_{2}}\;,\qquad \n = 0,1\;,\nn\\
&& -\sum_{\m=0,1} \eta_{\m}(\underline{p}) p_{\m} \sum_{z_{2} = 0}^{L} \pmb{\langle} {\bf T} \hat j^{s}_{\m, \underline{p}, z_{2}}\,; \hat a^{-}_{\underline{k}+\underline{p}, y_{2}, r'} \hat a^{+}_{\underline{k}, x_{2},r} \pmb{\rangle}_{\beta,L} \nn\\&&\qquad\qquad\qquad\qquad = \s[\pmb{\langle} {\bf T} \hat a^{-}_{\underline{k}, y_{2}, r'} \hat a^{+}_{\underline{k}, x_{2}, r}\pmb{\rangle}_{\beta,L} -  \pmb{\langle} {\bf T} \hat a^{-}_{\underline{k}+\underline{p}, y_{2}, r'} \hat a^{+}_{\underline{k} + \underline{p}, x_{2}, r} \pmb{\rangle}_{\beta,L} ]\;,
\eea
where $r = (\s, \bar r)$. 
%
%
%
%
All these lattice Ward identities will allow to compute the renormalized coefficients that connect the lattice correlations to those of the reference model, recall Proposition \ref{prp:relref}.
\subsection{Ward identities for the reference model}\label{sec:WIref}
Ward identities can be derived for the reference model as well. They follow from the fact that the $U(1)$ transformation $\psi^{\pm}_{\underline{x}, \o, \s} \to e^{i\alpha_{\underline{x},\s,\o}} \psi^{\pm}_{\underline{x}, \o, \s}$ has Jacobian equal to one. Thus, we have:
\bea\label{eq:WIref}
&&\int P_{\leq N}^{(\text{ref})}(d\psi) e^{-V(\sqrt{Z^{\text{(ref)}}}\psi) + \G(\sqrt{Z^{\text{(ref)}}}\psi; A) + B(\psi; \phi)} \nn\\&&\qquad\qquad\qquad\qquad= \int P_{\leq N}^{(\text{ref})}(d\psi) e^{-V(\sqrt{Z^{\text{(ref)}}}\psi) + \G(\sqrt{Z^{\text{(ref)}}}\psi; A) + B(\psi; e^{i\alpha}\phi) + \widetilde\G(\psi; \alpha)}
\eea
where $B(\psi; e^{i\alpha}\phi)$ is given by Eq. (\ref{eq:source}) with $\phi_{\underline{x}, x_{2}, r, \o}^{\mp}\to e^{\pm i\alpha_{\underline{x}, \s, \o}}\phi_{\underline{x}, x_{2}, r, \o}^{\mp}$, and
\be
\widetilde\G(\psi; \alpha) := (e^{i\alpha}\psi^{+}, g^{(\leq N)-1} e^{-i\alpha}\psi^{-}) - (\psi^{+}, g^{(\leq N)-1} \psi^{-})\;.
\ee
The identity (\ref{eq:WIref}) can be used to obtain relations among the correlation functions of the reference model. Let $\alpha_{\underline{x}, \o, \s}$ be a periodic function on $\L_{\beta, L, \frak{M}}$. Differentiating both sides of Eq. (\ref{eq:WIref}) with respect to $\hat \alpha_{\underline{p}, \o, \s'}$, $\hat \phi^{+}_{\underline{k}+\underline{p},x_{2}, r, \o}$, $\hat \phi^{-}_{\underline{k}, y_{2}, r', \o'}$ one gets the following identity, Eq. (91) of \cite{FM}:
\bea
D_{\o}(\underline{p}) Z^{\text{(ref)}} \pmb{\langle} \hat n_{\underline{p}, \o, \s'}\,; \hat \psi^{-}_{\underline{k}+\underline{p}, \o, \s} \hat \psi^{+}_{\underline{k}, \o, \s} \pmb{\rangle}_{\beta, L, N}^{(\text{ref})} &=& \d_{\s\s'} [ \pmb{\langle} \hat \psi^{-}_{\underline{k}, \o, \s} \hat \psi^{+}_{\underline{k}, \o, \s} \pmb{\rangle}_{\beta, L, N}^{(\text{ref})} - \pmb{\langle} \hat \psi^{-}_{\underline{k}+\underline{p}, \o, \s} \hat \psi^{+}_{\underline{k}+\underline{p}, \o, \s} \pmb{\rangle}_{\beta, L, N}^{\text{(ref)}} ]\nn\\&& + \D^{\text{(ref)}}_{\beta, L, N; \o, \s,\s'}(\underline{k}, \underline{p})\;,
\eea
where:
\bea
\D^{(\text{ref})}_{\beta, L, N; \o, \s,\s'}(\underline{k}, \underline{p}) &:=& \int_{\beta, L} \frac{d\underline{q}}{(2\pi)^2}\, Z^{\text{(ref)}} C_{N; \o}(\underline{q} + \underline{p}, \underline{q}) \pmb{\langle} \hat\psi^{+}_{\underline{p}+\underline{q},\o,\s'}\hat\psi^{-}_{\underline{p}, \o, \s'}; \hat\psi^{-}_{\underline{k}+\underline{p}, \o, \s} \hat \psi^{+}_{\underline{k},\o, \s} \pmb{\rangle}_{\beta, L, N}^{(\text{ref})}\nn\\
C_{N; \o}(\underline{q} + \underline{p}, \underline{q}) &:=& D_{\o}(\underline{q}+\underline{p}) [1 - \chi_{N}(\underline{q}+\underline{p})^{-1}] - D_{\o}(\underline{q})[1 - \chi_{N}(\underline{q})^{-1}]
\eea
with $D_{\o}(\underline{k}) = -ik_{0} + \o v^{(\text{ref})} k_{1}$. Similarly, differentiating with respect to $\hat \alpha_{\underline{p}, \o, \s}$, $\hat A_{-\underline{p}, \o, \s'}$ one finds:
\be
D_{\o}(\underline{p}) {Z^{\text{(ref)}}}^2 \pmb{\langle} \hat n_{\underline{p}, \o, \s}\,; \hat n_{-\underline{p}, \o, \s'} \pmb{\rangle}_{\beta,L,N}^{(\text{ref})} = \widetilde{\D}^{(\text{ref})}_{\beta,L,N; \o, \s, \s'}(\underline{p})\;,
\ee
where:
\be
\widetilde{\D}^{(\text{ref})}_{\beta,L,N; \o, \s, \s'}(\underline{p}) := \int_{\beta, L} \frac{d\underline{q}}{(2\pi)^2}\, {Z^{\text{(ref)}}}^{2}C_{N; \o}(\underline{q} + \underline{p}, \underline{q}) \pmb{\langle} \hat\psi^{+}_{\underline{p}+\underline{q},\o,\s}\hat\psi^{-}_{\underline{p}, \o, \s}\,; \hat n_{-\underline{p}, \o, \s'} \pmb{\rangle}_{\beta, L, N}^{(\text{ref})}
\ee
It turns out that the correction terms $\D^{(\text{ref})}_{\beta, L, N}$, $\widetilde\D^{(\text{ref})}_{\beta, L, N}$ {\it do not vanish} in the limit $N\to\infty$. Instead, they produce {\it anomalies} in the Ward identities.

\begin{prop}\label{eq:anWIref}{\bf (Anomalous Ward identities for the reference model.)}\label{prp:WIan} Let $\pmb{\langle}\cdot \pmb{\rangle}_{\infty}^{\text{(ref)}} := \lim_{\beta,L\to\infty} \pmb{\langle} \cdot \pmb{\rangle}_{\beta, L}^{\text{(ref)}}$. There exists $\bar \l>0$ such that for $|\l^{\text{(ref)}}| < \bar \l$ the following identities hold true:
\bea\label{eq:WIan}
\pmb{\langle} \hat n_{\underline{p}, \o, \s'}\,; \hat \psi^{-}_{\underline{k}+\underline{p}, \o, \s} \hat \psi^{+}_{\underline{k}, \o, \s} \pmb{\rangle}_{\infty}^{(\text{ref})} &=& \frac{a_{\o}(\underline{p}) + \s\s' \tilde{a}_{\o}(\underline{p})}{2 Z^{\text{(ref)}}} \big[ \pmb{\langle} \hat \psi^{-}_{\underline{k}, \o, \s} \hat \psi^{+}_{\underline{k}, \o, \s} \pmb{\rangle}^{(\text{ref})}_{\infty} - \pmb{\langle} \hat \psi^{-}_{\underline{k}+\underline{p}, \o, \s} \hat \psi^{+}_{\underline{k}+\underline{p}, \o, \s} \pmb{\rangle}_{\infty}^{(\text{ref})} \big]\nn\\
D_{\o}(\underline{p}) \pmb{\langle} \hat n_{\underline{p}, \o, \s}\,; \hat n_{-\underline{p},\o, \s'} \pmb{\rangle}_{\infty}^{\text{(ref)}} &=& -\frac{\d_{\s\s'}}{4\pi v^{(\text{ref})}{Z^{(\text{ref})}}^2} D_{-\o}(\underline{p}) \nn\\&&+ \frac{1}{4\pi v^{(\text{ref})}} D_{-\o}(\underline{p}) \o(\underline{p}) \l^{\text{(ref)}} \sum_{\bar \s} \pmb{\langle} \hat n_{\underline{p}, \o, \bar\s}\,; \hat n_{-\underline{p},\o, \s'} \pmb{\rangle}_{\infty}^{\text{(ref)}}
\eea
where $a_{\o}(\underline{p}) := \frac{1}{D_{\o}(\underline{p}) - \tau \hat w(\underline{p}) D_{-\o}(\underline{p})}$ and $\tilde a_{\o}(\underline{p}) := \frac{1}{D_{\o}(\underline{p})}$, with $\t = \frac{\l^{(\text{ref})}}{2 \pi v^{(\text{ref})}}$.
\end{prop}
The first Ward identity is the content of Theorem 3 of \cite{FM}. The proof of the second identity can be reconstructed from the proof of Theorem 3 of \cite{FM}, see also Section 3 of \cite{BFMhub2}. We will omit the details. Notice that, with respect to \cite{FM}, our coupling constant is $-\l$ instead of $\l$.

These Ward identities are called anomalous due to the presence of a nontrivial prefactor in the right-hand side of the first of (\ref{eq:WIan}), and a nonzero right-hand side in the last of (\ref{eq:WIan}); $\t$ is called the {\it chiral anomaly}. It satisfies a nonrenormalization property, as in the Adler-Bardeen theorem \cite{AB}, see \cite{M2} for a rigorous analysis for one-dimensional system. 

Finally, the last equation in (\ref{eq:WIan}) can be used to derive the following exact identities for the correlation functions of $\hat n^{c}_{\underline{p}, \o}$ and $\hat n_{\underline{p}, \o}^{s}$, which will be used in the next section:
\bea\label{eq:ncns}
\widetilde{D}_{\o}(\underline{p}) \pmb{\langle} \hat n^{c}_{\underline{p}, \o}\,; \hat n^{c}_{-\underline{p},\o} \pmb{\rangle}_{\infty}^{\text{(ref)}} &=& -\frac{1}{2\pi v^{(\text{ref})} {Z^{(\text{ref})}}^{2}} \frac{D_{-\omega}(\underline{p})}{1 - \tau \hat w(\underline{p})}\nn\\ D_{\o}(\underline{p}) \pmb{\langle} \hat n^{s}_{\underline{p}, \o}\,; \hat n^{s}_{-\underline{p}, \o} \pmb{\rangle}_{\infty}^{\text{(ref)}} &=& -\frac{D_{-\o}(\underline{p})}{2\pi v^{(\text{ref})} {Z^{(\text{ref})}}^{2}} \;,
\eea
with $\widetilde{D}_{\o}(\underline{p}) := -ip_{0} + \o v^{(\text{ref})}\Big( \frac{1 + \tau \hat w(p)}{1 - \tau \hat w(p)}\Big) p_{1}$.

\section{Proof of Theorem \ref{thm:1}}\label{sec:proof}

We have now all the ingredients to prove Theorem \ref{thm:1}. In Section \ref{sec:edgecond} we will compute the edge transport coefficients, item $i)$ of Theorem \ref{thm:1}. In Section \ref{sec:SC} we will prove the spin-charge separation, item $ii)$ of Theorem \ref{thm:1}. As discussed after Theorem \ref{thm:1}, the bulk-edge correspondence, item $iii)$, follows from the analogous result for noninteracting systems, Eq. (\ref{eq:benonint}), together with the universality of $G$ and of $\s_{21}$. The universality of the Hall conductivity has been proven in \cite{GMP}.

\subsection{The edge transport coefficients}\label{sec:edgecond}

Let us start by discussing the charge transport coefficients. By Proposition \ref{prp:wick} we can rewrite the edge transport coefficients in terms of Euclidean correlation functions. Let $L> a > a'> 0$. We have, setting $p_{0} \equiv \eta$:
\be\label{eq:Dkk'}
G^{\underline{a}}_{\m\n}(\underline{p}) =  \sum_{y_{2} = 0}^{a'}\Big[ \sum_{x_{2}=0}^{a} \pmb{\langle}{\bf T} \hat j_{\m,\underline{p}, x_{2}}\,; \hat j_{\n,-\underline{p}, y_{2}}\pmb{\rangle}_{\infty} + \D_{\m,y_{2}} \d_{\m\n}  \Big](-1)^{\d_{\m,1}}\;,\qquad \m,\n=0,1\;,
\ee
where $\pmb{\langle}{\bf T} \cdot \pmb{\rangle}_{\infty} = \lim_{\beta, L\to \infty} (\beta L)^{-1} \langle {\bf T} \cdot \rangle_{\beta, L}$. The existence of the limits can be proven as in Lemma 2.6 of \cite{BFMhub1}. We will first use Proposition \ref{prp:relref} to express Eq. (\ref{eq:Dkk'}) in terms of the correlation functions of the reference model, that can be computed explicitly thanks Proposition \ref{prp:WIan}, see Eqs. (\ref{eq:ncns}), up to some renormalized coefficients. We will then exploit the lattice Ward identities, Section \ref{prp:WIan}, and the anomalous Ward identity, Proposition \ref{prp:WIan}, to fix the values of the coefficients. By Proposition \ref{prp:relref}, we have:
\be\label{eq:JJ}
\pmb{\langle} {\bf T} \hat j_{\m, \underline{p}, x_{2}}\,; \hat j_{\n, -\underline{p}, y_{2}} \pmb{\rangle}_{\infty} = Z^{(\text{ref})}_{c,\m}(x_{2})Z^{(\text{ref})}_{c,\n}(y_{2}) \pmb{\langle} \hat n^{c}_{\underline{p}, \o}\,; \hat n^{c}_{-\underline{p},\o} \pmb{\rangle}_{\infty}^{(\text{ref})} + R^{\infty}_{\m\n,c}(\underline{p}; x_{2}, y_{2})\;,
\ee
for some exponentially decaying functions $Z^{(\text{ref})}_{\m,c}(x_{2})$, so far unknown, and an error term satisfying the bounds (\ref{eq:bdR}). Plugging Eq. (\ref{eq:JJ}) into Eq. (\ref{eq:Dkk'}) we have:
\be\label{eq:Dkk'2}
G^{\underline{a}}_{\m\n}(\underline{p}) = \sum_{y_{2} = 0}^{a'}\Big[Z^{(\text{ref})}_{c,\n}(y_{2}) \sum_{x_{2} = 0}^{a}Z^{(\text{ref})}_{c,\m}(x_{2})  \pmb{\langle} \hat n^{c}_{\underline{p}, \o}\,; \hat n^{c}_{-\underline{p},\o} \pmb{\rangle}_{\infty}^{(\text{ref})} + A^{a}_{\m\n}(y_{2}) \Big](-1)^{\d_{\m,1}} + \widetilde R_{\m\n,c}^{\underline{a}}(\underline{p})\;,
\ee
where:
\bea\label{eq:JJ2}
A^{a}_{\m\n}(y_{2}) &=& \D_{\m,y_{2}} \d_{\m\n} + \sum_{x_{2} = 0}^{a} R^{\infty}_{\m\n,c}(\underline{0}; x_{2}, y_{2})\\
|\widetilde R_{\m\n,c}^{\underline{a}}(\underline{p})| &\leq& \sum_{x_{2} = 0}^{a}\sum_{y_{2} = 0}^{a'} |R^{\infty}_{\m\n,c}(\underline{p}; x_{2}, y_{2}) - R^{\infty}_{\m\n,c}(\underline{0}; x_{2}, y_{2}) |\;.\nn
\eea
Thanks to the third of Eq. (\ref{eq:bdR}), 
\be
|A^{a}_{\m\n}(y_{2}) - A^{\infty}_{\m\n}(y_{2})| \leq \sum_{x_{2} > a} | R^{\infty}_{\m\n,c}(\underline{0}; x_{2}, y_{2}) | \leq C_{n}/(1 + |y_{2} - a|^{n} )
\ee
for all $n\in \mathbb{N}$. Therefore, $\sum_{y_{2}\leq a'} |A^{a}_{\m\n}(y_{2}) - A^{\infty}_{\m\n}(y_{2})|\leq C_{n} a'/{|a- a'|}^{n} \to 0$ as $a\to \infty$. Moreover, the last of Eq. (\ref{eq:bdR}) implies that $|\widetilde{R}^{\underline{a}}_{\m\n,c}(\underline{p})|\leq Ca' |\underline{p}|^{\theta}$.  The next step is to combine the lattice Ward identity together with the anomalous Ward identity to compute explicitly $A^{\infty}_{\m\n}(x_{2})$. In the $L\to \infty$ limit, the Ward identity reads:
%
%
%
\be\label{eq:WIJJ}
\sum_{\m=0,1}\eta_{\m}(\underline{p}) p_{\m} \Big[ \sum_{x_{2}=0}^{\infty} \pmb{\langle} {\bf T} \hat j_{\m,(\underline{p}, x_{2})}\,; \hat j_{\n,(-\underline{p},y_{2})} \pmb{\rangle}_{\infty} + \d_{\m\n} \D_{\m, y_{2}}\Big] = 0\;.
\ee
Let $Z^{(\text{ref})}_{c,\m} := \sum_{x_{2} = 0}^{\infty} Z^{(\text{ref})}_{c,\m}(x_{2})$; the sum exists thanks to the bound $|Z^{(\text{ref})}_{c,\m}(y_{2})| \leq C_{n} /(1 + |y_{2}|^{n})$ for all $n\in \mathbb{N}$, recall Proposition \ref{prp:relref}. Rewriting the current-current correlation in Eq. (\ref{eq:WIJJ}) in terms of the reference model we get:
\bea\label{eq:Akk'}
&&\sum_{\m=0,1} \eta_{\m}(\underline{p}) p_{\m} \Big( Z^{(\text{ref})}_{c,\m} Z^{(\text{ref})}_{c,\n}(y_{2})\pmb{\langle} \hat n^{c}_{\underline{p}, \o}\,; \hat n^{c}_{-\underline{p},\o} \pmb{\rangle}_{\infty}^{(\text{ref})} + A^{\infty}_{\m\n}(y_{2})\nn\\&&\qquad\qquad + \sum_{x_{2} = 0}^{\infty} [R^{\infty}_{\m\n,c}(\underline{p}; x_{2}, y_{2}) - R^{\infty}_{\m\n,c}(\underline{0}; x_{2}, y_{2})]\Big) = 0\;.
\eea
The density-density correlation functions of the reference model have been explicitly computed in Eq. (\ref{eq:ncns}). We will use this computation together with the identity Eq. (\ref{eq:Akk'}) to determine $A^{\infty}_{\m\n}(y_{2})$. Let us define:
\be\label{eq:vsvc}
v_{s} := v^{(\text{ref})}\;,\qquad v_{c} :=  v^{(\text{ref})} \Big( \frac{1 + \tau}{1 - \tau}\Big)\;.
\ee
As it will be clear from the discussion below, these quantities play the role of {\em spin velocity} $v_{s}$ and the {\em charge velocity} $v_{c}$ for the edge excitations. By the first of Eq. (\ref{eq:ncns}), recalling that $\hat w(\underline{0}) = 1$ and that $|\hat w(\underline{p}) - \hat w(\underline{0})|\leq C|\underline{p}|$, we have:
\be\label{eq:nn}
\pmb{\langle} \hat n^{c}_{\underline{p}, \o}\,; \hat n^{c}_{-\underline{p},\o}  \pmb{\rangle}_{\infty}^{(\text{ref})} = -\frac{1}{2\pi v_{s} {Z^{(\text{ref})}}^2} \frac{1}{1 - \t} \frac{-ip_{0} - \o v_{s} p_{1}}{-ip_{0} + \o v_{c} p_{1}} + R_{\o,c}(\underline{p})\;,
\ee
with $|R_{\o,c}(\underline{p})|\leq C|\underline{p}|^{\theta}$ for some $\theta >0$. Plugging Eq. (\ref{eq:nn}) into Eq. (\ref{eq:Akk'}) we get, for some new error term $|\widehat{R}_{\m\n,c}(\underline{p})|\leq  C|\underline{p}|^{\theta}$:
\be\label{eq:nn2}
\sum_{\m=0,1} \eta_{\m}(\underline{p}) p_{\m} \Big( -\frac{ Z^{(\text{ref})}_{c,\m} Z^{(\text{ref})}_{c,\n}(y_{2})}{2\pi v_{s} {Z^{(\text{ref})}}^2} \frac{1}{1 - \t} \frac{-ip_{0} - \o v_{s} p_{1}}{-ip_{0} + \o v_{c} p_{1}}   + A^{\infty}_{\m\n,\o}(y_{2}) + \widehat{R}_{\m\n,c}(\underline{p}) \Big) = 0\;.
\ee
The crucial remark now is that the error term $\widehat{R}_{\m\n,c}(\underline{p})$ is continuous in $\underline{p}$, and vanishes for $\underline{p}\to \underline{0}$. Let us set $p_{0} = 0$ in Eq. (\ref{eq:nn2}). We get:
\be
\eta_{1}(\underline{p}) p_{1} \Big( \frac{ Z^{(\text{ref})}_{c,1}Z^{(\text{ref})}_{c,\n}(y_{2}) }{2\pi v_{s} {Z^{(\text{ref})}}^2} \frac{1}{1 - \t} \frac{v_{s}}{v_{c}}   + A^{\infty}_{1\n}(y_{2}) + \widehat{R}_{1\n,c}(0,p_{1})\Big) = 0\;;
\ee
therefore, dividing by $p_{1}$ and taking the limit $p_{1}\to 0$:
\be\label{eq:Ak1}
A^{\infty}_{1\n}(y_{2}) = -\frac{Z^{(\text{ref})}_{c,1} Z^{(\text{ref})}_{c,\n}(y_{2})}{2\pi v_{c} {Z^{(\text{ref})}}^2} \frac{1}{1 - \t}\;,\qquad \n=0,1\;.
\ee
In order to compute $A^{\infty}_{0\n}(y_{2})$, we repeat the strategy interchanging the roles of $p_{0}$ and $p_{1}$. We get:
\be\label{eq:Ak0}
A^{\infty}_{0\n}(y_{2}) =  \frac{Z^{(\text{ref})}_{c,0} Z^{(\text{ref})}_{c,\n}(y_{2})}{2\pi v_{s} {Z^{(\text{ref})}}^2} \frac{1}{1-\tau}\;,\qquad \n =0,1\;.
\ee
Therefore, plugging Eqs. (\ref{eq:nn}), (\ref{eq:Ak1}), (\ref{eq:Ak0}) into Eq. (\ref{eq:Dkk'2}), we obtain, after also writing $\sum_{y_{2}\leq a'} Z^{(\text{ref})}_{c,\n}(y_{2}) \sum_{x_{2} \leq a} Z^{(\text{ref})}_{c,\m}(x_{2}) =  Z^{\text{(ref)}}_{c,\n}Z^{\text{(ref)}}_{c,\m} + O(a^{-n}) + O({a'}^{-n})$:
%
%
\bea\label{eq:comp}
G^{\underline{a}}_{1\n}(\underline{p}) &=& \frac{Z_{c,1}^{\text{(ref)}} Z_{c, \n}^{\text{(ref)}} }{\pi v_c Z^{\text{(ref)}2}} \frac{1}{(1 - \tau)^2} \frac{-ip_{0}}{ -ip_{0} + \omega v_{c} p_{1} } + R^{\underline{a}}_{1\n,c}(\underline{p})\nn\\
G^{\underline{a}}_{0\n}(\underline{p}) &=& \frac{Z_{c,0}^{\text{(ref)}} Z_{c,\n}^{\text{(ref)}} }{\pi v_{c} Z^{\text{(ref)}2}} \frac{1}{(1-\tau)^2} \frac{\omega p_{1} v_{c}}{-ip_{0} + \omega v_{c} p_{1}} + R^{\underline{a}}_{0\n,c}(\underline{p})
\eea
for some new error term $R^{\underline{a}}_{\m\n,c}(\underline{p})$, $|R^{\underline{a}}_{\m\n,c}(\underline{p})|\leq C_{n}a'(|\underline{p}|^{\theta} + |{a'} - a|^{-n}) + C_{n}a'^{-n}$. To conclude, we are left with computing the parameters $Z^{\text{(ref)}}_{c,\m},\, Z^{\text{(ref)}}$. This will be done by comparing the vertex Ward identities for the lattice and reference model. Recall the lattice vertex Ward identity, Eq. (\ref{eq:sumWI}):
\be\label{eq:WIvert}
-\sum_{\m=0,1} \eta_{\m}(\underline{p}) p_{\m} \sum_{z_{2}=0}^{\infty} \pmb{\langle} {\bf T} \hat j_{\m, \underline{p}, z_{2}}\,; \hat a^{-}_{\underline{k}+\underline{p}, y_{2}, r'} \hat a^{+}_{\underline{k}, x_{2},r}\pmb{\rangle}_{\infty}  = \pmb{\langle} {\bf T} \hat a^{-}_{\underline{k}, y_{2}, r'} \hat a^{+}_{\underline{k}, x_{2}, r} \pmb{\rangle}_{\infty} -  \langle {\bf T} \hat a^{-}_{\underline{k}+\underline{p}, y_{2}, r'} \hat a^{+}_{\underline{k} + \underline{p}, x_{2}, r} \pmb{\rangle}_{\infty}\;.
\ee
We can use Proposition \ref{prp:relref} to rewrite the correlations appearing in both sides of Eq. (\ref{eq:WIvert}) in terms of those of the reference model. We have, for $\| \underline{k}' \| = \kappa$, $\| \underline{k}' + \underline{p} \|\leq \kappa$ and $\kappa$ small enough, recalling that $r = (\s, \bar r)$:
\bea\label{eq:WIvert2}
&&-\sum_{\m=0,1}\eta_{\m}(\underline{p}) p_{\m} \pmb{\langle} \hat n^{c}_{\underline{p}, \o}\,; \hat \psi^{-}_{\underline{k}'+\underline{p} ,\o,\s'} \hat \psi^{+}_{\underline{k}' ,\o,\s} \pmb{\rangle}_{\infty}^{(\text{ref})}\big( Z^{(\text{ref})}_{c,\m}Q^{(\text{ref})-}_{r}(x_{2}) Q^{(\text{ref})+}_{ r'}(y_{2})  + O(\kappa^{\theta})\Big) \\
&&= \d_{\s\s'}\Big( \pmb{\langle} \hat\psi^{-}_{\underline{k}',\o,\s} \hat\psi^{+}_{\underline{k}',\o,\s} \pmb{\rangle}_{\infty}^{\text{(ref)}} - \pmb{\langle} \hat\psi^{-}_{\underline{k}'+\underline{p},\o,\s}\hat\psi^{+}_{\underline{k}'+\underline{p}, \o,\s} \pmb{\rangle}_{\infty}^{(\text{ref})}\Big) \Big(Q^{(\text{ref})-}_{r}(x_{2}) Q^{(\text{ref})+}_{ r'}(y_{2})  + O(\kappa^{\theta}) \Big)\;.\nn
\eea
Multiplying left-hand side and right-hand side of Eq. (\ref{eq:WIvert2}) by $\overline{Q^{(\text{ref})-}_{r}(x_{2})} \overline{Q^{(\text{ref})+}_{ r'}(y_{2})}$, summing over $x_{2}, \bar r, y_{2}, \bar r'$, and using that $\| Q^{(\text{ref})\pm} \|_{2} = 1 + O(\l) >0$, $\| Q^{(\text{ref})\pm} \|_{1} \leq C$ we get:
\bea\label{eq:WIlatan}
&&-\sum_{\m=0,1} \eta_{\m}(\underline{p}) p_{\m}  \pmb{\langle} \hat n^{c}_{\underline{p}, \o}\,; \hat\psi^{-}_{\underline{k}'+\underline{p},\o,\s} \hat \psi^{+}_{\underline{k}', \o,\s} \pmb{\rangle}_{\infty}^{(\text{ref})}\big( Z^{(\text{ref})}_{c,\m} + O(\kappa^{\theta}) \big)\nn\\
&& \qquad = \Big( \pmb{\langle} \hat\psi^{-}_{\underline{k}',\o,\s} \hat\psi^{+}_{\underline{k}',\o,\s} \pmb{\rangle}_{\infty}^{\text{(ref)}} - \pmb{\langle} \hat\psi^{-}_{\underline{k}'+\underline{p},\o,\s}\hat\psi^{+}_{\underline{k}'+\underline{p}, \o,\s} \pmb{\rangle}_{\infty}^{(\text{ref})}  \Big) \Big( 1 + O(\kappa^{\theta}) \Big)\;,
\eea
At the same time, recall the anomalous Ward identity for the vertex function of the reference model, Eq. (\ref{eq:WIan}); summing it over $\s'$ we get:
\bea\label{eq:WIvert3}
Z^{(\text{ref})}\widetilde{D}_{\o}(\underline{p}) \pmb{\langle} \hat n^{c}_{\underline{p}, \o}\,; \hat \psi^{-}_{\underline{k}+\underline{p}, \o, \s} \hat \psi^{+}_{\underline{k}, \o, \s} \pmb{\rangle}_{\infty}^{(\text{ref})} &=& \frac{1}{1 - \hat w(\underline{p}) \tau}\big[ \pmb{\langle} \hat \psi^{-}_{\underline{k}, \o, \s} \hat \psi^{+}_{\underline{k}, \o, \s} \pmb{\rangle}^{(\text{ref})}_{\infty} - \pmb{\langle} \hat \psi^{-}_{\underline{k}+\underline{p}, \o, \s} \hat \psi^{+}_{\underline{k}+\underline{p}, \o, \s} \pmb{\rangle}_{\infty}^{(\text{ref})} \big]\nn\\
\eea
where we recall that $\widetilde{D}_{\o}(\underline{p}) = -ip_{0} + \o v_{s}\Big( \frac{1 + \tau \hat w(p)}{1 - \tau \hat w(p)}\Big) p_{1}$. Therefore, comparing Eqs. (\ref{eq:WIlatan}), (\ref{eq:WIvert3}), we easily find, for $\kappa\to 0$:
\be\label{eq:Zc}
Z^{(\text{ref})}_{c,0} = Z^{(\text{ref})}(1 - \tau)\;,\qquad Z^{(\text{ref})}_{c,1} = -Z^{(\text{ref})} \o v_{c}(1-\tau) \equiv -Z^{(\text{ref})} \o v_{s} (1 + \tau)\;.
\ee
Plugging these relations into Eq. (\ref{eq:comp}), we finally get:
\bea\label{eq:comp2}
&&G^{\underline{a}}_{00}(\underline{p}) = \frac{1}{\pi v_{c}} \frac{\o v_{c} p_{1}}{-ip_{0} + \o v_{c} p_{1}} + R^{\underline{a}}_{00,c}(\underline{p})\;,\quad G^{\underline{a}}_{01}(\underline{p}) =  -\frac{\o}{\pi}  \frac{\o v_{c} p_{1}}{-ip_{0} + \o v_{c} p_{1}} + R^{\underline{a}}_{01,c}(\underline{p})\;,\nn\\
&&G^{\underline{a}}_{10}(\underline{p}) = -\frac{\o}{\pi} \frac{-i p_{0}}{-ip_{0} + \o v_{c} p_{1}} + R_{10,c}^{\underline{a}}(\underline{p})\;,\quad G^{\underline{a}}_{11}(\underline{p}) = \frac{v_{c}}{\pi} \frac{-i p_{0}}{-ip_{0} + \o v_c p_{1}} + R^{\underline{a}}_{11,c}(\underline{p})\;.
\eea
This proves the desired claim, Eqs. (\ref{eq:edgeint}), for the charge transport coefficients. Let us now consider the spin transport coefficients. After Wick rotation:
\be
G^{\underline{a},s}_{\m\n}(\underline{p}) = \sum_{y_{2} = 0}^{a'}\Big[ \sum_{x_{2}=0}^{a} \pmb{\langle}{\bf T} \hat j^{s}_{\m,\underline{p}, x_{2}}\,; \hat j^{s}_{\n,-\underline{p}, y_{2}}\pmb{\rangle}_{\infty} + \D_{\m,y_{2}} \d_{\m\n}\Big](-1)^{\d_{\m,1}}\;.
\ee
The strategy to compute $G^{\underline{a},s}_{\m\n}(\underline{p})$ will be identical to the one followed for $G^{\underline{a}}_{\m\n}(\underline{p})$. We write:
\be\label{eq:Ds}
G^{\underline{a},s}_{\m\n}(\underline{p}) = \sum_{y_{2} = 0}^{a'}\Big[ Z^{(\text{ref})}_{s,\n}(y_{2}) \sum_{x_{2} = 0}^{a} Z^{(\text{ref})}_{s,\m}(x_{2}) \pmb{\langle} \hat n^{s}_{\underline{p}, \o}\,; \hat n^{s}_{-\underline{p},\o}  \pmb{\rangle}_{\infty}^{(\text{ref})} + \widetilde{A}^a_{\m\n}(y_{2}) \Big](-1)^{\d_{\m,1}} + \widetilde R_{\m\n,s}^{\underline{a}}(\underline{p})\;,
\ee
where, by the second of Eq. (\ref{eq:ncns}):
\be
\pmb{\langle} n^{s}_{\underline{p}, \o}\,; n^{s}_{-\underline{p}, \o} \pmb{\rangle}_{\infty}^{\text{(ref)}} = -\frac{1}{2\pi v_{s} {Z^{(\text{ref})}}^2} \frac{-ip_{0} - \o v_{s} p_{1} }{-ip_{0} + \o v_{s} p_{1}}\;.
\ee
Then, proceeding as in Eqs. (\ref{eq:nn2})--(\ref{eq:Ak0}), this time using the first Ward identity in Eq. (\ref{eq:WIspinlat}), we get:
\be
\tilde A^{\infty}_{1\n}(y_{2}) = -\frac{Z^{(\text{ref})}_{s,1} Z^{(\text{ref})}_{s,\n}(y_{2})}{2\pi v_{s} {Z^{(\text{ref})}}^2}\;,\qquad \tilde A^{\infty}_{0\n}(y_{2}) = \frac{Z^{(\text{ref})}_{s,0} Z^{(\text{ref})}_{s,\n}(y_{2})}{2\pi v_{s} {Z^{(\text{ref})}}^2}\;.
\ee
To compute the $Z^{(\text{ref})}_{s,\m}$ coefficients, we proceed as in Eqs. (\ref{eq:WIvert})--(\ref{eq:Zc}). The first of Eq. (\ref{eq:WIan}) implies:
\bea
D_{\o}(\underline{p}) \pmb{\langle} \hat n^{s}_{\underline{p}, \o}\,; \hat \psi^{-}_{\underline{k}+\underline{p}, \o, \s} \hat \psi^{+}_{\underline{k}, \o, \s} \pmb{\rangle}_{\infty}^{(\text{ref})} &=& \s\big[ \pmb{\langle} \hat \psi^{-}_{\underline{k}, \o, \s} \hat \psi^{+}_{\underline{k}, \o, \s} \pmb{\rangle}^{(\text{ref})}_{\infty} - \pmb{\langle} \hat \psi^{-}_{\underline{k}+\underline{p}, \o, \s} \hat \psi^{+}_{\underline{k}+\underline{p}, \o, \s} \pmb{\rangle}_{\infty}^{(\text{ref})} \big]\;.
\eea
Then, recalling the vertex Ward identity in Eq. (\ref{eq:WIspinlat}), we get:
\be
Z^{(\text{ref})}_{s,0} = Z^{(\text{ref})}\;,\qquad Z^{(\text{ref})}_{s,1} = -Z^{(\text{ref})}\o v_{s}\;. 
\ee
Plugging these last expressions in Eq. (\ref{eq:Ds}) we finally find:
\bea
&&G^{\underline{a},s}_{00}(\underline{p}) = \frac{1}{\pi v_{s}} \frac{\o v_{s} p_{1}}{-ip_{0} + \o v_{s} p_{1}} + R^{\underline{a}}_{00,s}(\underline{p})\;,\quad G^{\underline{a},s}_{01}(\underline{p}) =  -\frac{\o}{\pi}  \frac{\o v_{s} p_{1}}{-ip_{0} + \o v_{s} p_{1}} + R^{\underline{a}}_{01,s}(\underline{p})\;,\nn\\
&&G^{\underline{a},s}_{10}(\underline{p}) = -\frac{\o}{\pi} \frac{-i p_{0}}{-ip_{0} + \o v_{s} p_{1}} + R_{10,s}^{\underline{a}}(\underline{p})\;,\quad G^{\underline{a},s}_{11}(\underline{p}) = \frac{v_{s}}{\pi} \frac{-i p_{0}}{-ip_{0} + \o v_s p_{1}} + R^{\underline{a}}_{11,s}(\underline{p})
\eea
with $|R_{\m\n,s}^{\underline{a}}(\underline{p})|\leq C_{n}a'(|\underline{p}|^{\theta} + |{a'} - a|^{-n}) + C_{n}a'^{-n}$. This proves the claim about the spin transport coefficients, and concludes the proof of item $i)$ of Theorem \ref{thm:1}.
\qed
\medskip

To conclude, the same strategy can of course be repeated to evaluate the edge transport coefficients in the (much simpler) case $\l = 0$. This time, we can assume the presence of an arbitrary number of edge states, as specified by Assumption 1. The only difference is that the current-current correlations have now the form, for $\sharp = c,s$:
\be
\pmb{\langle} {\bf T} \hat j^{\sharp}_{\m, \underline{p}, x_{2}}\,; \hat j^{\sharp}_{\n, -\underline{p}, y_{2}} \pmb{\rangle}^{(0)}_{\infty} = \sum_{e}^{*} Z^{(\text{ref})}_{\sharp,\m,e}(x_{2})Z^{(\text{ref})}_{\sharp,\n,e}(y_{2}) \pmb{\langle} \hat n^{\sharp}_{\underline{p}, \o}\,; \hat n^{\sharp}_{-\underline{p},\o} \pmb{\rangle}_{\infty,e}^{(\text{ref}),(0)} + R^{\infty}_{\m\n,\sharp}(\underline{p}; x_{2}, y_{2})\;,
\ee
where the asterisk restricts the sum to the edge states localized around $x_{2} = 0$, and $\pmb{\langle}\cdot \pmb{\rangle}_{\infty,e}^{(\text{ref}),(0)}$ is a noninteracting reference model with velocity $v^{\text{(ref)}} = |v_{e}|$. By repeating the same argument described above, item $i)$ of Proposition \ref{prop:edgenonint} follows.

\subsection{Spin-charge separation}\label{sec:SC}

The proof of spin-charge separation immediately follows from Proposition \ref{prp:relref}, together with the spin-charge separation for the reference model. One has, \cite{FM, BFMhub2}:
\be\label{eq:SCref}
\langle \psi^{-}_{\underline{x}, \o, \s} \psi^{+}_{\underline{y}, \o, \s'} \rangle^{(\text{ref})}_{\infty} = \frac{\d_{\s\s'}}{Z} \frac{(1 + R_{\o,\s}(\underline{x}, \underline{y}))}{\sqrt{(v_{s} (x_{0} - y_0) + i \o (x_{1} - y_1))(v_{c} (x_{0} - y_0) + i \o (x_{1} - y_1))}}\;,
\ee
where $Z = 1 + O(\l)$, $|R_{\o, \s}(\underline{x},\underline{y})|\leq C\|\underline{x} - \underline{y}\|^{-\theta}$ for some $\theta >0$, and $v_{s}$, $v_{c}$ given by Eq. (\ref{eq:vsvc}). The analogous claim for the lattice model, Eq. (\ref{eq:SC}), follows after plugging Eq. (\ref{eq:SCref}) into the first of Eq. (\ref{eq:latWI}). This concludes the proof of item $ii)$ of Theorem \ref{thm:1}.\qed

\section{Renormalization group analysis}\label{sec:RG}

In the remaining part of the paper, we will discuss the renormalization group analysis leading to the proof of Proposition \ref{prp:relref}. The method can be used to construct all correlation function of both lattice and reference model.

Our goal will be to set up a convergent expansion for the generating functional of the correlation functions, $\mathcal{W}_{\beta, L}(A,\psi)$, uniformly $\beta$ and $L$. The starting point is the outcome of the integration of the bulk and ultraviolet degrees of freedom, Eq. (\ref{eq:RGstart}). The trouble in evaluating the Grassmann integral in Eq. (\ref{eq:RGstart}) is that, due to the absence of a mass gap, the Grassmann field $\psi^{(\text{i.r.})}$ cannot be integrated in a single step. In fact, its propagator decays as $|g^{\text{(i.r.)}}_{e}(\underline{x} - \underline{y})| \leq C\|\underline{x} - \underline{y}\|^{-1}$, a bound which gives rise to apparent {\it infrared divergences} in the fermionic cluster expansion. To solve this problem, we will perform a multiscale analysis, by decomposing the field into a sum of single scale fields. The single scale fields will be integrated in a progressive way, starting from the high energy scales until the lowest energy scales. The covariance and the effective action at a given scale will be defined inductively, via the combination of suitable {\it localization} and {\it renormalization} operations. This will ultimately allow to exploit {\it cancellations} in the naive expansion in a systematic way, and to prove analyticity of the correlation functions for $|\l| < \bar \l$ for some $\bar \l$ independent of $\beta, L$.

\subsection{Quasi-particles}\label{sec:qp}

Let $\psi^{(\leq 0)\pm}_{\underline{x}, e} = e^{\mp i \underline{k}_{F}^{e} x_{1}} \psi^{\text{(i.r.)}\pm}_{\underline{x}, e}$. The field $\psi^{(\leq 0)}_{e}$ is called the {\it quasi-particle field}.
 Let $\underline{k}' = \underline{k} - \underline{k}_{F}^{e}$. Recalling the short-hand notation $\int_{\beta, L} d\underline{k}' \equiv \frac{1}{\beta L}\sum_{\underline{k}'\in \mathbb{D}^{*}_{\beta,L}}$, the propagator of this new field is:
\bea
\int P_{ 0}(d\psi^{(\leq 0)}) \psi^{(\leq 0)-}_{\underline{x}, e} \psi^{(\leq 0)+}_{\underline{y}, e'}
&=& \delta_{ee'} \int_{\beta,L} \frac{d\underline{k}'}{(2\pi)^2}\, e^{-i \underline{k}'\cdot (\underline{x} - \underline{y})} \frac{\chi_{0,e}(\underline{k}')}{-ik_{0} + \e_{e}(k_{F} + k'_{1}) - \m_{e}}\nn\\
&\equiv& \delta_{ee'} g_{e}^{(\leq 0)}(\underline{x} - \underline{y})\;,
\eea
that is $g_{e}^{(\leq 0)}(\underline{x} - \underline{y}) = e^{-ik_{F}^{e}(\underline{x} - \underline{y})} g_{e}^{(\text{i.r.})}(\underline{x} - \underline{y})$. After rewriting the effective action in terms of the quasi-particle fields, we get:
\be\label{eq:leq0}
\mathcal{W}_{\beta,L}(A,\phi) = \mathcal{W}^{( 0)}_{\beta,L}(A,\phi) + \log \int P_{ 0}(d\psi^{(\leq 0)}) e^{V^{( 0)}(\psi^{(\leq 0)}; A, \phi)}
\ee
where $\mathcal{W}^{( 0)}_{\beta,L}(A,\phi) = \mathcal{W}^{(\text{bulk})}_{\beta,L}(A,\phi) + \mathcal{W}^{(\text{u.v.})}_{\beta,L}(A,\phi)$ and
\be\label{eq:V0}
V^{(0)}(\psi; A, \phi) = \sum_{\G} \int_{\beta, L} D \underline{X}  D {\bf Y}  D {\bf Z}\,\psi_{\G}(\underline{X}) \phi_{\G}({\bf Y}) A_{\G}({\bf Z})  W_{\G}^{\text{(0)}}(\underline{X},{\bf Y},{\bf Z})
\ee
with
\be
W_{\G}^{(0)}(\underline{X},{\bf Y},{\bf Z}) = \Big[\prod_{k=1}^{n} e^{i\e_{i}k_{F}^{e_{k}} x_{k,1}}\Big] W_{\G}^{\text{(1d)}}(\underline{X},{\bf Y},{\bf Z})\;.
\ee
Note that, by translation invariance,
\be
\widehat{W}_{\G}^{(0)}(\underline{K}', (\underline{Q}', y_{2}),(\underline{P}, z_{2})) = \delta\Big(\sum_{i} (-1)^{\e_{i}}\underline{k}_{i} + \sum_{j}(-1)^{\kappa_{j}}\underline{q}_{j} + \sum_{s} \underline{p}_{s} \Big) \widehat{W}_{\G}^{(0)}(\underline{K}', (\underline{Q}', y_{2}),(\underline{P}, z_{2}))\nn
\ee
where $\delta(\cdot)$ is the Kronecker delta function. In the following, we shall only write the independent momenta at the argument of the Fourier transforms of the kernels. For instance: $\widehat{W}^{(0)}_{\G}(\underline{k}'_{1}, \underline{k}'_{2}) \equiv \delta(\underline{k}_{1} - \underline{k}_{2})\widehat{W}^{(0)}_{\G}(\underline{k}'_{1})$. The support of the fermionic field $\hat \psi_{e}^{(\leq 0)}$ coincides with the support of the propagator $\hat g^{(\leq 0)}_{e}$. In the following, we shall choose the parameter $\d'_{e}$ appearing in the definition of the cutoff function $\chi_{0, e}(\underline{k}')$, Eq. (\ref{eq:1N}), small enough so that the only nonvanishing contributions to the effective action (\ref{eq:V0}) are those verifying $\sum_{i} (-1)^{\e_{i}} k_{F}^{e_{i}} = 0$. 

For later convenience, we rewrite the effective action as:
\be\label{eq:VGBV}
V^{(0)}(\psi; A, \phi) = -V^{(0)}(\psi) + \G^{(1)}(\psi; A,\phi) + B^{(1)}(\psi; \phi) + V_{R}^{(0)}(\psi; A,\phi)\;,
\ee
where: $V^{(0)}(\psi) = -V^{(0)}(\psi; 0, 0)$; $\hat n^{c}_{\underline{p}, \bar e} = \hat n_{\underline{p}, (\bar e, \uparrow)} + \hat n_{\underline{p}, (\bar e, \downarrow)}$, $\hat n^{s}_{\underline{p}, \bar e} = \hat n_{\underline{p}, (\bar e, \uparrow)} - \hat n_{\underline{p}, (\bar e, \downarrow)}$;
\bea
\G^{(1)}(\psi; A, \phi) &=& \sum_{\bar e,\sharp}\sum_{\m} \sum_{z_{2}} \int_{\beta, L} \frac{d\underline{p}}{(2\pi)^2}\,\hat A_{\m, \underline{p}, z_{2}}^{\sharp} \hat n^{\sharp}_{\underline{p},\bar e} Z_{1, \sharp, \m, e}(z_{2}) \nn\\
B^{(1)}(\psi; \phi) &=& \sum_{r, e}\sum_{x_{2}} \int_{\beta, L} \frac{d\underline{k}'}{(2\pi)^2}\, \Big[\hat \phi^{+}_{\underline{k}'+\underline{k}_{F}^{e}, x_{2}, r} Q^{+}_{1,r,e}(\underline{k}'_{1}, x_{2}) \hat\psi^{-}_{\underline{k'},e} + \hat\psi^{+}_{\underline{k}',e} Q^{-}_{1,r,e}(\underline{k}', x_{2})\hat\phi^{-}_{\underline{k}'+\underline{k}_{F}^{e},x_{2}, r} \Big]\;. \nn
\eea
Eq. (\ref{eq:edgeto1d}), together with the fast decay of the edge kernels, easily implies the estimates: 
\be\label{eq:ZQ}
|Z_{1,\sharp,\m,e}(z_{2})|\leq \frac{C_{n}}{1 + |z_{2}|_{e}^{n}}\;,\qquad |Q^{\pm}_{1,e,r}(\underline{k}',x_{2})|\leq \frac{C_{n}}{1 + |x_{2}|_{e}^{n}}\;,
\ee
with $|\cdot|_{e} = |\cdot|, |\cdot - L|$, depending on whether the edge state labeled by $e$ satisfies the first or the second bound in Eq. (\ref{eq:dec1}). In particular, $Q^{\pm}_{1,r,e}(\underline{k}',x_{2}) = \xi^{e}_{x_{2}}(k_{1}; r) + \zeta^e_{x_{2}}(\underline{k}'; r)$, with $\| \zeta^e(\underline{k}') \|_{p} = O(\l)$ for all $p\geq 1$. In the absence of interactions, $\l=0$, these functions can be computed explicitly. For simplicity, let us suppose that only nearest neighbour hoppings are present: $H(\vec x, \vec y) = 0$ for $\|\vec x - \vec y\| > 1$. Recalling that $H_{rr'}(\vec x, \vec y) \equiv H_{rr'}(x_{1} - y_{1}; x_{2}, y_{2})$:
\bea\label{eq:Q1}
&&Q^{(0)+}_{1,r,e}(\underline{k}', x_{2}) = \xi^e_{x_{2}}(k'_{1} + k_{F}^{e}; r)\;,\quad Q^{(0)-}_{1,r,e} = \overline{Q^{(0)+}_{1,r,e}}\;, \quad Z^{(0)}_{1,\sharp,0, e}(z_{2}) = \sum_{r} |\xi^{e}_{z_{2}}(k_{F}^{e}; r) |^2\;,\\&& Z^{(0)}_{1,\sharp,1, e}(z_{2}) = \sum_{r,r'} \overline{\xi^{e}_{z_{2}}(k_{F}^{e}; r)} \big(ie^{-ik_{F}^{e}} H_{rr'}(-1; z_{2}, z_{2}) - ie^{ik_{F}^{e}} H_{rr'}(1; z_{2}, z_{2})\big) \xi^{e}_{z_{2}}(k_{F}^{e}; r')\;.\nn
\eea
The coefficient $Z^{(0)}_{1,\sharp,2, e}(z_{2})$ can be computed in a similar way. Notice that, as $L\to \infty$, we get $\sum_{z_{2} = 0}^{\infty} Z^{(0)}_{1,\sharp,1,e}(z_{2}) = -v_{e}$. In fact:
\bea
\sum_{z_{2} = 0}^{\infty} Z^{(0)}_{1,\sharp,1,e}(z_{2}) &=& \sum_{z_{2}=0}^{\infty} \sum_{r,r'} \overline{\xi^{e}_{z_{2}}(k_{F}^{e}; r)} \big( ie^{-ik_{F}^{e}} H^{\s}_{rr'}(-1; z_{2}, z_{2}) - ie^{ik_{F}^{e}} H^{\s}_{rr'}(1; z_{2}, z_{2}) \big) \xi^{e}_{z_{2}}(k_{F}^{e};  r')\nn\\
&\equiv& -\sum_{z_{2}=0}^{\infty} \sum_{r,r'} \overline{\xi^{e}_{z_{2}}(k_{F}^{e}; r)} \partial_{k_{1}} \hat H^{\s}_{rr'}(k_{F}^{e}; z_{2}, z_{2}) \xi^{e}_{z_{2}}(k_{F}^{e};  r')\nn\\
&=& -\partial_{k_{1}} \langle \xi^{e}(k_{1}), \hat H^{\s}(k_{1}) \xi^{e}(k_{1})\rangle|_{k_{1} = k_{F}^{e}} \equiv -v_{e}\;,
\eea
where in the last step we used that, since $H^{\s}_{rr'}(\vec x, \vec y) = 0$ whenever $\| \vec x - \vec y \|>1$, $\partial_{k_{1}} \hat H^{\s}(k_{1}; x_{2}, y_{2}) = \sum_{z_{1}\neq 0} e^{ik_{1}z_{1}} iz_{1} H^{\s}(z_{1}; x_{2}, y_{2}) = \d_{x_{2},y_{2}} \sum_{z_{1}\neq 0} e^{ik_{1}z_{1}} iz_{1} H^{\s}(z_{1}; x_{2}, x_{2})$. Finally, the effective potential $V^{(0)}(\psi)$ has the form:
\be\label{eq:V00}
V^{(0)}(\psi) =  \int_{\beta,L} d\underline{x}\, \Big[ \sum_{e} \psi^{+}_{\underline{x}, e} \psi^{-}_{\underline{x},e}  n_{0,e} + \sum_{\underline{e}}\psi^{+}_{\underline{x},e_{1}}\psi^{-}_{\underline{x}, e_{2}}\psi^{+}_{\underline{x}, e_{3}}\psi^{+}_{\underline{x}, e_{4}} u_{0,\underline{e}}\Big] + \text{remainder}
\ee
where $n_{0,e} = \n_{e} + O(\l)$ and
\be\label{eq:lambda0}
u_{0,\underline{e}} = \l \d\big(\sum_{i=1}^{4} (-1)^{i+1} k_{F}^{e_{i}}\big)\sum_{\substack{x_{2}, y_{2} \\ r,r'}} \hat w_{rr'}(0; x_{2}, y_{2}) \overline{\xi^{e_{1}}_{x_{2}}(k_{F}^{e_{1}}; r)} \xi^{e_{2}}_{x_{2}}(k_{F}^{e_{2}}; r) \overline{\xi^{e_{3}}_{y_{2}}(k_{F}^{e_{3}};  r')} \xi^{e_{4}}_{y_{2}}(k_{F}^{e_{4}};  r') + O(\l^2)\;,
\ee
while the remainder term is {\it irrelevant} in the RG sense, as discussed in the next sections.

\subsection{Multiscale analysis}\label{sec:multi} In order to compute the functional integral in the right-hand side of Eq. (\ref{eq:leq0}), we will proceed in an inductive way. Let $h \in \mathbb{Z}_{-}$, $0\geq h \geq h_{\beta} := \min_{e}\lfloor \log_{2} (\pi/\d'_{e}\beta) \rfloor$. Suppose that the generating functional can be written as:
\be\label{eq:leqh}
\mathcal{W}_{\beta,L}(A,\phi) = \mathcal{W}^{( h)}_{\beta,L}(A,\phi) + \log \int P_{ h}(d\psi^{(\leq h)}) e^{V^{( h)}(\sqrt{Z_{h}}\psi^{(\leq h)}; A, \phi)}\;,
\ee
where: 
\begin{enumerate}
\item the notation $\sqrt{Z_{h}}\psi^{(\leq h)}$ means that every field $\psi^{(\leq h)\pm}_{\underline{x},e}$ appearing in the effective action $V^{(h)}$ is multiplied by a factor $\sqrt{Z_{h,e}}$.

\item The effective interaction and the generating functional on scale $h$ have the form:
\bea\label{eq:Vh}
V^{(h)}(\psi; A, \phi) &=& \sum_{\G} \int_{\beta, L} D \underline{X}  D {\bf Y}  D {\bf Z}\,\psi_{\G}(\underline{X}) \phi_{\G}({\bf Y}) A_{\G}({\bf Z})  W_{\G}^{(h)}(\underline{X},{\bf Y},{\bf Z}) \nn\\\mathcal{W}^{( h)}_{\beta,L}(A,\phi) &=& \sum_{\G:\, \G_{\psi} = \emptyset} \int_{\beta,L} D {\bf Y}  D {\bf Z}\, \phi_{\G}({\bf Y}) A_{\G}({\bf Z})  W_{\G}^{(h)}({\bf Y},{\bf Z})
\eea
for some kernels analytic in $\l < |\bar \l|$. The effective interaction can be further rewritten as:
\be\label{eq:TQ}
V^{(h)}(\psi; A, \phi) = -V^{(h)}(\psi) + \G^{(h+1)}(\psi; A) +  B^{(h+1)}(\psi; \phi) + V_{R}^{(h)}(\psi; A,\phi)\;,
\ee
where: $V^{(h)}(\psi) = -V^{(h)}(\psi; 0, 0)$,
\bea
&&\G^{(h+1)}(\psi; A) = \sum_{\bar e,\m,\sharp} \sum_{z_{2}} \int_{\beta,L}  \frac{d\underline{p}}{(2\pi)^2} \, \hat A^{\sharp}_{\m, \underline{p}, z_{2}} \hat n^{\sharp}_{\underline{p}, \bar e} Z_{h+1,\sharp,\m, e}(z_{2})\nn\\
&&B^{(h+1)}(\sqrt{Z_{h}}\psi; \phi) =\nn\\&& \sum_{r, e}\sum_{x_{2}} \int_{\beta,L} \frac{d\underline{k}'}{(2\pi)^2}\, \Big[\hat \phi^{+}_{\underline{k}'+\underline{k}_{F}^{e}, x_{2}, r} Q^{+}_{h+1,r,e}(\underline{k}'_{1}, x_{2}) \hat\psi^{-}_{\underline{k'},e} + \hat\psi^{+}_{\underline{k}',e} Q^{-}_{h+1,r,e}(\underline{k}', x_{2})\hat\phi^{-}_{\underline{k}'+\underline{k}_{F}^{e},x_{2}, r} \Big] \nn
\eea
where $Q^{\pm}_{h+1}$, $Z_{h+1,\sharp,\m, e}$ are analytic functions of $\l$, to be defined inductively and for a suitable $V^{(h)}_{R}$, to be defined inductively as well.
\item The Grassmann Gaussian measure $P_{ h}(d\psi^{(\leq h)})$ has covariance given by:
\bea
\int P_{ h}(d\psi) \psi^{(\leq h)-}_{\underline{x}, e} \psi^{(\leq h)+}_{\underline{y}, e'} &=& \delta_{ee'} \int_{\beta,L} \frac{d\underline{k}'}{(2\pi)^2}\, e^{-i \underline{k}'\cdot (\underline{x} - \underline{y})} \hat g^{(\leq h)}_{e}(\underline{k}')\nn\\
\hat g^{(\leq h)}_{e}(\underline{k}') &=& \frac{1}{Z_{h,e}} \frac{\chi_{h,e}(\underline{k}')}{-ik_{0} + v_{h,e} k'_{1}}(1 + r_{h,e}(k'_{1}))\;,
\eea
where $\chi_{h,e}(k'_{1}) \equiv \chi\big((2^{-h}/\d_{e}')\sqrt{k_{0}^{2} + v_{h,e}^{2} {{k}_{1}^{'2}}}\big)$, and $\hat g^{(\leq h)}_{e}(\underline{k}')$ is the {\it renormalized propagator} on scale $h$; the quantities $r_{h,e}$, $Z_{h,e}$, $v_{h,e}$ are analytic functions of $\l$ for $|\l| < \bar \l$ such that:
\be\label{eq:rcc}
|r_{h,e}(k'_{1})|\leq C|k'_{1}|^{\theta}\quad \text{with $\theta >0$,}\qquad  |v_{h,e} - v_{e}|\leq C|\l|\;,\qquad \Big| \frac{Z_{h,e}}{Z_{h-1,e}} \Big| \leq e^{c |\l|}\;.
\ee
\end{enumerate}
The parameter $v_{h,e}$ is the {\it effective Fermi velocity} while $Z_{h,e}$ is the {\it wave function renormalization}. The inductive assumption is trivially true for $h=0$, see Section \ref{sec:qp}. We claim that it is true for $h$ replaced by $h-1$. To prove this, we proceed as follows.

\subsubsection{Localization and renormalization: $A= 0$, $\phi =0$}\label{sec:Lnoext} In this section we will discuss the {\it localization} and {\it renormalization} procedure, that will allow us to integrate the single scale field. We will start by discussing the case  $A=0$, $\phi = 0$. In order to inductively prove Eq. (\ref{eq:leqh}), we split $V^{(h)}(\psi)$ as $\mathcal{L} V^{(h)}(\psi) + \mathcal{R} V^{(h)}(\psi)$, where $\mathcal{R} = 1-\mathcal{L}$ and $\mathcal{L}$, the {\it localization operator}, is a linear operator on monomials  of the Grassmann fields, defined in the following way. Let:
\be
V_{n}^{(h)}(\psi) =  \sum_{\G:\, \G_{A} = \G_{\phi} = \emptyset,\, |\G_{\psi}| = n} \int_{\beta,L} D \underline{X}\,\psi_{\G}(\underline{X}) W_{\G}^{(h)}(\underline{X})\;.
\ee
Then, we define:
\bea\label{eq:locdef}
\mathcal{L} V^{(h)}_{2}(\psi) &=&  \sum_{\underline{e}}\int_{\beta,L} d\underline{x}_{1} d\underline{x}_{2}\,  \psi^{+}_{\underline{x}_{1}, e_{1}} W^{(h)}_{2, 0, 0 ; e_{1}, e_{2}}(\underline{x}_{1}, \underline{x}_{2})(\psi^{-}_{\underline{x}_{1},e_{2}} + (\underline{x}_{2} - \underline{x}_{1})\cdot {\underline{\partial}}\psi^{-}_{\underline{x}_{1},e_{2}}) \;,\nn\\
\mathcal{L} V^{(h)}_{4}(\psi) &=& \sum_{\underline{e}} \int_{\beta,L} d\underline{x}_{1} d\underline{x}_{2} d\underline{x}_{3} d\underline{x}_{4}\, \psi^{+}_{\underline{x}_{1},e_{1}}\psi^{-}_{\underline{x}_{1},e_{2}}\psi^{+}_{\underline{x}_{1},e_{3}}\psi^{-}_{\underline{x}_{1}, e_{4}} W^{(h)}_{4, 0,0; \underline{e}}(\underline{x}_{1}, \underline{x}_{2}, \underline{x}_{3}, \underline{x}_{4})\;,\nn\\
\mathcal{L} V_{n}^{(h)}(\psi) &=& 0\;,\qquad \forall n>4\;,
\eea
where $\partial_{1}$ is the symmetrized discrete derivative, defined as $\partial_{1} f(x_{0}, x_{1}) = (1/2) ( f(x_{0}, x_{1} + 1) - f(x_{0}, x_{1} - 1))$. That is, the {\it renormalization operator} $\mathcal{R} = 1-\mathcal{L}$ acts on a given Grassmann monomial as:
\bea\label{eq:R}
&&\mathcal{R} \psi^{+}_{\underline{x}_{1}, e_{1}} \psi^{-}_{\underline{x}_{2}, e_{2}} =  \psi^{+}_{\underline{x}_{1}, e_{1}} \psi^{-}_{\underline{x}_{2}, e_{2}} - \psi^{+}_{\underline{x}_{1}, e_{1}}(\psi^{-}_{\underline{x}_{1},e_{2}} + (\underline{x}_{2} - \underline{x}_{1})\cdot \underline{\partial}\psi^{-}_{\underline{x}_{1},e_{2}}) \nn\\
&&\mathcal{R} \psi^{+}_{\underline{x}_{1},e_{1}}\psi^{-}_{\underline{x}_{2},e_{2}}\psi^{+}_{\underline{x}_{3},e_{3}}\psi^{-}_{\underline{x}_{4}, e_{4}} = \psi^{+}_{\underline{x}_{1},e_{1}}\psi^{-}_{\underline{x}_{2},e_{2}}\psi^{+}_{\underline{x}_{3},e_{3}}\psi^{-}_{\underline{x}_{4}, e_{4}} - \psi^{+}_{\underline{x}_{1},e_{1}}\psi^{-}_{\underline{x}_{1},e_{2}}\psi^{+}_{\underline{x}_{1},e_{3}}\psi^{-}_{\underline{x}_{1}, e_{4}} 
\eea
Will abbreviate ``$n, 0, 0$'' at the subscript of the kernels by just ``$n$''. In the $\beta,L\to\infty$ limit, one can think the operator $\mathcal{L}$ as acting directly on the Fourier transforms as follows:
\bea\label{eq:locF}
\mathcal{L}\widehat W^{(h)}_{2;\underline{e}}(\underline{k}') &=& \widehat{W}^{(h)}_{2;e_{1}, e_{2}}(\underline{0}) + \underline{k}'\cdot \partial_{\underline{k}'} \widehat{W}^{(h)}_{2;e_{1}, e_{2}}(\underline{0})\nn\\
\mathcal{L}\widehat W^{(h)}_{4; \underline{e}} (\underline{k}'_{1}, \underline{k}'_{2}, \underline{k}'_{3}) &=& \widehat W^{(h)}_{4; \underline{e}} (\underline{0}, \underline{0}, \underline{0}) \nn\\
\mathcal{L} \widehat{W}^{(h)}_{n;\underline{e}}(\underline{k'}_{1}, \ldots, \underline{k}'_{n-1}) &=& 0\qquad \forall n>4.\;
\eea
Notice that, by momentum conservation, the kernels are automatically zero unless $\sum_{i}(-1)^{\e_{i}} k_{F}^{e_{i}} = 0$. It is convenient to define:
\bea\label{eq:L}
&&\widehat{W}^{(h)}_{2; e,e}(\underline{0}) = 2^{h} n_{h,e}\;,\qquad \partial_{k_{0}} \widehat{W}^{(h)}_{2; e,e}(\underline{0}) =  -i z_{0,h,e}\nn\\
&&\partial_{k'_{1}}  \widehat{W}^{(h)}_{2; e,e}(\underline{0}) =  z_{1,h,e}\;,\qquad \widehat W^{(h)}_{4; \underline{e}} (\underline{0}, \underline{0}, \underline{0}) = u_{h, \underline{e}}\;.
\eea
We assume inductively that:
\be\label{eq:asspt}
n_{h,e} \in \mathbb{R},\quad z_{i,h,e} \in \mathbb{R},\quad u_{h,\underline{e}} \in \mathbb{R}\;,\quad |n_{h,e}|\leq C|\l|, \quad |z_{i,h,e}|\leq C|\l|,\quad  |u_{\underline{e}}|\leq C|\l|\;.
\ee
The fact that the coefficients are real is a straightforward consequence of the fact $S_{\beta,L}(\Psi)$, defined in Eq. (\ref{eq:action}), is invariant under the following transformation, which has Jacobian equal to $1$:
\be\label{eq:sym}
\Psi^{+}_{\underline{k}, x_{2}, r} \to -\Psi^{-}_{\underline{k}, x_{2}, r}\;,\qquad \Psi^{-}_{\underline{k}, x_{2}, r} \to \Psi^{+}_{\underline{k}, x_{2}, r}\;,\qquad c\to \overline{c}\;,\qquad k_{0}\to -k_{0}\;,
\ee
where $c$ denotes a generic constant appearing in $S_{\beta,L}(\Psi)$. This symmetry property can be readily checked using that $\overline{H_{rr'}(k_{1}; x_{2}, y_{2})} = H_{r'r}(k_{1}; y_{2}, x_{2})$ (selfadjointness). Instead, the bounds in Eq. (\ref{eq:asspt}) are harder to prove; they will be proven in Section \ref{sec:rcc}, for the case of models exhibiting single-channel edge modes, recall Assumption 2. We now reabsorb part of $\mathcal{L} V_{2}^{(h)}(\sqrt{Z_{h}}\psi)$ in a redefinition of the Gaussian integration. We have:
\be\label{eq:tildeV2}
P_{h}(d \psi^{(\leq h)}) e^{-V^{(h)}(\sqrt{Z_{h}}\psi^{(\leq h)})} = e^{\beta L t_{h}}\widetilde{P}_{h}(d\psi^{(\leq h)})e^{-\widetilde{V}^{(h)}(\sqrt{Z_{h}}\psi^{(\leq h)})}
\ee
where: $\widetilde{P}_{h}(d\psi^{(\leq h)})$ is a Gaussian integration with covariance:
%
%
\bea
\int \widetilde{P}_{ h}(d\psi^{(\leq h)}) \psi^{(\leq h)-}_{\underline{x}, e} \psi^{(\leq h)+}_{\underline{y}, e'} &=& \delta_{ee'} \int_{\beta, L} \frac{d\underline{k}'}{(2\pi)^2}\, e^{-i \underline{k}'\cdot (\underline{x} - \underline{y})} \tilde{g}^{(\leq h)}_{e}(\underline{k}')\nn\\
\tilde{g}^{(\leq h)}_{e}(\underline{k}') &=& \frac{1}{\widetilde{Z}_{h-1,e}(\underline{k}')} \frac{\chi_{h,e}(\underline{k}')}{-ik_{0} + \widetilde{v}_{h-1,e}(\underline{k}') k'_{1}}(1 + r_{h-1,e}(\underline{k}'))\;,
\eea
with $|r_{h-1,e}(\underline{k}')|\leq C|k'_{1}|^{\theta}$ for $\theta>0$, and new renormalized parameters
\bea
\widetilde{Z}_{h-1,e}(\underline{k}') &=& \widetilde{Z}_{h,e}(\underline{k}') + Z_{h,e} z_{0,h,e} \chi_{h,e}(\underline{k}')\nn\\
\widetilde{Z}_{h-1,e}(\underline{k}') \widetilde{v}_{h-1,e}(\underline{k}') &=& \widetilde{Z}_{h,e}(\underline{k}') \widetilde{v}_{h,e}(\underline{k}') + Z_{h,e}z_{1,h,e} \chi_{h,e}(\underline{k}')\;;
\eea
the new effective action is:
\be\label{eq:tildeV}
\widetilde{V}^{(h)}(\psi) = \mathcal{L} V^{(h)}_{4}(\psi) + \sum_{e} 2^{h} n_{h, e} \int_{\beta, L} d\underline{x}\, \psi^{+}_{\underline{x}, e}\psi^{-}_{\underline{x}, e} +\mathcal{R}V^{(h)}(\psi)\;;
\ee
the constant $t_{h}$ takes into account the change of normalization. Defining $Z_{h-1,e} = \widetilde{Z}_{h-1,e}(\underline{0})$, we {\it rescale} the fermionic fields,
\be\label{eq:hatV}
\widetilde{V}^{(h)}(\sqrt{Z_{h}} \psi) =: \widehat V^{(h)}(\sqrt{Z_{h-1}}\psi)\;;
\ee
therefore, the local part of the new effective action becomes:
\bea
&&\mathcal{L}\widehat{V}^{(h)}(\psi) \nn\\
&& = \int_{\beta,L} d\underline{x}\, \Big[ \sum_{e}\psi^{+}_{\underline{x}, e} \psi^{-}_{\underline{x},e} 2^{h}\frac{Z_{h,e}}{Z_{h-1,e}}n_{h,e}\nn\\&&\quad + \sum_{\underline{e}}\psi^{+}_{\underline{x},e_{1}}\psi^{-}_{\underline{x}, e_{2}}\psi^{+}_{\underline{x}, e_{3}}\psi^{+}_{\underline{x}, e_{4}} \frac{ \sqrt{ Z_{h,e_{1}} Z_{h,e_{2}} Z_{h,e_{3}} Z_{h,e_{4}}}} {\sqrt{Z_{h-1,e_{1}} Z_{h-1,e_{2}} Z_{h-1,e_{3}}Z_{h-1,e_{4}}}}u_{h,\underline{e}}\Big]\nn\\
&&\equiv \int_{\beta,L} d\underline{x}\, \Big[ \sum_{e} \psi^{+}_{\underline{x}, e} \psi^{-}_{\underline{x},e} 2^{h} \n_{h,e} + \sum_{\underline{e}}\psi^{+}_{\underline{x},e_{1}}\psi^{-}_{\underline{x}, e_{2}}\psi^{+}_{\underline{x}, e_{3}}\psi^{+}_{\underline{x}, e_{4}} \l_{h,\underline{e}}\Big]\;,
\eea
where we defined:
\be
\n_{h,e} := \frac{Z_{h,e}}{Z_{h-1,e}} n_{h,e}\;,\qquad \l_{h,\underline{e}} := \frac{ \sqrt{ Z_{h,e_{1}} Z_{h,e_{2}} Z_{h,e_{3}} Z_{h,e_{4}}}} {\sqrt{Z_{h-1,e_{1}} Z_{h-1,e_{2}} Z_{h-1,e_{3}}Z_{h-1,e_{4}}}}u_{h,\underline{e}}\;. 
\ee
This concludes the discussion of the localization operation in the absence of external fields. 

\subsubsection{Localization and renormalization: $A\neq 0$, $\phi\neq 0$} Let us now suppose that $A\neq 0$, $\phi \neq 0$. Let:
\be\label{eq:Vh0}
V^{(h)}_{n,r,m}(\psi; A, \phi) = \sum_{\G:\, |\G_{\psi}| = n,\, |\G_{\phi}| = r,\, |\G_{A}| = m} \int_{\beta,L} D \underline{X}  D {\bf Y}  D {\bf Z}\,\psi_{\G}(\underline{X}) \phi_{\G}({\bf Y}) A_{\G}({\bf Z})  W_{\G}^{(h)}(\underline{X},{\bf Y},{\bf Z})
\ee
In addition to Eq. (\ref{eq:locdef}), we define:
\bea\label{eq:locA}
\mathcal{L} V^{(h)}_{R; 2,0,1}(\psi; A,\phi) &=& \sum_{e,\m,\sharp} \int_{\beta,L} d\underline{x}_{1} d\underline{x}_{2} d\zz \, \psi^{+}_{\underline{z}, e} \psi^{-}_{\underline{z}, e}A_{\zz,\m}^{\sharp} W_{2,0,1;e,e,\m,\sharp}^{(h)}(\underline{x}_{1}, \underline{x}_{2}, \zz)\;,\nn\\
\mathcal{L} V^{(h)}_{R; 1,1,0} (\psi; A, \phi) &=& \sum_{e,r} \int_{\beta,L} d \yy d\underline{x}\, \big[ \phi^{+}_{\yy, r} W^{(h)+}_{1,1,0; r, e}(\underline{x}, \yy) \psi^{-}_{\underline{x}, e} + \psi^{+}_{\underline{x}, e} W^{(h)-}_{1,1,0; r, e}(\underline{x}, \yy) \phi^{-}_{\yy, r} \big]\;, \nn\\
\mathcal{L} V^{(h)}_{R; n,r, m}(\psi; A,\phi) &=& 0\qquad \text{otherwise.}
\eea
By spin symmetry $W_{2,0,1;(\bar e, \uparrow),(\bar e, \uparrow)\m,c}^{(h)} = W_{2,0,1;(\bar e, \downarrow),(\bar e, \downarrow)\m,c}^{(h)}$, $W_{2,0,1;(\bar e, \uparrow),(\bar e, \uparrow)\m,s}^{(h)} = -W_{2,0,1;(\bar e, \downarrow),(\bar e, \downarrow)\m,s}^{(h)}$. Hence:
\bea
\sum_{e,\m}  \psi^{+}_{\underline{z}, e} \psi^{-}_{\underline{z}, e}A_{\zz,\m}^{c} W_{2,0,1;e,e,\m,c}^{(h)}(\underline{x}_{1}, \underline{x}_{2}, \zz) &=& \sum_{\bar e,\m} n^{c}_{\underline{z}, \bar e} A_{\zz,\m}^{c} W_{2,0,1;(\bar e,\uparrow),(\bar e,\uparrow),\m,c}^{(h)}(\underline{x}_{1}, \underline{x}_{2}, \zz)\nn\\
\sum_{e,\m}  \psi^{+}_{\underline{z}, e} \psi^{-}_{\underline{z}, e}A_{\zz,\m}^{s} W_{2,0,1;e,e,\m,s}^{(h)}(\underline{x}_{1}, \underline{x}_{2}, \zz) &=& \sum_{\bar e,\m} n^{s}_{\underline{z}, \bar e} A_{\zz,\m}^{c} W_{2,0,1;(\bar e,\uparrow),(\bar e, \uparrow),\m,s}^{(h)}(\underline{x}_{1}, \underline{x}_{2}, \zz)\;.
\eea
Proceeding as in Eqs. (\ref{eq:tildeV2})--(\ref{eq:tildeV}), we have:
\be
\int P_{h}(d\psi^{(\leq h)}) e^{V^{( h)}(\sqrt{Z_{h}}\psi^{(\leq h)}; A, \phi)} = e^{\beta L t_{h}}\int \widetilde{P}_{ h}(d\psi^{(\leq h)}) e^{\widetilde{V}^{( h)}(\sqrt{Z_{h}}\psi^{(\leq h)}; A, \phi)}
\ee
where: $\widetilde{P}_{h}$ has been defined in Eq. (\ref{eq:tildeV2}), and the effective potential $\widetilde{V}^{(h)}$ is:
\bea
\widetilde{V}^{(h)}(\psi; A,\phi) &=& -\widetilde{V}^{(h)}(\psi) + \widetilde{\G}^{(h+1)}(\psi; A) + \widetilde{B}^{(h+1)}(\psi; \phi) + \mathcal{R} V_{R}^{(h)}(\psi; A,\phi)\nn\\
\widetilde{\G}^{(h+1)}(\psi; A) &=& \G^{(h+1)}(\psi; A) + \mathcal{L} V^{(h)}_{R; 2,0,1}(\psi; A,\phi)\nn\\
\widetilde{B}^{(h+1)}(\psi; \phi) &=& B^{(h+1)}(\psi; \phi) + \mathcal{L} V^{(h)}_{R; 1,1,0}(\psi; A, \phi)\;,
\eea
with $\widetilde{V}^{(h)}$ given by Eq. (\ref{eq:tildeV}). We finally rescale the fermionic field:
\bea
\widetilde{V}^{(h)}(\sqrt{Z_{h}} \psi; A,\phi) &=:& \widehat V^{(h)}(\sqrt{Z_{h-1}}\psi; A,\phi)\nn\\
&\equiv& -\widehat{V}^{(h)}(\sqrt{Z_{h-1}}\psi) + \G^{(h)}(\sqrt{Z_{h-1}}\psi; A)\nn\\&& + B^{(h)}(\sqrt{Z_{h-1}}\psi; \phi) + \mathcal{R} \widehat{V}_{R}^{(h)}(\sqrt{Z_{h-1}}\psi; A,\phi)
\eea
where:
\bea
&&\G^{(h)}(\psi; A, \phi) = \sum_{\m,\bar e,\sharp}\sum_{z_{2}}\int_{\beta,L}  \frac{d\underline{p}}{(2\pi)^2}\, \hat A^{\sharp}_{\m,\underline{p},z_{2}} \hat n^{\sharp}_{\underline{p}, e} Z_{h,\sharp,\m, e}(z_{2}) \\
&&B^{(h)}(\sqrt{Z_{h-1}}\psi; \phi) =\nn\\&& \sum_{r, e}\sum_{x_{2}} \int_{\beta,L} \frac{d\underline{k}'}{(2\pi)^2}\, \Big[\hat \phi^{+}_{\underline{k}'+\underline{k}_{F}^{e}, x_{2}, r}  Q^{+}_{h,r,e}(\underline{k}'_{1}, x_{2}) \hat\psi^{-}_{\underline{k'},e} + \hat\psi^{+}_{\underline{k}',e} Q^{-}_{h,r,e}(\underline{k}', x_{2})\hat\phi^{-}_{\underline{k}'+\underline{k}_{F}^{e},x_{2}, r} \Big] \nn
\eea
and we introduced:
\bea\label{eq:Qgamma}
Q^{+}_{h,r, e}(\underline{k}',y_{2}) &=& Q^{+}_{h+1,r, e}(\underline{k}',y_{2}) + \widehat{W}^{(h)+}_{1,1,0; r, e}(\underline{k}',y_{2})\nn\\
Q^{-}_{h,r,e}(\underline{k}',y_{2}) &=& Q^{-}_{h+1,e, r}(\underline{k}',y_{2}) + \widehat{W}^{(h)-}_{1,1,0; r,e}(\underline{k}',y_{2})\nn\\
Z_{h-1} Z_{h,\sharp,\m, e}(z_{2}) &=& Z_{h} Z_{h+1,\sharp,\m, e}(z_{2}) + Z_{h} \widehat{W}_{2,0,1;(\bar e,\uparrow), (\bar e,\uparrow),\m,\sharp}^{(h)}(\underline{0}, \underline{0},z_{2})\;.
\eea

\subsubsection{Single-scale integration} We are now ready to perform the single scale integration. We write the Grassmann field $\psi^{(\leq h)}$ as a sum of two independent Grassmann fields, $\psi^{(\leq h)} = \psi^{(\leq h-1)} + \psi^{(h)}$, where the {\it single-scale field $\psi^{(h)}$} has covariance given by:
\bea
\int \widetilde{P}_{h}(d\psi^{(h)}) \psi^{(h)-}_{\underline{x}, e} \psi^{(h)+}_{\underline{y}, e'} &=& \delta_{ee'} \int_{\beta,L} \frac{d\underline{k}'}{(2\pi)^2}\, e^{-i \underline{k}'\cdot (\underline{x} - \underline{y})} g^{(h)}_{e}(\underline{k}')\nn\\
g^{(h)}_{e}(\underline{k}') &=& \frac{1}{\widetilde{Z}_{h-1,e}(\underline{k}')} \frac{f_{h,e}(\underline{k}')}{-ik_{0} + \widetilde{v}_{h-1,e}(\underline{k}') k'_{1}}(1 + r_{h-1,e}(\underline{k}'))
\eea
where $f_{h,e}(\underline{k}') = \chi_{h,e}(\underline{k}') - \chi_{h-1,e}(\underline{k}')$ with $\chi_{h-1,e}(\underline{k}') = \chi\big((2^{-h+1}/\d_{e}') \sqrt{k_{0}^{2} + v_{h-1,e}^{2} {{k}_{1}^{'2}}} \big)$ and $v_{h-1,e} = \widetilde{v}_{h-1,e}(\underline{0})$. Instead, the field $\psi^{(\leq h-1)}$ has covariance:
\bea\label{eq:covleq}
\int P_{ h-1}(d\psi^{(\leq h-1)}) \psi^{(\leq h-1)-}_{\underline{x}, e} \psi^{(\leq h-1)+}_{\underline{y}, e'} &=& \delta_{ee'} \int_{\beta,L} \frac{d\underline{k}'}{(2\pi)^2}\,  e^{-i \underline{k}'\cdot (\underline{x} - \underline{y})} \hat g^{(\leq h-1)}_{e}(\underline{k}')\nn\\
\hat g^{(\leq h-1)}_{e}(\underline{k}') &=& \frac{1}{Z_{h-1,e}} \frac{\chi_{h-1,e}(\underline{k}')}{-ik_{0} + v_{h-1,e} k'_{1}}(1 + r_{h-1,e}(k'_{1}))\;.
\eea
In deriving Eq. (\ref{eq:covleq}), we used that $Z_{h-1,e}(\underline{k}') = Z_{h-1,e}(\underline{0})$, $v_{h-1,e}(\underline{k}') = v_{h-1,e}(\underline{0})$ for $\underline{k}'$ in the support of $\chi_{h-1,e}(\underline{k}')$. We then write:
\be
\int \widetilde{P}_{ h}(d\psi^{(\leq h)}) e^{\widetilde{V}^{( h)}(\sqrt{Z_{h}}\psi^{(\leq h)}; A, \phi)} = \int P_{h-1}(d\psi^{(\leq h-1)}) \int \widetilde{P}_{h}(d\psi^{(h)}) e^{\widehat{V}^{(h)}(\sqrt{Z_{h-1}}(\psi^{(\leq h-1)} + \psi^{(h)}); A, \phi)}
\ee
and we define the new effective potential on scale $h-1$ as:
\be\label{eq:Vh-1}
e^{\mathcal{W}^{( h-1)}(A,\phi) + V^{(h-1)}(\sqrt{Z_{h-1}}\psi^{(\leq h-1)}; A,\phi)} = e^{\mathcal{W}^{(h)}(A,\phi) - \beta L t_{h}} \int \widetilde{P}_{h}(d\psi) e^{\widehat{V}^{(h)}(\sqrt{Z_{h-1}}(\psi^{(\leq h-1)} + \psi^{(h)}); A, \phi)}\;.
\ee
The integration of the scale $h$ is done by expanding the exponential, and taking the Gaussian expectation of all monomials; it will be discussed in Section \ref{sec:tree}. As we shall see, under the inductive assumptions (\ref{eq:rcc}), (\ref{eq:asspt}), we will be able to recover the expression (\ref{eq:leqh}) with $h$ replaced by $h-1$. This allows to iterate the process until the last scale\footnote{The fact that the iteration stops at the scale $h = h_{\beta}$ is a consequence of the fact that $k_{0}\in \frac{2\pi}{\beta}(\mathbb{Z} + \frac{1}{2})$, hence $|k_{0}|\geq \frac{\pi}{\beta}$. That is, the temperature provides a natural infrared cutoff.} $h=h_{\beta}$. As discussed below, the inductive assumption (\ref{eq:asspt}) will be verified for the class of models exhibiting single-channel edge modes.

Finally, following for instance Section 3.4 of \cite{GM}, it is possible to prove that the kernels $Q^{\pm}_{h}$ satisfy the recursion relation:
\bea\label{eq:Q+}
Q^{+}_{h,r, e}(\underline{k}',y_{2}) &=& Q^{+}_{h+1,r, e}(\underline{k}',y_{2}) - \widehat{W}_{2,0,0; e, e}^{(h)}(\underline{k}') \sum_{k=h+1}^{1} \hat g^{(k)}_{e}(\underline{k}') Q^{+}_{k,r, e}(\underline{k}',y_{2})\nn\\
Q^{-}_{h,r,e}(\underline{k}',y_{2}) &=& \overline{Q^{+}_{h,r, e}(\underline{k}',y_{2})}
\eea
with $Q^{\pm}_{h,r,e}(\underline{k}';x_{2})$ given by Eq. (\ref{eq:Q1}). Notice that, in the support of $\hat g^{(h)}_{e}(\underline{k}')$, the interation of Eq. (\ref{eq:Q+}) implies:
\be\label{eq:Qh}
Q^{+}_{h,r, e}(\underline{k}',y_{2}) = Q^{+}_{1,r, e}(\underline{k}',y_{2})[ 1 - \widehat{W}^{(h)}_{2,0,0; e,e}(\underline{k}')\hat g^{(h+1)}_{e}(\underline{k}')]\;.
\ee
As we shall see, Eq. (\ref{eq:Qh}) together with the inductive assumptions (\ref{eq:rcc}), (\ref{eq:asspt}) and the bounds discussed in Section \ref{sec:tree} will allow to prove that $Q^{+}_{h,r, e}(\underline{k}',y_{2}) = Q^{+}_{1,r, e}(\underline{k}',y_{2})(1 + O(\l))$.

\subsection{Tree expansion}\label{sec:tree}

The scale $h$ is integrated expanding the exponential in Eq. (\ref{eq:Vh-1}), and taking the Gaussian expectation. We have:
\be\label{eq:ET}
\mathcal{W}^{( h-1)}(A,\phi) + V^{(h-1)}(\sqrt{Z_{h-1}}\psi^{(\leq h-1)}; A,\phi) = \mathcal{W}^{(h)}(A,\phi) - \beta L t_{h} + \sum_{n\geq 0} \frac{1}{n!}\mathcal{E}^{T}_{h} (\underbrace{\widehat{V}^{(h)}\,; \cdots \,; \widehat{V}^{(h)}}_{\text{$n$ times}})
\ee
where $\mathcal{E}^{T}_{h}$ denotes the {\it truncated expectation} (or cumulant) with respect to the Grassmann Gaussian measure $\widetilde{P}_{h}(d\psi)$. Eq. (\ref{eq:ET}) can be iterated over all scales $h,\, h+1,\, \ldots$, until $h=0$. The result can be expressed as a sum of {\it Gallavotti-Nicol\`o} (GN) {\it trees}. To begin, let us discuss the case $A=0$, $\phi = 0$.

\subsubsection{Tree expansion for the free energy}\label{sec:treef}

In this section we will derive a convergent expansion for the kernels $W_{\G}^{(h)}$, with $\G_{A} = \G_{\phi} = \emptyset$, for all scales $h$. We will discuss later the modifications needed in order to take into account the external fields. To describe the expansion, we need to introduce some definitions.

\begin{enumerate}
\item An {\it unlabeled} tree is a connected graph with no loops connecting a point $r$, the root, with an ordered set of $n\geq 1$ points, the endpoints of the tree, so that $r$ is not a branching point. The number $n$ will be called the {\it order} of the unlabeled tree, and its branching points will be called the {\it non-trivial vertices}. The unlabeled trees are partially ordered from the root to the endpoints in a natural way. Two unlabeled trees are identified if they can be superposed by a continuous deformation, which is compatible with the partial ordering of the nontrivial vertices. It is then easy to see that the number of unlabeled trees with $n$ endpoints is bounded by $4^{n}$. We shall also consider the {\it labeled trees} (or just trees, in the following), see Fig. \ref{fig:GN}; they are defined by associating suitable labels with the unlabeled trees, as explained in the following items.
\item To each vertex $v$ of the labeled tree $\tau$ we associate a scale label $h_{v} = h, h+1, \ldots , 0$, as in Fig. \ref{fig:GN}. The scale of the root is $h$. Note that if $v_{1}$ and $v_{2}$ are two vertices and $v_{1}<v_{2}$, then $h_{v_{1}} < h_{v_{2}}$.
\item There is only one vertex immediately following the root, which will be denoted by $v_{0}$ and cannot be an endpoint. Its scale is $h_{v_{0}} = h+1$.
\item With each endpoint $v$ on scale $h_{v} = 0$ we associate one of the monomials contributing to $\widehat{V}^{(0)}(\sqrt{Z_{-1}}\psi^{(\leq 0)})$. Instead, with each endpoint $v$ on scale $h_{v}<0$ we associate one of the monomials contributing to $\mathcal{L} \widehat{V}^{(h_{v})}(\sqrt{Z_{h_{v} -1}} \psi^{(\leq h_{v})})$.
\item We introduce a {\it field label} $f$ to distinguish the field variables appearing in the monomials associated with the endpoints. The set of field labels associated with the endpoint $v$ will be called $I_{v} = \{ f_{1}, \ldots, f_{|I_{v}|} \}$. Given a field $\psi$ labeled by $f$, we denote by $\underline{x}(f)$ its position, by $e(f)$ its quasiparticle label, and by $\e(f)$ its particle-hole label. If $v$ is not an endpoint, we shall call $I_{v}$ the set of field labels associated with the endpoints following the vertex $v$.
\end{enumerate}

\begin{figure}
\centering
 \begin{tikzpicture} 
[scale=0.9, transform shape]
\foreach \i in {4,5,6,7,8,9,10,11,12,13,14} {%
\draw  (\i,2.9) -- (\i, 11.2); }
\foreach \j in {4,5} {%
\draw [very thick] (\j,7) -- ++ (1,0);
\fill (\j,7) circle (0.1);
\node at (\j, 6.7) {$\mathcal R$};
\fill (6,7) circle (0.1);
\node at (6, 6.7) {$\mathcal R$};
}
\foreach \j in {0,1,2,3,4,5} {%
\draw [very thick] (6+\j, 7 -\j *0.5) -- +(1,-0.5);
\fill (6+\j,7-\j*0.5) circle (0.1);
\node at (6+\j, 6.7-\j*0.5) {$\mathcal R$};}
\fill (6+6, 7-3) circle (0.1);
\node at (12, 4.7) {$\mathcal L$};
\foreach \j in {0,1,2,3} {%
\draw [very thick] (6+\j, 7 +\j *0.5) -- +(1,+0.5);
\fill (6+\j,7+\j*0.5) circle (0.1);
\node at (6+\j, 6.7+\j*0.5) {$\mathcal R$};}
\fill (6+4, 7+2) circle (0.1);
\node at (10, 8.7) {$\mathcal R$};
\foreach \j in {0,1} {%
\draw [very thick] (10+\j, 9 +\j *0.5) -- +(1,+0.5);
\fill (10+\j,9+\j*0.5) circle (0.1);
\node at (10+\j, 8.7+\j*0.5) {$\mathcal R$};}
\fill (12, 10) circle (0.1);
\node at (12, 9.7) {$\mathcal L$};
\foreach \j in {0,1,2} {%
\draw [very thick] (10+\j, 9 -\j *0.5) -- +(1,-0.5);
\fill (10+\j,9-\j*0.5) circle (0.1);
\node at (10+\j, 8.7-\j*0.5) {$\mathcal R$};}
\fill (13,7.5) circle (0.1);
\node at (13, 7.2) {$\mathcal L$};
\foreach \j in {0,1} {%
\draw [very thick] (12+\j, 8 +\j *0.5) -- +(1,+0.5);
\fill (12+\j,8+\j*0.5) circle (0.1);
\node at (12+\j, 7.7+\j*0.5) {$\mathcal R$};
}
\fill(14,9) circle (0.1);
\foreach \j in {0} {%
\draw [very thick] (12+\j, 4 +\j *0.5) -- +(1,+0.5);
\fill (12+\j,4+\j*0.5) circle (0.1);
\node at (12+\j, 3.7+\j*0.5) {$\mathcal R$};
}
\foreach \j in {0,1} {%
\draw [very thick] (12+\j, 4 -\j *0.5) -- +(1,-0.5);
\fill (12+\j,4-\j*0.5) circle (0.1);
\node at (12+\j, 3.7-\j*0.5) {$\mathcal R$};}
\fill (14,3) circle (0.1);
\fill (13,4.5) circle (0.1);
\node at (13, 4.2) {$\mathcal L$};
\draw [very thick] (8,8) -- (9, 7.5);
\fill (9,7.5) circle (0.1);
\node at (9, 7.2) {$\mathcal L$};
\draw [very thick] (11,8.5) -- (12, 9);
\fill (12,9) circle (0.1);
\node at (12, 8.7) {$\mathcal L$};
\draw [very thick] (6,7) -- (11,6);
\fill (11,6) circle (0.1);
\node at (11, 5.7) {$\mathcal R$};
\draw [very thick] (11,6) -- (12, 5);
\fill (12,5) circle (0.1);
\draw [very thick]  (11, 6) -- ++ (2,0);
\fill (13,6) circle (0.1);
\node at (13, 5.7) {$\mathcal L$};
\node at (4,2.7) {$h$};
\node at (5,2.7) {$h+1$};
\node at (6,2.7) {$h+2$};
\foreach \i in {7,8} {%
\node at (\i,2.8) {...};}
\node at (9,2.7) {$h_v$};
\node at (10,2.7) {$h_v+1$};
\foreach \i in {11,12,13} {%
\node at (\i,2.8) {...};}
\node at (14,2.7) {$0$};
\node at (9,8.8) {$ v$};
\node at (4,7.3) {$ r$};
\node at (5,7.3) {$ v_0$};
\fill (7,6.8) circle (0.1);
\node at (7, 7.8) {$\mathcal R$};
\fill (8,6.6) circle (0.1);
\node at (8, 6.3) {$\mathcal R$};
\fill (9,6.4) circle (0.1);
\node at (9, 6.1) {$\mathcal R$};
\fill (10,6.2) circle (0.1);
\node at (10, 5.9) {$\mathcal R$};
\fill (12, 6) circle (0.1);
\node at (12, 5.7) {$\mathcal R$};
\end{tikzpicture}
\caption{Example of a Gallavotti-Nicol\`o tree.}
\label{fig:GN}
\end{figure}
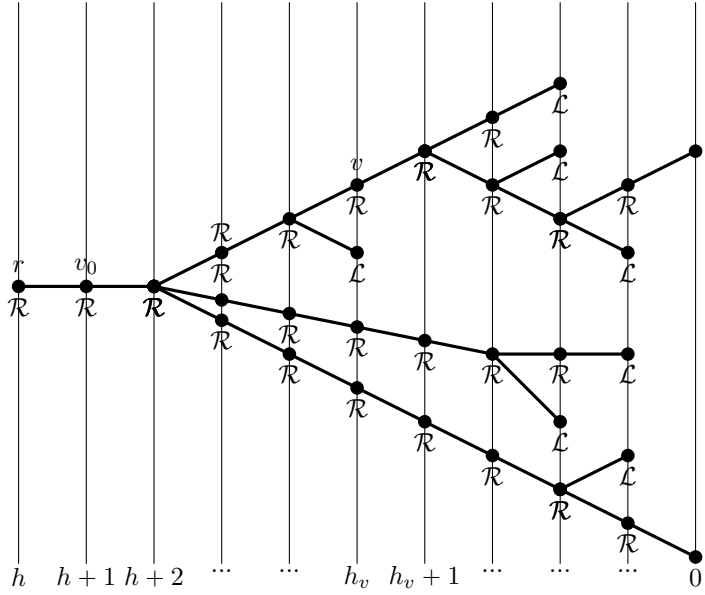

Let $\mathcal{T}_{h,n}$ be the set of labeled trees of order $n$, and with root on scale $h$. We have, setting $\beta L^2 E_{h} \equiv \mathcal{W}^{(h)}(0,0)$:
\be\label{eq:iter}
\beta L^2 E_{h-1} + V^{(h-1)}(\sqrt{Z}_{h-1}\psi^{(\leq h-1)})  = \beta L^2 E_{h} - \beta L t_{h} + \sum_{n=1}^{\infty} \sum_{\tau\in \mathcal{T}_{h,n}}  V^{(h)}(\psi^{(\leq h)}; \tau)\;;
\ee
if $\tau_{1},\, \ldots,\, \tau_{s}$ are the subtrees of $\tau$ with root $v_{0}$, $V^{(h)}(\psi^{(\leq h)}; \tau)$ is defined inductively as follows:
\be\label{eq:iter2}
V^{(h)}(\psi^{(\leq h)}; \tau) = \frac{(-1)^{s+1}}{s!} \mathcal{E}^{T}_{h+1} \big( \overline{V}^{(h+1)}(\psi^{(\leq h)}; \tau_{1})\,;\cdots\,; \overline{V}^{(h+1)}(\psi^{(\leq h+1)}; \tau_{s}) \big)\;,
\ee 
where: $\overline{V}^{(h+1)}(\psi^{(\leq h)}; \tau_{i}) = \mathcal{R} \widehat{V}(\sqrt{Z}_{h} \psi^{(\leq h+1)})$ if $\tau_{i}$ contains more than one endpoint, or if it contains one endpoint but it is not a trivial subtree; $\overline{V}^{(h+1)}(\psi^{(\leq h)}; \tau_{i}) = \mathcal{L} \widehat{V}^{(h +1)}(\sqrt{Z_{h}} \psi^{(\leq h + 1)})$ if $\tau_{i}$ is trivial and $h+1 <0$; or finally $\overline{V}^{(h+1)}(\psi^{(\leq h)}; \tau_{i}) = \widehat{V}^{(0)}(\sqrt{Z_{-1}} \psi^{(\leq 0)})$ if the subtree $\tau_{i}$ is trivial and $h+1 =0$. Using this inductive definition, the right-hand side of Eq. (\ref{eq:iter}) can be further expanded, and in order to describe the resulting expansion we proceed as follows. We associate with any vertex $v$ of the tree a subset $P_{v}$ of $I_{v}$, the {\it external fields} of $v$. These subsets of fields must satisfy the following constraints. First of all, if $v$ is not an endpoint and $v_{1},\ldots, v_{s_{v}}$ are the $s_{v}\geq 1$ vertices immediately following it, then $P_{v} \subseteq \bigcup_{i} P_{v_{i}}$; if $v$ is an endpoint, $P_{v} = I_{v}$. If $v$ is not an endpoint, we shall denote by $Q_{v_{i}} = P_{v} \cap P_{v_{i}}$; this implies that $P_{v} = \bigcup_{i} Q_{v_{i}}$. The union $\mathcal{I}_{v}$ of the subsets $P_{v_{i}} \setminus Q_{v_{i}}$ is, by definition, the set of {\it internal fields} of $v$, and is nonempty if $s_{v} >1$. Given $\tau\in \mathcal{T}_{h,n}$, there are many possible choices of the subsets $P_{v}$, $v\in \tau$, compatible with all the constraints. We shall denote by $\mathcal{P}_{\tau}$ the family of all these choices and ${\bf P}$ the elements of $\mathcal{P}_{\tau}$.

The final expansion is obtained iterating Eq. (\ref{eq:iter2}). In order to better explain the outcome of the iteration, {\it let us suppose for the moment that $\mathcal{R} = 1$:} we will take into account the $\mathcal{R}$ operator in a second moment. Given $v\in \tau$, let $\tau_{v}$ be the subtree of $\tau$ with first nontrivial vertex corresponding to $v$. We have:
\be
V^{(h)}(\psi^{(\leq h)}; \tau_{v}) = \sum_{{\bf P}\in \mathcal{P}_{\tau_{v}}} V^{(h)}({\bf P}; \tau_{v})\;,\qquad V^{(h)}({\bf P}; \tau_{v}) = \int_{\beta, L} D\underline{X}_{v}\, \widetilde{\psi}^{(\leq h)}(P_{v}) K^{(h+1)}_{\tau_{v}, {\bf P}}(\underline{X}_{v})\;,
\ee
where $\underline{X}_{v} = \bigcup_{f\in P_{v}} \underline{x}(f)$ and:
\be\label{eq:tildepsi2}
\widetilde{\psi}^{(\leq h)}(P_{v}) = \prod_{f\in P_{v}} \sqrt{Z_{h, e(f)}} \psi^{(\leq h)\e(f)}_{\underline{x}(f), e(f)}\;.
\ee
%
%
%
The function $K^{(h+1)}_{\tau, {\bf P}}(\underline{X}_{v_{0}})$ is defined inductively by the following equation, valid for any $v\in \tau$ which is not an endpoint:
\bea\label{eq:iterK}
&&K^{(h_{v})}_{\tau, {\bf P}}(\underline{X}_{v})\\&& = \frac{1}{s_{v}!} \Big[\prod_{f\in P_{v}} \frac{\sqrt{Z_{h_{v},e(f)}}}{\sqrt{Z_{h_{v}-1,e(f)}}}\Big] \prod_{i=1}^{s_{v}} [ K^{(h_{v} + 1)}_{v_{i}}(\underline{X}_{v_{i}}) ] \mathcal{E}^{T}_{h_{v}} \big( \widetilde{\psi}^{(h_{v})}(P_{v_{1}} \setminus Q_{v_{1}})\,; \cdots \,; \widetilde{\psi}^{(h_{v})}(P_{v_{s_{v}}} \setminus Q_{v_{s_{v}}} ) \big)\nn
\eea
where $\widetilde{\psi}^{(h_{v})}(P_{v_{i}} \setminus Q_{v_{i}})$ has a definition similar to Eq. (\ref{eq:tildepsi2}). In Eq. (\ref{eq:iterK}), if $v_{i}$ is an endpoint on scale $h_{v} = 0$, $K^{(h_{v} + 1)}_{v_{i}}(\underline{X}_{v_{i}})$ is equal to one of the kernels of the monomials contributing to $\widehat{V}^{(0)}(\sqrt{Z_{-1}}\psi^{(\leq 0)})$; if $v_{i}$ is an endpoint on scale $h_{v} < 0$, $K^{(h_{v} + 1)}_{v_{i}}(\underline{X}_{v_{i}})$  is equal to one of the kernels of the monomials contributing to $\mathcal{L} \widehat{V}^{(h_{v})}(\sqrt{Z_{h_{v}-1}}\psi^{(\leq h_{v})})$; instead, if $v_{i}$ is not an endpoint, $K^{(h_{v} + 1)}_{v_{i}} = K^{(h_{v} + 1)}_{\tau_{i}, {\bf P}_{i}}$, where ${\bf P}_{i} = \{ P_{w}, w\in \tau_{i} \}$. If $v>v_{0}$ then $P_{v} > 0$.

\medskip

Eqs. (\ref{eq:iter})--(\ref{eq:iterK}) are the final form of our expansion. We further decompose $V^{(h)}(\tau, {\bf P})$ by using a convenient representation of the truncated expectation appearing in the right-hand side of Eq. (\ref{eq:iterK}). Let us set $s \equiv s_{v}$, $P_{i} \equiv P_{v_{i}} \setminus Q_{v_{i}}$; moreover, we order in an arbitrary way the sets $P_{i}^{\pm} = \{ f\in P_{i}  \mid \e(f) = \pm\}$, we call $f_{ij}^{\pm}$ their elements and we define $\underline{X}^{-}_{i} = \bigcup_{f\in P_{i}^{-}} \underline{x}(f)$, $\underline{X}^{+}_{i} = \bigcup_{f \in P_{i}^{+}} \underline{x}(f)$, $\underline{x}^{-}_{ij} = \underline{x}(f_{ij}^{-})$, $\underline{x}^{+}_{ij} = \underline{x}(f^{+}_{ij})$. Note that $\sum_{i=1}^{s} |P_{i}^{-}| = \sum_{i=1}^{s} |P_{i}^{+}| \equiv n$, otherwise the truncated expectation vanishes. A pair $\ell = (f^{-}_{ij}, f^{+}_{ij}) \equiv (f^{-}_{\ell}, f^{+}_{\ell})$ will be called a line joining the fields with labels $f^{-}_{ij}$, $f^{+}_{ij}$ and quasi-particle indices $e^{-}_{\ell} = e(f^{-}_{\ell})$, $e^{+}_{\ell} = e(f^{+}_{ij})$, connecting the points $\underline{x}^{-}_{\ell} = \underline{x}^{-}_{ij}$, $\underline{x}^{+}_{\ell} = \underline{x}^{+}_{ij}$, called the endpoints of $\ell$. We then use the {\it Brydges-Battle-Federbush} formula \cite{Bry} for truncated expectations, saying that, up to a sign, if $s>1$:
\be\label{eq:BBF}
\mathcal{E}^{T}_{h}(\widetilde{\psi}^{(h)}(P_{1})\,; \cdots\,; \widetilde{\psi}^{(h)}(P_{s})) = \sum_{T} \prod_{\ell\in T} \delta_{e^{-}_{\ell}e^{+}_{\ell}} Z_{h, e_{\ell}} g^{(h)}_{e_{\ell}}(\underline{x}^{-}_{\ell} - \underline{x}^{+}_{\ell}) \int d \m_{T}({\bf t}) \det G^{h,T}({\bf t})\;,
\ee
where $T$ is a set of lines forming an {\it anchored tree graph} between the clusters of points $\underline{X}_{i} = \{\underline{X}^{-}_{i} \cup \underline{X}^{+}_{i}\}$, that is $T$ is a set of lines which becomes a tree graph if one identifies all the points in the same cluster. Moreover, ${\bf t} = \{ t_{ii'} \in [0,1], 1\leq i,i'\leq s \}$, $d\m_{T}({\bf t})$ is a probability measure with support on a set of ${\bf t}$ such that $t_{ii'} = {\bf u}_{i} \cdot {\bf u}_{i'}$ for some family of vectors ${\bf u}_{i}\in \mathbb{R}^{s}$ of unit norm. Finally, $G^{h,T}(\bf t)$ is a $(n-s+1)\times (n-s+1)$ matrix, whose elements are given by:
\be
G^{h,T}_{ij,i'j'}({\bf t}) = t_{ii'} \delta_{e^{-}_{\ell} e^{+}_{\ell}} Z_{h, e_{\ell}} g^{(h)}_{e_{\ell}}(\underline{x}^{-}_{ij} - \underline{x}^{+}_{i'j'})\;, 
\ee
with $(f^{-}_{ij}, f^{+}_{i'j'})$ not belonging to $T$. In the following, we shall use Eq. (\ref{eq:BBF}) even for $s=1$, when $T$ is empty, by interpreting the r.h.s. as $1$ if $|P_{1}| = 0$, otherwise as equal to $\det G^{h} = \mathcal{E}^{T}_{h} (\widetilde{\psi}^{(h)}(P_{1}))$. Using Eq. (\ref{eq:BBF}) for each vertex of $\tau$ different from its endpoints, we get:
\be\label{eq:sumT}
V^{(h)}({\bf P}; \tau) = \sum_{T \in {\bf T}} \int_{\beta,L} D\underline{X}_{v_{0}} \widetilde{\psi}^{(\leq h)}(P_{v_{0}}) W^{(h)}_{\tau, {\bf P}, T}(\underline{X}_{v_{0}})\;,
\ee
where ${\bf T}$ is a special family of graphs on the set of points $\underline{X}_{v_{0}}$, obtained by putting together an anchored graph $T_{v}$ for each nontrivial vertex $v$. Note that any graph $T\in {\bf T}$ becomes a tree graph on $\underline{X}_{v_{0}}$, if one identifies all the points in the sets $\underline{X}_{v}$, with $v$ an endpoint of the tree. Given $\tau \in \mathcal{T}_{h,n}$ and the labels ${\bf P}, T$, calling $v^{*}_{1},\ldots, v^{*}_{n}$ the endpoints of $\tau$ and setting $h_{i} \equiv h_{v^{*}_{i}}$, the explicit representation of $W^{(h)}_{\tau, {\bf P}, T}(\underline{X}_{v_{0}})$ in Eq. (\ref{eq:sumT}) is:
\bea\label{eq:sumT2}
&&W^{(h)}_{\tau, {\bf P}, T}(\underline{X}_{v_{0}}) = \Big[ \prod_{i=1}^{n} K^{(h_{i})}_{v^{*}_{i}}(\underline{X}_{v^{*}_{i}}) \Big] \\ &&\cdot \prod_{\substack{v \\ \text{not e.p.}}} \prod_{f\in P_{v}}\Big(\frac{\sqrt{Z_{h_{v}, e(f)}}}{\sqrt{Z_{h_{v} -1, e(f)}}}\Big)  \frac{1}{s_{v}!} \int dP_{T_{v}}({\bf t}_{v})\, \det G^{h_{v}, T_{v}}({\bf t}_{v}) \Big[ \prod_{\ell \in T_{v}} \delta_{e^{+}_{\ell} e^{-}_{\ell}} Z_{h_{v}, e_{\ell}} g^{(h_{v})}_{e_{\ell}}(\underline{x}^{-}_{\ell} - \underline{x}^{+}_{\ell})  \Big] \;.  \nn
\eea
Now, recall the inductive assumptions (\ref{eq:rcc}),  (\ref{eq:asspt}):
\be\label{eq:ass}
\Big|\frac{Z_{k,e}}{Z_{k-1,e}}\Big|\leq e^{c|\l|}\;,\qquad |v_{k-1,e} - v_{e}|\leq C|\l|\;,\qquad |\l_{k,\underline{e}}| \leq C|\l|\;,\qquad |\n_{k, e}|\leq C|\l|\;.\nn\\
\ee
These estimates can be used to prove, following \cite{BM}, Section 3.14:
\bea\label{eq:bdGg}
\| \det G^{h_{v}, T_{v}} ({\bf t}_{v}) \|_{\infty} &\leq& C^{\sum_{i=1}^{s_{v}} \frac{|P_{v_{i}}|}{2} - \frac{|P_{v}|}{2} - (s_{v} - 1) } 2^{h_{v}(\sum_{i=1}^{s_{v}} \frac{|P_{v_{i}}|}{2} - \frac{|P_{v}|}{2} - (s_{v} - 1))}\nn\\
| Z_{k, e} \partial_{x_{0}}^{n_1} \partial_{x_{1}}^{n_2} g_{e}^{(k)}(\underline{x} - \underline{y}) | &\leq& \frac{ C_{n+ n_{1} + n_{2}} 2^{k(1 + n_1 + n_2)} }{1 + (2^{k} \| \underline{x} - \underline{y} \|_{\b, L} )^{n+n_{1} + n_{2}}}\;,\qquad \forall n\in \mathbb{N}\;.
\eea
The first bound follows from the Gram-Hadamard inequality for determinants (using that the single-scale propagator admits a Gram representation, Eq. (3.97) of \cite{BM}), while the second follows from the smoothness and  support properties of the single-scale propagator in momentum space. The bounds (\ref{eq:bdGg}) can be used to prove:
\bea\label{eq:treenonrinbd}
&&\frac{1}{\beta L} \int_{\beta,L} D \underline{X}\, | W^{(h)}_{\tau, {\bf P}, T}(\underline{X}_{v_{0}})| \leq \sum_{n\geq 1} C^{n} 2^{h(2 - \frac{|P_{v_{0}}|}{2})}\\&&\qquad \cdot\Big[ \prod_{\text{$v$ not e.p.}} e^{c|\l||P_{v}|} \frac{1}{s_{v}!} 2^{-(h_{v} - h_{v'})(\frac{|P_{v}|}{2} - 2)}  \Big] \Big[ \prod_{i=1}^{n} 2^{h_{v'}(\frac{|I_{v}|}{2} - 2)} \Big]\Big[ \prod_{i=1}^{n} C^{p_{i}} |\l|^{\frac{p_{i}}{2} - 1} \Big]\;,\nn
\eea
where $v'$ is the vertex immediately preceding $v$ on the tree $\t$, and hence $h_{v} - h_{v'} >0$ (in fact, $h_{v} - h_{v'} = 1$). This bound however {\it does not} imply summability in the scale labels for all ${\bf P}$, due to the fact that $P_{v}$ can be smaller than $4$. The contributions to the effective action on scale $h_{v}$ with $|P_{v}| = 2$ are called {\it relevant}; if $|P_{v}| = 4$ they are {\it marginal}, while if $|P_{v}|>4$ they are {\it irrelevant}.

To prove an estimate which is uniform in $h$, we reintroduce the $\mathcal{R}$ operator. The analogous expansion is, see Section 3 of \cite{BM}:
\be\label{eq:sumRT}
V^{(h)}({\bf P}; \tau) = \sum_{T \in {\bf T}}\sum_{\alpha} \int D\underline{X}_{v_{0}} [\widetilde{\psi}^{(\leq h)}(P_{v_{0}})]_{\alpha} W^{(h)}_{\tau, {\bf P}, T,\alpha}(\underline{X}_{v_{0}})\;,
\ee
where:
\be\label{eq:tildepsi}
[\widetilde{\psi}^{(\leq h)}(P_{v})]_{\alpha} = \prod_{f\in P_{v}} \sqrt{Z_{h, e(f)}} \mathcal{D}^{q_{\alpha}(f)}_{\alpha}\psi^{(\leq h)\e(f)}_{\underline{x}(f), e(f)}\;,
\ee
with $q_{\alpha}(f) \in \{ 0,1, 2\}$ and $\mathcal{D}$ the discrete derivative, $\mathcal{D}_{\alpha} \psi^{\e}_{\underline{x},e} = \psi^{\e}_{\underline{x}, e} - \psi^{\e}_{\underline{x}(\alpha), e}$, $\mathcal{D}^2_{\alpha} \psi^{\e}_{\underline{x},e} = \psi^{\e}_{\underline{x}, e} - \psi^{\e}_{\underline{x}(\alpha), e} - (\underline{x} - \underline{x}(\alpha))\cdot \underline{\partial}\psi^{\e}_{\underline{x}, e}$ for a suitable localization point $\underline{x}(\alpha)$; the index $\alpha$ is a multi-index that keeps track of the various terms arising after the action of $\mathcal{R}$ on a Grassmann monomial, see e.g. Eq. (\ref{eq:R}). The kernels $W^{(h)}_{\tau, {\bf P}, T,\alpha}$ admit a representation similar to (\ref{eq:sumT2}), see Eq. (3.81) of \cite{BM}. Proceeding as in Section 3.14 of \cite{BM}, we have:
\bea\label{eq:treebd}
&&\frac{1}{\beta L} \int_{\beta,L} D \underline{X}\, | W^{(h)}_{\tau, {\bf P}, T,\alpha}(\underline{X}_{v_{0}})| \leq \sum_{n\geq 1} C^{n} 2^{h(2 - \frac{|P_{v_{0}}|}{2})}\\&&\qquad \cdot\Big[ \prod_{\text{$v$ not e.p.}} e^{c|\l||P_{v}|} \frac{1}{s_{v}!} 2^{-(h_{v} - h_{v'})(\frac{|P_{v}|}{2} - 2 + z_{v})}  \Big] \Big[ \prod_{i=1}^{n} 2^{h_{v'}(\frac{|I_{v}|}{2} - 2)} \Big]\Big[ \prod_{i=1}^{n} C^{p_{i}} |\l|^{\frac{p_{i}}{2} - 1} \Big]\;,\nn
\eea
where the {\it dimensional gain} $z_{v}$ is due to the presence of the $\mathcal{R}$ operator, and it is:
\be\label{eq:zv0}
z_{v} = \left\{ \begin{array}{cc} 2 & \text{if $|P_{v}| = 2$} \\ 1 & \text{if $|P_{v}| = 4$} \\ 0 & \text{otherwise.} \end{array}\right.
\ee
The number $D_{v} = |P_{v}|/2 - 2 + z_{v}$ is called the {\it scaling dimension} of the vertex $v$. We are now ready to bound the kernel $W^{(h)}_{\Gamma, \alpha}(\underline{X}) = \sum_{\tau \in \mathcal{T}_{h,n}} \sum^*_{{\bf P} \in \mathcal{P}_{\tau}} \sum_{T\in {\bf T}} W^{(h)}_{\tau, {\bf P}, T,\alpha}(\underline{X})$, where, for any $l\geq 1$,  $\G = \{ (e_{1}, \e_{1}), \ldots, (e_{l}, \e_{l}) \}$, and the asterisk denotes the constrains $|P_{v_{0}}| = |\Gamma|\equiv 2l$, $e(f_{i}) = e_{i}$, $\e(f_{i}) = \e_{i}$. Using that the number of unlabeled trees or order $n$ is bounded by $4^{n}$, that the number of addends in $\sum_{T \in {\bf T}}$ is bounded by $C^{n} \prod_{\text{$v$ not e.p.}} s_{v}!$, and definition (\ref{eq:zv0}), we have:
\be\label{eq:bdWh2}
\frac{1}{\beta L} \int_{\beta,L} D\underline{X}\, | W^{(h)}_{2 l, \alpha; \underline{e}}(\underline{X})| \leq 2^{h(2 - l)} \sum_{u\geq 1} C^{u} |\l|^{u}\;.
\ee
In a similar way, one can prove that:
\be\label{eq:bdWh3}
\frac{1}{\beta L} \int_{\beta, L} D \underline{X}\, \prod_{ij}\| \underline{x}_{i} - \underline{x}_{j} \|^{m_{ij}} | W^{(h)}_{2l, \alpha; \underline{e}}(\underline{X})| \leq 2^{h(2 - l - \sum_{ij} m_{ij})} \sum_{u\geq 1} C^{u} |\l|^{u}\;.
\ee
This is the final result of this section. Analyticity in $\l$ can be proven inductively, assuming that the running coupling constants are analytic on all scales $\geq h$, which is true on scale $0$. The bound (\ref{eq:bdWh3}), combined with Eq. (\ref{eq:Qh}), also allows to prove that $Q^{+}_{h,r, e}(\underline{k}',y_{2}) = Q^{+}_{1,r, e}(\underline{k}',y_{2})(1 + O(\l))$.
\begin{rem}{\bf (The short memory property.)}\label{rem:sm} Thanks to the presence of the $z_{v}$ factors in Eq. (\ref{eq:treebd}), each branch between two nontrivial vertices $v>v'$ comes with a factor $2^{-D_{v}(h_{v} - h_{v'})}$, where $D_{v}\geq D = 1 - O(\l)$. This implies that {\em long trees are exponentially suppressed}: if one restricts the sum $\sum_{\tau\in \mathcal{T}_{h, n}}$ to the trees having at least one vertex on scale $k>h$, then the final bound is improved by a factor $2^{\theta(h - k)}$, for $\theta \in (0, D)$. This property is usually referred to as the {\em short memory property of the GN trees.}
\end{rem}
\begin{rem}{\bf (The continuity property.)}\label{rem:cont} Consider a tree $\tau \in \mathcal{T}_{h, n}$ with value $\text{Val}(\tau) \equiv W^{(h)}_{\tau, {\bf P}, T}(\underline{X}_{v_{0}})$. Suppose that $|\text{Val}(\tau)|\leq B_{D}(\tau)$, where $B_{D}(\tau)$ is the dimensional bound of the tree, obtained as in (\ref{eq:treebd}), with $D_{v}\geq D$. Consider a propagator $\hat g^{(h_{v})}_{e}$ arising from the truncated expectation associated with the vertex $v\in \tau$. Suppose that $\hat g^{(h_{v})}_{e}(\underline{k}')$ is replaced by $\hat g^{(h_{v})}_{e}(\underline{k}') + \delta \hat g^{(h_{v})}_{e}(\underline{k}')$, with $|\delta \hat g^{(h_{v})}_{e}(\underline{k}')| \leq C2^{-h_{v}} 2^{\theta h_{v}}$ for some $1/2>\theta >0$, and admitting a Gram representation. Let $\text{Val}(\tau) + \delta \text{Val}(\tau)$ be the new value of tree. The contribution $\delta \text{Val}(\tau)$ can be bounded as follows. Let $\mathcal{C}$ be the path on $\tau$ that connects $v$ to the root. Then, writing $2^{\theta h_{v}} = 2^{\theta h} \prod_{v\in \mathcal{C}} 2^{-\theta}$ and using that $D_{v} \geq 1 - O(\l)$, we have:
\be\label{eq:dV}
|\delta\text{Val}(\tau)|\leq 2^{\theta h} B_{D - \theta}(\tau)\;.
\ee
Analogously, suppose that $Z_{h_{v},e}/ Z_{h_{v,e} - 1}$ is replaced by $Z_{h_{v}, e}/ Z_{h_{v} - 1, e} + \d z_{0, h_{v}, e}$, with $|\d z_{0,h_{v}, e}|\leq C2^{\theta h_{v}}$, or that $\l_{h_{v}, \underline{e}}$ is replaced by $\l_{h_{v}, \underline{e}} + \d \l_{h_{v}, \underline{e}}$, with $|\l_{h_{v}, \underline{e}}|\leq C 2^{\theta h_{v}}$. Then, the value of the corresponding tree is $\text{Val}(\tau) + \d \text{Val}(\tau)$, with $\d \text{Val}(\tau)$ satisfying the bound (\ref{eq:dV}). We shall refer to this property as the {\em continuity property of the GN trees.}
\end{rem}

\subsubsection{Tree expansion for the Schwinger functions}\label{sec:treesc} Let us briefly discuss the changes needed in order to adapt the previous expansion to the Schwinger functions. For each $n\geq 0$ and $m\geq 1$, we introduce a family $\mathcal{T}^{m}_{h,n}$ of rooted labeled trees, defined in a way similar to $\mathcal{T}_{h,n}$, but now allowing for the presence of $m$ special vertices, corresponding to the source terms. We will set $P_{v} = P_{v}^{\psi} + P_{v}^{A} + P_{v}^{\phi}$, where $P_{v}^{A}$, $P_{v}^{\phi}$ collect the external fields of type $A$, $\phi$ associated with the vertex $v$. We then have:
\be\label{eq:Wtree}
\mathcal{W}_{\beta, L}(A,\phi) = \sum_{n\geq 0} \sum_{m\geq 1} \sum_{h\geq h_{\beta}} \sum_{\tau\in \mathcal{T}^{m}_{h, n}} \sum_{{\bf P}\in \mathcal{P}_{\tau}} \sum_{T\in {\bf T}}  \mathcal{W}^{(h)}({\bf P}, T; \tau)\;,
\ee
where the contribution $\mathcal{W}^{(h)}({\bf P}, T; \tau)$ can be further rewritten as:
\be
\mathcal{W}^{(h)}({\bf P}, T; \tau) = \int_{\beta,L} D {\bf Y}_{v_{0}} D{\bf Z}_{v_{0}}\, A(P^{A}_{v_{0}}) \phi(P^{\phi}_{v_{0}}) W^{(h)}_{\tau, {\bf P}, T}({\bf Y}_{v_{0}}, {\bf Z}_{v_{0}})\;,
\ee
where $W^{(h)}_{\tau, {\bf P}, T}({\bf Y}_{v_{0}}, {\bf Z}_{v_{0}})$ admits a tree representation, which can be obtained proceeding as in the case $\phi = A = 0$, with the following differences. If $v$ is a normal endpoint the function $K^{(h_{v})}_{v}$ is defined as before, while if $v$ is a special endpoint it is defined as follows. For a special endpoint of type $\phi^{\pm}_{r,y_{2}}$, then $K^{(h_{v})}_{v}({\bf x})= \int_{\beta, L} d\underline{y}\, \widecheck{Q}^{\pm}_{h_{v},r, e_{v}}(\underline{x} - \underline{y}, x_{2}) g^{(h_{v})}_{e_{v}}(\underline{y})$, with $\widecheck{Q}(\underline{x}, x_{2})$ the inverse Fourier transform of $Q(\underline{k}', x_{2})$; while for a special endpoint of type $A^{\sharp}_{\m}$, $K^{(h_{v})}_{v}({\bf z})  = Z_{h_{v},\sharp, \m, e_{v}}(z_{2})$. The various contributions to Eq. (\ref{eq:Wtree}) can be bounded as in (\ref{eq:treebd}). The main difference is that the scaling dimension of a given vertex (non endpoint) is now $D_{v} = |P_{v}^{\psi}|/2 + |P_{v}^{\phi}|/2 + |P_{v}^{A}| - 2 + z_{v}$, with dimensional gain given by:
\be\label{eq:zv}
z_{v} = \left\{ \begin{array}{cc} 2 & \text{if $|P_{v}^{\psi}| = 2$, $|P^{\phi}_{v}| = 0$, $|P^{A}_{v}| = 0$} \\
1 & \text{if $|P_{v}^{\psi}| = 4$, $|P^{\phi}_{v}| = 0$, $|P^{A}_{v}| = 0$} \\ 1 & \text{if $|P^{\psi}_{v}| = 2$, $|P^{\phi}_{v}| = 0$, $|P^{A}_{v}| = 1$} \\ 0 & \text{otherwise.} \end{array}\right.
\ee
We shall use this expansion to get bounds on the current-current correlation functions. Proceeding in a way analogous to, say, Section 2.2.7, of \cite{GMPcond}, Eqs. (2.79)--(2.84), one finds:
\be\label{eq:JJbd}
\big|\langle {\bf T} j_{\m, \xx}\,; j_{\n, \yy} \rangle_{\beta, L}\big| \leq \frac{C}{1 + \| \underline{x} - \underline{y} \|_{\b, L}^2} \frac{C_{n}}{1 + | x_{2} - y_{2} |^{n}}\;,\qquad \forall n\in \mathbb{N}\;,
\ee
for some $C, C_{n}>0$ independent of $\beta, L$. The $|x_{2} - y_{2}|$-dependence of the bound comes from the decay of the bulk propagators, see Eq. (\ref{eq:gbound}), and from $| Q^{\pm}_{h,r,e}(\underline{k}', x_{2})| \leq C_{n}/(1 + |x_2|_{e}^{n})$ with $|\cdot|_{e} = |\cdot|, |\cdot - L|$, depending on which boundary the edge state is localized, see Eq. (\ref{eq:Qh}).

Concerning the $\|\underline{x} - \underline{y}\|_{\b, L}$-dependence, the bound (\ref{eq:JJbd}) {\it is not enough} to prove the boundedness of $\langle {\bf T} \hat j_{\m, \underline{p}, x_{2}}\,; \hat j_{\n, -\underline{p}, y_{2}} \rangle_{\beta, L}$ uniformly in $\underline{p}$. These correlation functions are the crucial objects entering the Euclidean formulation of the edge transport coefficents, see Eq. (\ref{eq:Dkk'0}). In order to compute the edge transport coefficients, Eq. (\ref{eq:Dkk'0}), we will need to exploit {\it cancellations} in the expansion.

\subsection{The flow of the running coupling constants}\label{sec:rcc}

In order to conclude the discussion of the integration of the single scale, we are left with verifying our inductive assumptions (\ref{eq:ass}). To do this, we shall expoit the recursion relations induced by the multiscale integration, which read as follows:
\bea\label{eq:flow}
\frac{Z_{h-1,e}}{Z_{h,e}} &=& 1+z_{0,h,e} =: 1+ \beta^{z}_{h,e}\nn\\
v_{h-1,e} &=& \frac{Z_{h,e}}{Z_{h-1,e}}(v_{h,e} + z_{1,h,e}) =: v_{h} + \beta^{v}_{h,e}\nn\\
2^{h}\nu_{h,e} &=& \frac{Z_{h,e}}{Z_{h+1,e}} \widehat{W}^{(h)}_{2,0,0; e,e}(\underline{0}) =: 2^{h+1} \nu_{h+1,e} + 2^{h+1}\beta^{\nu}_{h+1,e}\nn\\
\lambda_{h,\underline{e}} &=& \frac{ \sqrt{ Z_{h,e_{1}} Z_{h,e_{2}} Z_{h,e_{3}} Z_{h,e_{4}}}} {\sqrt{Z_{h-1,e_{1}} Z_{h-1,e_{2}} Z_{h-1,e_{3}}Z_{h-1,e_{4}}}}\widehat{W}^{(h)}_{4,0,0;\underline{e}}(\underline{0},\underline{0},\underline{0}) =: \lambda_{h+1,\underline{e}} + \beta^{\l}_{h+1,\underline{e}}\;\nn\\
Z_{h,\sharp,\m, e}(z_{2}) &=& \frac{Z_{h,e}}{Z_{h-1,e}} \widehat{W}^{(h)}_{2,0,1;\m, e,e,\sharp}(\underline{0}, \underline{0},z_{2}) =:  Z_{h+1,\sharp,\m, e}(z_{2}) + \beta^{\sharp}_{h+1,\m,e}(z_{2})\;.
\eea
The function $\beta_{h,\underline{e}} \equiv \beta_{h}(\{Z_{k,e},\, v_{k,e},\, \nu_{k,e},\, \l_{k,\underline{e}}\}_{k=h+1}^{0}) = (\beta^{z}_{h,e}, \beta^{v}_{h,e}, \beta^{\nu}_{e}, \beta^{\l}_{h,\underline{e}}, \beta^{\sharp}_{h,e})$ is called the {\it beta function}, and its properties determine the flow of the running coupling constants. The running coupling constant $\nu_{h,e}$ is relevant, while all the others are marginal. From a mathematical point of view, the main difficulty is to prove that the flow of the effective coupling constant $\l_{h,\underline{e}}$ is {\it bounded}. This is a highly nontrivial property, that needs not to be true in general.

The general solution of this problem is beyond the scope of this paper. Instead, we will restrict the attention to the class of models satisfying Assumpion 2, that is exhibiting single-channel edge currents. We will distinguish two kinds of effective coupling constants, namely:
\be
\l_{h,\underline{e}} = \Big\{ \begin{array}{cc} \l_{h, e} & \text{for } \underline{e} = ((\bar e, \uparrow), (\bar e, \downarrow), (\bar e, \downarrow), (\bar e, \uparrow)) \\ \widetilde \l_{h,\underline{e}} & \text{if  $e_{i}\neq e_{j}$ for some $i,j$.} \end{array}
\ee
We denote by $\beta^{\l}_{h,e}$ the beta function of $\l_{h,e}$, and by $\widetilde{\b}^{\l}_{h,\underline{e}}$ the beta function of $\widetilde{\l}_{h,\underline{e}}$. Notice that for $h=0$, $|\l_{0,e}|\leq C|\l|$ and $|\widetilde \l_{0,\underline{e}}|\leq C |\l| e^{-cL}$. This last bound can be proven from Eq. (\ref{eq:lambda0}), using the exponential decay of the edge modes and of the bulk propagator. The flow of the running coupling constants can be controlled thanks to the following proposition.

\begin{prop}{\bf (Bounds for the beta function.)}\label{prp:flow} Suppose that $|\n_{k,e}|\leq C2^{\theta k}|\l|$, for all $k > h$. Then, under the same assumption of Theorem \ref{thm:1}, the following is true. Let $2^{\theta h_{\beta}} \geq Ce^{-cL}$. The beta function satisfies the following bounds:
\bea\label{eq:beta}
&&|\beta^{z}_{k,e}| \leq C|\l|^22^{\theta k}\,,\qquad |\beta^{v}_{k,e}| \leq C|\l| 2^{\theta k}\;,\qquad |\beta^{\sharp}_{k+1,\m,e}(z_{2})| \leq C|Z_{1,\sharp,\m, e}(z_{2})||\l|^2 2^{\theta (k+1)}\nn\\ &&|\beta^{\l}_{k+1,e}| \leq C|\l|^2 2^{\theta (k+1)}\;,\qquad |\widetilde\beta^{\l}_{k+1,\underline{e}}| \leq C|\l|^2 e^{-2cL}\;,
\eea
for all $k\geq h$, and for some $C,\, c>0$ independent of $h$.
\end{prop}
This proposition together with the recursive equations (\ref{eq:flow}) immediately allows to prove that:
\bea\label{eq:rccbd}
&&|Z_{h,e} - Z_{0,e}|\leq C|\l|^2\;,\qquad |v_{h,e} - v_{0,e}|\leq C|\l|\;,\qquad |Z_{h,\sharp,\m, e}(z_{2})|\leq C|Z_{1,\sharp,\m, e}(z_{2})|\nn\\
&&|\l_{h,e} - \l_{0,e}|\leq C|\l|^2\;,\qquad |\widetilde{\l}_{h,\underline{e}}|\leq C|\l| |h|e^{-2cL} \leq C|\l| e^{-cL}\;.
\eea
which concludes the check of the inductive assumption for all running coupling constants except $\n_{k,e}$, that will be discussed in Proposition \ref{prp:nu}. Also, Proposition \ref{prp:flow} gives a bound on the speed of convergence of the flow of the running coupling constants:
\bea\label{eq:speed}
&&|\l_{h_{\beta}, e} - \l_{h,e}| \leq C|\l|^2 2^{\theta h}\;,\qquad | Z_{h_{\beta},\sharp, \m, e}(z_{2}) - Z_{h,\sharp, \m, e}(z_{2})|\leq C|\l|^22^{\theta h} Z_{1, \sharp,\m, e}(z_{2}) \nn\\
&&|v_{h_{\beta, e}}- v_{h, e}| \leq C|\l| 2^{\theta h}\;,\qquad  \Big| \frac{Z_{h-1, e}}{Z_{h,e}} - 1 \Big|\leq C|\l|^2 2^{\theta h}\;.
\eea
\begin{proof} We proceed by induction. Suppose that the bounds (\ref{eq:beta}) are true for $k>h$. Consider the beta function of the effective coupling constant $\l_{h}$. We rewrite it as:
\be
\beta^{\l}_{h+1, e} = \beta^{\l,(\leq -1)}_{h+1, e} + \delta \b^{\l}_{h+1, e}
\ee
where $\beta^{\l,(\leq -1)}_{h+1,e}$ is given by a sum of trees with endpoints at most on scale $-1$. By the short-memory property of the Gallavotti-Nicol\`o trees, Remark \ref{rem:sm}, we get $|\delta \b^{\l}_{h+1, e} | \leq C|\l|^2 2^{\theta (h+1)}$. Then, we further rewrite:
\be
\beta^{\l,(\leq -1)}_{h+1, e} = \beta^{\l,(\text{rel})}_{h+1, e} + \delta \beta^{\l,(\leq -1)}_{h+1, e}\;,
\ee
where $\beta^{\l,(\text{rel})}_{h+1, e}$ is obtained from $\beta^{\l,(\leq -1)}_{h+1, e}$ by setting, for all $k>h$:
\be\label{eq:rel}
r_{k, e}(\underline{k}') = 0\;,\quad \nu_{k,e} = 0\;,\quad \frac{Z_{k,e}}{Z_{k-1,e}} = 1\;,\quad v_{k,e} =v_{h,e}\;,\quad \l_{k,e} = \l_{h,e}\;,\quad  \widetilde{\l}_{k,\underline{e}} = 0\;. 
\ee
Now, using the inductive assumption on the beta functions of the running coupling constants we have:
\bea\label{eq:cau}
&&| \l_{k,e} - \l_{h+1,e} | \leq C|\l|^2 2^{\theta k}\;,\quad |\widetilde{\l}_{k,\underline{e}}|\leq C|\l|e^{-cL}\;,\quad| \nu_{k,e} | \leq C2^{\theta k}|\l|\;,\\&& | v_{k,e} - v_{h+1,e} | \leq C|\l|2^{\theta k} \;,\qquad\Big| \frac{Z_{k,e}}{Z_{k-1,e}} - 1 \Big|\leq C|\l|^2 2^{\theta k}\;,\qquad |r_{k,e}(k_{1})|\leq C2^{\theta k}\;.\nn
\eea
Correspondingly, we write $\hat g^{(k)}_{e}(\underline{k}') = \hat g^{(k)}_{e,\text{rel}}(\underline{k}') + \delta \hat g^{(k)}_{e}(\underline{k}')$ where $\hat g^{(k)}_{e,\text{rel}}(\underline{k}')$ is the chiral relativistic propagator specified by (\ref{eq:rel}) and $\delta \hat g^{(k)}_{e}(\underline{k}')$ is a correction term:
\be\label{eq:grel}
\hat g^{(k)}_{e, \text{rel}}(\underline{k}') = \frac{\tilde f_{k,e}(\underline{k}')}{-ik_{0} + v_{h,e} k'_{1}}\;,\qquad |\delta \hat g^{(h)}_{e}(\underline{k}')|\leq C 2^{\theta h} 2^{-h_{\underline{k}}}\;,
\ee
with $\tilde f_{k,e}(\underline{k}') = \chi\big( (2^{-k}/\d'_{e})\sqrt{k_{0}^2 + v_{h,e}^2 {k'_{1}}^2} \big) - \chi\big( (2^{-k+1}/\d'_{e})\sqrt{k_{0}^2 + v_{h,e}^2 {k'_{1}}^2}\big)$. Also, we assume that all sums over momenta in $\beta^{\l(\text{rel})}_{h+1, e}$ are replaced by integrals; the error introduced by this replacement is of order $L^{-1} + \beta^{-1} \leq C2^{\theta h}$, see Lemma 2.6 of \cite{BFMhub1}. Therefore, using the continuity of the GN trees, Remark \ref{rem:cont}, we have $|\delta \beta^{\l(\leq -1)}_{h+1, e}| \leq C|\l|^22^{\theta(h+1)}$. Finally, we are left with $\beta^{\l(\text{rel})}_{h+1, e}$. By construction, this contribution is given by a sum of trees involving the relativistic approximation of the propagator, Eq. (\ref{eq:grel}).
Let $\o = \text{sgn}(v_{e})$, that is $v_{h, e} = \o |v_{h, e}|$. Since the velocity is the same in all propagators, we can perform a change of variables in the integrals, and replace $v_{h,e} = \o$ everywhere. We denote by $\hat g^{(k)}_{\o, \text{rel}}$ the propagator (\ref{eq:grel}) with $v_{h, e} = \o$. Let $R_{\alpha}$ be the rotation matrix:
\be
R_{\a} = \begin{pmatrix} \cos\alpha & \sin\alpha \\ -\sin\alpha & \cos\alpha \end{pmatrix}\;.
\ee
We have:
\be\label{eq:rot}
g^{(k)}_{\o, \text{rel}}(R_{\alpha} \underline{k}) = e^{-i \o \alpha} g^{(k)}_{\o, \text{rel}}(\underline{k})\;.
\ee
By the analyticity of $\beta^{\l,(\text{rel})}_{h+1, e}$ in $\l_{h+1, e}$, we write: $\beta^{\l,(\text{rel})}_{h+1, e} = \sum_{p\geq 2} \beta^{\l,(\text{rel})}_{h+1, e, p} \l_{h+1,e}^{p}$. The Taylor coefficient $\beta^{\l,(\text{rel})}_{h+1, e, p}$ is given by the sum of Feynman graphs obtained from the contraction of $4p - 4$ fields, which means $2p - 2$ propagators. Let us consider the $p$-dependent rotation $R_{p} \equiv R_{\frac{\pi}{2p}}$. By Eq. (\ref{eq:rot}), we get:
\be
\beta^{\l,(\text{rel})}_{h+1, e, p} = e^{-i \o \frac{\pi}{2p}(2p - 2)}\beta^{\l,(\text{rel})}_{h+1, e, p} \equiv e^{-i\o\pi(1 - \frac{1}{p})} \beta^{\l,(\text{rel})}_{h+1, e, p}
\ee
which is zero, for all $p\geq 2$. This concludes the inductive step for the proof of the bound on $\beta^{\l}_{h+1,e}$. Then, consider $\beta^{v}_{h,e}$, $\beta^{z}_{h,e}$. Proceeding as before one gets $\beta^{\alpha}_{h,e} = \beta^{\alpha,(\text{rel})}_{h,e} + \d\beta^{\alpha}_{h,e}$ for $\alpha = v, z$,  with $|\d\beta^{v}_{h,e}| \leq C|\l|2^{\theta h}$, $\d\beta^{z}_{h,e} \leq C|\l|^2 2^{\theta h}$. Let us now study the relativistic contribution. We define the directional derivative:
\be
\partial_{\o} := \frac{1}{2}\Big( i\frac{\partial}{\partial k_{0}} + \o \frac{\partial}{\partial k'_{1}} \Big)\;.
\ee
Clearly, $\partial_{k_{0}} = (-i)(\partial_{+} + \partial_{-})$ and $\partial_{k'_1} = \partial_{+} - \partial_{-}$. We have:
\be\label{eq:rot2}
\partial_{\o}\hat g^{(k)}_{\o,\text{rel}}(R_{\a}\underline{k}') = e^{-2 i\o \alpha} \partial_{\o}\hat g^{(k)}_{\o,\text{rel}}(\underline{k}')\;,\qquad \partial_{-\o}\hat g^{(k)}_{\o,\text{rel}}(\underline{k}') = 0\;.
\ee
The Taylor coefficients $\beta^{\alpha,(\text{rel})}_{h,e,p}$ are given by Feynman graphs containing $2p - 1$ propagators, one of them being differentiated by $\partial_{+} \pm \partial_{-}$. We rewrite: $\beta^{\alpha,(\text{rel})}_{h,e,p} = \beta^{\alpha,(\text{rel}),\o}_{h,e,p} + \beta^{\alpha,(\text{rel}),-\o}_{h,e,p}$, depending on whether the coefficient contains $\partial_{\o}\hat g^{(k)}_{\o}(\underline{k}')$ or $\partial_{-\o}\hat g^{(k)}_{\o}(\underline{k}')$. Using Eqs. (\ref{eq:rot}), (\ref{eq:rot2}) we get:
\bea\label{eq:oo}
\beta^{\alpha,(\text{rel}),\o}_{h,e,p} &=& e^{-i\o \pi} \beta^{\alpha,(\text{rel}),\o}_{h,e,p} = 0\qquad \forall p\geq 1\nn\\
\beta^{\alpha,(\text{rel}),-\o}_{h,e,p} &=& 0\qquad \forall p\geq 1\;.
\eea
Therefore, this implies $\beta^{\alpha,(\text{rel})}_{h,e} = 0$, and concludes the inductive step for the proof of the bound on $\beta^{\alpha}_{h+1,e}$. Consider now $\beta^{\sharp}_{h+1,\m,e}$, $\sharp = c,s$. Repeating the previous argument, we get $\beta^{\sharp}_{h+1,\m,e} = \beta^{\sharp,(\text{rel})}_{h+1,\m,e} + \d\beta^{\sharp}_{h+1,e}$,  with $|\d\beta^{\sharp}_{h+1,\m,e}(z_{2})|\leq C |\l|^2 |Z_{1,\sharp,e,\m}(z_{2})| 2^{\theta (h+1)}$. We now notice that the Taylor coefficient $\beta^{\sharp,(\text{rel})}_{h+1,\m,e,p}$ is given by sums of Feynman graphs obtained contracting $4p$ fields, hence $2p$ propagators. We have:
\be
\beta^{\sharp,(\text{rel})}_{h+1,\m,e,p} = e^{-i\o\pi}\beta^{\sharp,(\text{rel})}_{h+1,\m,e,p} = 0\qquad \forall p\geq 1\;.
\ee
This implies $\beta^{\sharp, (\text{rel})}_{h+1,\m,e} = 0$, and concludes the inductive step for the proof of the bound on $\beta^{\sharp}_{h+1,\m,e}$. Finally, the proof of the last of Eq. (\ref{eq:beta}) follows from the observation that linear contribution in $\tilde \l_{e,k}$ to $\widetilde\beta^{\l,\text{(rel)}}_{h+1,\underline{e}}$ is vanishing: its $p$-th Taylor coefficients $\widetilde\beta^{\l,(\text{rel})}_{h+1,\underline{e, p}}$ is given by a sum of diagrams all containing propagators with the same chirality, which allows to repeat the previous argument and conclude that $\widetilde\beta^{\l,(\text{rel})}_{h+1,\underline{e}, p} = 0$ for all $p$. Thus, $\widetilde{\beta}^{\l,(\text{rel})}_{h+1,\underline{e}}$ is at least quadratic in $\widetilde{\l}_{k,\underline{e}}$, $k>h$; the bound in Eq. (\ref{eq:beta}) follows from $|k| e^{-2cL} \leq e^{-cL}$ for $L$ large enough. This concludes the proof.
\end{proof}

To conclude the section, let us show how to control the flow of the parameters $\nu_{k,e}$.

\begin{prop}\label{prp:nu}{\bf (The flow of the chemical potential.)} Suppose that the bounds (\ref{eq:rccbd}) hold true for all scales $k > h$. Then, under the same assumptions of Theorem \ref{thm:1}, there exist $\nu_{e}$, $|\n_{e}|\leq C'|\l|$, such that:
\be\label{eq:nu}
|\nu_{k,e}|\leq C 2^{\theta k}|\l|\;,\qquad \text{for all $k\geq h$,}
\ee
for some constants $C, \theta >0$ independent of $h$.
\end{prop}

\begin{proof}
To begin, we notice that $|\beta^{\nu}_{k+1,e}| \leq C|\l| 2^{\theta k}$, for $\theta >0$. This follows from the fact that $\beta_{k+1,e,p}^{\nu(\text{rel})}$ contains $2p - 1$ propagators; therefore, proceeding as in the proof of Proposition \ref{prp:flow} $\beta_{k+1,e,p}^{\nu(\text{rel})} = e^{-i\o\pi(1 - \frac{1}{2p})}\beta_{k+1,e,p}^{\nu(\text{rel})} = 0$ for all $p\geq 1$. Then, we set: $\beta^{\nu}_{k+1,e} \equiv 2^{\theta k} \widetilde{\beta}^{\nu}_{k+1,e}$, with $|\widetilde{\beta}^{\nu}_{k+1,e}| \leq C|\l|$. We rewrite flow equation for $\nu_{k,e}$ as:
\be
\nu_{k,e} = 2 \nu_{k+1,e} + 2 2^{\theta k}\widetilde{\beta}^{\nu}_{k+1,e} \equiv 2^{-k} \n_{e} + \sum_{j=k}^{0} 2^{j-k+1} 2^{\theta j} \widetilde{\beta}^{\n}_{j+1,e}\;.
\ee
Choosing $\nu_{e} = -\sum_{j=h}^{0} 2^{j+1} 2^{\theta j} \widetilde{\beta}^{\n}_{j+1,e}$, we get:
\be\label{eq:fixpt}
\n_{k,e} = -\sum_{j=h}^{k} 2^{j-k+1}2^{\theta j} \widetilde{\beta}^{\nu}_{j+1,e}\;.
\ee
Eq. (\ref{eq:fixpt}) can be understood as a fixed point equation for the sequence $\{ \nu_{k,e},\, \ldots\,, \nu_{0,e} \}$. It is not difficult to show that Eq. (\ref{eq:fixpt}) defines a {\it contraction} in the space of sequences:
\be
\mathfrak{M}_{h} := \{ (\nu_{k,e})_{k=h}^{0} \mid |\nu_{k,e}|\leq C2^{\theta k} |\l| \}\;.
\ee
See, for instance, Section 4.1 of \cite{BM}. Then, the proof follows from a standard application of Banach contraction mapping principle.
\end{proof}

\subsection{Proof of Proposition \ref{prp:relref}}\label{sec:relref}

In this section we shall use the expansion discussed in the previous sections to prove Proposition \ref{prp:relref}, which is the main technical ingredient in the proof of Theorem \ref{thm:1}. The proof will be based on a comparison of the expansions for the lattice and reference models. In fact, the reference model can be constructed using RG methods similar to those introduced in the previous sections; we refer the reader to \cite{FM, BFMhub2} for details. Let us give a brief overview of the outcome of the RG analysis. The integration of the single scale fields is performed in an iterative way, this time starting from momenta on scale $|\underline{k}|\sim 2^{N}$, with $N>0$, down to the infrared scales. As a result, one has to control the flow of suitable running coupling constants, associated with the relevant and marginal terms appearing in the effective action of the reference model on any scale $h_{\beta} \leq h\leq N$. This has been done in \cite{FM}. The running coupling constants on scale $N$ are the bare couplings, $\l^{\text{(ref)}}$, $Z^{\text{(ref)}}$, $v^{\text{(ref)}}$, $Z^{\text{(ref)}}_{\sharp,\m}, Q^{(\text{ref})\pm}_{r}$. On scale $h<N$, the analog of the running coupling constants $\l_{h,e}$, $Z_{h,e}$, $v_{h,e}$, $Z_{h,\sharp,\m,e}, Q^{\pm}_{h,r,e}$ of the lattice model are denoted by $\l^{\text{(ref)}}_{h,\o}$, $Z^{\text{(ref)}}_{h,\o}$, $v^{\text{(ref)}}_{h,\o}$, $Z^{\text{(ref)}}_{h,\sharp,\m,\o}, Q^{(\text{ref})\pm}_{h,r,\o}$, for $\o = \pm$. By symmetry, there is no $\nu_{h,\o}^{\text{(ref)}}$ appearing in the iteration.

All these running coupling constants are {\it marginal}. Their flow in the ultraviolet regime, $1\leq h\leq N$, has been controlled in Theorem 2 of $\cite{FM}$. In the infrared regime, the flow of the running coupling constants can be controlled using the same argument of the proof of Proposition \ref{prp:flow}, to show the vanishing of the relativistic part of the beta function of the lattice model. We summarize these results in the following proposition.

\begin{prop}{\bf (The running coupling constants of the reference model.)}\label{prp:rccref} There exists $\bar \l>0$ such that for $|\l^{\text{(ref)}}|<\bar \l$ the following is true. The running coupling constants of the reference model on scale $0$ satisfy the bounds:
\bea\label{eq:uvref}
&&|\l^{\text{(ref)}} - \l^{\text{(ref)}}_{0, \o}|\leq C|\l^{\text{(ref)}}|^2\;,\qquad |Z^{\text{(ref)}} - Z^{\text{(ref)}}_{0, \o}|\leq C|\l^{\text{(ref)}}|^2\;,\nn\\
&&|Z^{\text{(ref)}}_{\sharp,\m}(z_{2}) - Z^{\text{(ref)}}_{0,\sharp,\m,\o}(z_{2})|\leq C|\l^{\text{(ref)}}|^2 |Z^{\text{(ref)}}_{\sharp,\m}|\;,\qquad | \o v^{\text{(ref)}} - v^{\text{(ref)}}_{0,\o}|\leq C|\l^{\text{(ref)}}|\nn\\ &&|Q^{\text{(ref)}\pm}_{r}(x_{2}) - Q^{\text{(ref)}\pm}_{0,r,\o}(x_{2})|\leq C|\l^{\text{(ref)}}| |Q^{\text{(ref)}\pm}_{r}(x_{2})|\;.\nn\\
\eea
Also, for all $h\leq 0$:
\bea\label{eq:irref}
&&\l^{\text{(ref)}}_{h, \o} = \l^{\text{(ref)}}_{0,\o}\;,\qquad Z^{\text{(ref)}}_{h,\o} = Z^{\text{(ref)}}_{0,\o}\;,\qquad v^{\text{(ref)}}_{h,\o} = v^{\text{(ref)}}_{0,\o}\;,\nn\\ &&Z^{\text{(ref)}}_{h,\sharp,\m,\o} = Z^{\text{(ref)}}_{0,\sharp,\m,\o}\;,\qquad Q^{\text{(ref)}\pm}_{h,r,\o}(x_{2}) = Q^{\text{(ref)}\pm}_{0,r,\o}(x_{2})\;.
\eea
Moreover, the running coupling constants are analytic in $\l^{\text{(ref)}}$ and are smooth functions of their bare counterparts, with first derivatives bounded away from zero uniformly $\l^{\text{(ref)}}$.
\end{prop}

The proof of Eq. (\ref{eq:uvref}) is given by Theorem 2 of \cite{FM}, while the proof of Eq. (\ref{eq:irref}) is a repetition of the argument used to prove the vanishing of the relativistic part of the beta function for the lattice model. Finally, the statement about differentiability of the running coupling constants is a byproduct of the proof of Theorem 2 of \cite{FM}.

To conclude the introduction to this section, let us briefly anticipate how the reference model will be used to describe the lattice model. Let $e = (1,\s)$ be the edge states localized around the $x_{2} = 0$ edge. By the implicit function theorem, we will choose the bare parameters of the reference model in such a way that the running coupling constants of the reference model parametrized by $\o = \text{sgn}(v_{e})$ match the running coupling constants of the lattice model on the smallest possible scale $h = h_{\beta}$. For this choice of bare parameters, we will be able to prove that the asymptotic behavior of the correlations of the lattice model is captured by the correlations of the reference model.

\medskip

\noindent{\it Proof of Proposition \ref{prp:relref}.} The proof is based on the tree expansion of Section \ref{sec:tree}. Let us start by proving the first of (\ref{eq:latWI}). The starting point is to separate the trees with endpoints on scale $\leq -1$ from the trees having at least one endpoint on scale $\geq 0$. Let us denote by $\mathcal{T}^{2(\leq- 1)}_{h,n}$ the former set of trees. Setting $\o := \text{sgn}(v_{(1,\s)})$ and $k_{F}^{\o} := k_{F}^{(1,\s)}$, we have:
\be\label{eq:2ptb}
\pmb{\langle} {\bf T} \hat a^{-}_{\underline{k}' + \underline{k}_{F}^{\o},x_{2}, r} \hat a^{+}_{\underline{k}' + \underline{k}_{F}^{\o},y_{2}, r'} \pmb{\rangle}_{\beta, L} = \sum_{n\geq 0}\sum_{h_{\beta} \leq h \leq -1} \sum^*_{\tau \in \mathcal{T}^{2(\leq -1)}_{h,n}} \sum^{*}_{{\bf P} \in \mathcal{P}_{\tau}} \sum_{T \in {\bf T}}  \widehat{W}^{(h)}_{\tau, {\bf P}, T}(\underline{k}'; x_{2}, y_{2}) + \mathcal{E}^{(1)}_{2,0}(\underline{k}'; x_{2}, y_{2})
\ee
where the asterisk on the ${\bf P}$ sum recalls that $P_{v_{0}} = P_{v_0}^{\phi} = \{ f_{1}, f_{2} \}$, $\e(f_{1}) = -\e(f_{2}) = -$, $\underline{k}(f_{1}) = \underline{k}(f_{2}) = \underline{k}' + \underline{k}_{F}^{e}$, $r(f_{1}) = r$, $r(f_{2}) =  r'$, $x_{2}(f_{1}) =x_{2}$, $x_{2}(f_{2}) = y_{2}$. The asterisk on the $\tau$ sum also recalls that only trees decorated with $e = (1,\s)$ labels appear in the sum; the contributions coming from the edge state $(2,\s)$ are collected in the error term. The dimensional bound for the two-point function that can be proven via the tree expansion of Section \ref{sec:treesc} is $O(2^{-h_{\underline{k}'}} |x_{2} - y_{2}|^{-n})$, with $h_{\underline{k}'} = \lfloor \log_{2}|\underline{k}'| \rfloor$. Therefore, by the short memory property of the GN trees (see Remark \ref{rem:sm}), using also $|\tilde \l_{h,\underline{e}}|\leq C|\l|e^{-cL}\leq C|\l|2^{\theta h}$, it is not difficult to see that $|\mathcal{E}^{(1)}_{2,0}(\underline{k}'; x_{2}, y_{2}) | \leq C_{n}2^{\theta h_{\underline{k}'}} 2^{-h_{\underline{k}'}}/(1+ |x_{2} - y_{2}|^{n}) $ for some $\theta >0$. Let us now focus on the first term in the right-hand side of Eq. (\ref{eq:2ptb}). The goal is to rewrite it as the contribution coming from a suitable reference model, plus subleading terms in $\underline{k}'$. We consider a reference model with bare parameters chosen so that:
\be\label{eq:refpar}
\l^{\text{(ref)}}_{h_{\beta}, \o}= \l_{h_{\b},e}\;,\qquad Z^{\text{(ref)}}_{h_{\beta}, \o} = Z_{h_{\beta}, e}\;,\qquad v^{\text{(ref)}}_{h_{\beta},\o} = v_{h_{\beta}, e}\;,\qquad Q^{\text{(ref)}\pm}_{h_{\beta}, r, \o}(x_{2}) = Q^{\pm}_{h_{\beta}, r, e}(\underline{k}_{F}^{e}, x_{2})\;.
\ee
This can be done by the implicit function theorem, thanks to the fact that all running coupling constants are nonconstant, differentiable functions of their bare counterparts, recall Proposition \ref{prp:rccref}. Notice that, by Eq. (\ref{eq:irref}), the flow of the beta function of the reference model is trivial on scales $h\leq 0$. Then, by Eqs. (\ref{eq:speed}) and by Proposition \ref{prp:rccref} we get:
\bea
&&|\l^{\text{(ref)}}_{k, \o} - \l_{k,e}| \leq C|\l|^22^{\theta k}\;,\qquad \Big| \frac{Z_{k-1,e}}{Z_{k,e}} - 1 \Big|\leq C|\l|^22^{\theta k}\;,\\  
&& |v^{\text{(ref)}}_{k,\o} - v_{k,e}|\leq C|\l| 2^{\theta k}\;,\qquad |Q^{\text{(ref)}\pm}_{k,r,\o}(x_{2}) - Q^{\pm}_{k,r,e}(\underline{k}_{F}^{e}, x_{2})|\leq C|\l| 2^{\theta k} Q^{\pm}_{1,r,e}(\underline{k}_{F}^{e}, x_{2})\;.\nn
\eea
Correspondingly, we write $\hat g^{(k)}_{e}(\underline{k}') = \hat g^{(k)}_{e,\text{ref}}(\underline{k}') + \delta \hat g^{(k)}_{e}(\underline{k}')$, where $\hat g^{(k)}_{e,\text{ref}}(\underline{k}')$ is the single-scale propagator of the reference model, given by Eq. (\ref{eq:grel}) with $v_{k,e}$ replaced by $v_{h_{\beta},e}$, and where $|\delta \hat g^{(k)}_{e}(\underline{k}')|\leq C2^{\theta k} 2^{-k}$. Therefore,
\bea\label{eq:2pt2}
&&\pmb{\langle} {\bf T} \hat a^{-}_{\underline{k}' + \underline{k}_{F}^{e},x_{2}, r} \hat a^{+}_{\underline{k}' + \underline{k}_{F}^{e},y_{2}, r'} \pmb{\rangle}_{\beta,L}\nn\\&&\qquad = \sum_{n\geq 0}\sum_{h_{\beta} \leq h \leq 1} \sum_{\tau \in \mathcal{T}^{2(\leq -1)}_{h,n}} \sum^{*}_{{\bf P} \in \mathcal{P}_{\tau}} \sum_{T \in {\bf T}}  \widehat{W}^{(h)(\text{ref})}_{\tau, {\bf P}, T}(\underline{k}'; x_{2}, y_{2}) + \sum_{\alpha = 1,2} \mathcal{E}^{(\alpha)}_{2,0}(\underline{k}'; x_{2}, y_{2})\;,
\eea
where $\widehat{W}^{(h)(\text{ref})}_{\tau, {\bf P}, T}(\underline{k}'; x_{2}, y_{2})$ is the value of the tree of a reference model specified by the previous discussion, and $\mathcal{E}^{(2)}_{2,0}(\underline{k}'; x_{2}, y_{2})$ is a new error term, collecting all the errors introduced with the replacement of the running coupling constants and the propagators with those of the reference model. By the continuity of GN trees, Remark \ref{rem:cont}, $|\mathcal{E}^{(2)}_{2,0}(\underline{k}'; x_{2}, y_{2}) | \leq C_n 2^{\theta h_{\underline{k}'}} 2^{-h_{\underline{k}'}} /(1+ |x_{2} - y_{2}|^{n})$. Therefore, we write:
\bea\label{eq:zom}
&&\sum_{n\geq 0}\sum_{h_{\beta} \leq h \leq -1} \sum_{\tau \in \mathcal{T}^{2(\leq -1)}_{h,n}} \sum^{*}_{{\bf P} \in \mathcal{P}_{\tau}} \sum_{T \in {\bf T}}  \widehat{W}^{(h)(\text{ref})}_{\tau, {\bf P}, T}(\underline{k}'; x_{2}, y_{2})\nn\\&&\qquad = Q^{\text{(ref)}+}_{r}(x_{2})Q^{(\text{ref})-}_{r'}(y_{2}) \pmb{\langle} \hat{\psi}^{-}_{\underline{k}',{\o},\s} \hat{\psi}^{+}_{\underline{k}',\o,\s'} \pmb{\rangle}^{\text{(ref)}}_{\beta, L} + \mathcal{E}^{(3)}_{2,0}(\underline{k}'; x_{2}, y_{2})
\eea
where $\mathcal{E}^{(3)}_{2,0}(\underline{k}'; x_{2}, y_{2})$ is the error term collecting all trees of the reference model with at least one endpoint on scale $\geq 0$. Thus, by the short-memory of the GN trees of the reference model and the decay of $Q^{\pm}$ functions, $|\mathcal{E}^{(3)}_{2,0}(\underline{k}'; x_{2}, y_{2})|\leq C_n2^{\theta h_{\underline{k}'}} 2^{-h_{\underline{k}'}}/(1 + |x_{2} - y_{2}|^{n})$. All in all, by Eqs. (\ref{eq:2pt2}), (\ref{eq:zom}), using that $\pmb{\langle} \hat{\psi}^{-}_{\underline{k}',{\o},\s} \hat{\psi}^{+}_{\underline{k}',\o,\s'} \pmb{\rangle}^{\text{(ref)}}_{\beta, L} \simeq 2^{-h_{\underline{k}'}}$:
\be
\pmb{\langle} {\bf T} \hat a^{-}_{\underline{k}' + \underline{k}_{F}^{\o},x_{2}, r} \hat a^{+}_{\underline{k}' + \underline{k}_{F}^{\o},y_{2}, r'} \pmb{\rangle}_{\beta,L} = \pmb{\langle} \hat{\psi}^{-}_{\underline{k}',{\o},\s} \hat{\psi}^{+}_{\underline{k}',\o,\s'} \pmb{\rangle}^{\text{(ref)}}_{\beta, L} [Q^{\text{(ref)}+}_{r}(x_{2})Q^{(\text{ref})-}_{r'}(y_{2})+ O(2^{\theta h_{\underline{k}'}} |x_{2} - y_{2}|^{-n} )]
\ee
This proves the first of Eq. (\ref{eq:latWI}). Consider now the second of Eq. (\ref{eq:latWI}). Proceeding as in Eqs. (\ref{eq:2ptb})--(\ref{eq:2pt2}) we get, for $\m=0,1$, $\|\underline{k}'\| = \kappa$, $\| \underline{k}' + \underline{p} \|\leq \kappa$, $\| \underline{p} \| \ll \kappa$ and $\kappa$ small enough:
\bea\label{eq:3decomp}
&& \pmb{\langle} {\bf T} \hat j^{\sharp}_{\m, \underline{p}, z_{2}}\,; \hat a^{-}_{\underline{k}' + \underline{k}_{F}^{\o},x_{2}, r} \hat a^{+}_{\underline{k}' + \underline{p} + \underline{k}_{F}^{\o},y_{2}, r'} \pmb{\rangle}_{\beta, L } \\&&= \sum_{n\geq 0}\sum_{h_{\beta} \leq h \leq -1} \sum_{\tau \in \mathcal{T}^{2(\leq -1)}_{h,n}} \sum^{**}_{{\bf P} \in \mathcal{P}_{\tau}} \sum_{T \in {\bf T}}  \widehat{W}^{(h)(\text{ref})}_{\tau, {\bf P}, T}(\underline{k}', \underline{p}; x_{2}, y_{2}, z_{2})  + \sum_{\alpha = 1,2} \mathcal{E}^{(\alpha)}_{2,1}(\underline{k}',\underline{p}; x_{2}, y_{2}, z_{2})\nn
\eea
where the double asterisk recalls the constraint $P_{v_{0}} = P_{v_0}^{A} \cup P_{v_0}^{\phi}$, with $P_{v_{0}}^{\phi}$ as after (\ref{eq:2ptb}), and $P_{v_0}^{A} = \{ f_{3} \}$ with $\m(f_{3}) = \m$, $\underline{p}(f_{3}) = \underline{p}$, $x_{2}(f_{3}) = z_{2}$. The first term in the right-hand side of Eq. (\ref{eq:3decomp}) corresponds to a reference model specified by the conditions (\ref{eq:refpar}), plus the presence of a source term coupled to the $A$ field, with a suitable $z_{2}$-dependent bare parameter $Z^{(\text{ref})}_{\sharp,\m}(z_{2})$, chosen so that $Z^{(\text{ref})}_{h_{\beta}, \sharp,\m, \o}(z_{2}) = Z_{h_{\beta}, \sharp,\m, e}(z_{2})$, which can be done by the implicit function theorem; also, by Eq. (\ref{eq:speed}) and by Proposition \ref{prp:rccref} we get, for $\m = 0, 1$, $|Z_{k,\sharp,\m,e}(z_{2}) - Z^{\text{(ref)}}_{k, \sharp, \m, \o}(z_{2})| \leq C|\l| 2^{\theta k} |Z_{1, \sharp, \m, \o}(z_{2})|$. The error term $\mathcal{E}^{(1)}_{2,1}$ takes into account the trees containing at least one $\widetilde{\l}_{h_{v},\underline{e}}$ endpoint, and the trees with at least one endpoint on scale $\geq 0$, while the error term $\mathcal{E}^{(2)}_{2,1}$ contains the contributions arising from the replacement of the running coupling constants and of the propagators with those of the reference model.  Notice that in the dimensional bound for the tree expansion of this Schwinger function all vertices $v$ have negative scaling dimension except for those corresponding to monomials $A\phi^{+}\psi^{-}$, $A\psi^{+}\phi^{-}$, which have zero scaling dimension. The constraint on the external momenta forces the number of vertices of this type to be of order $1$, therefore no logarithmic divergence arises in the sum over the scale labels.  Proceeding as in Eqs. (\ref{eq:2ptb})--(\ref{eq:zom}), we have:
\bea
&&\sum_{n\geq 0}\sum_{h_{\beta} \leq h \leq 1} \sum_{\tau \in \mathcal{T}^{2(\leq -1)}_{h,n}} \sum^{**}_{{\bf P} \in \mathcal{P}_{\tau}} \sum_{T \in {\bf T}}  \widehat{W}^{(h)(\text{rel})}_{\tau, {\bf P}, T}(\underline{k}', \underline{p}; x_{2}, y_{2}, z_{2})\\ && \qquad = Z^{\text{(ref)}}_{\sharp,\m}(z_{2})Q^{(\text{ref})+}_{r, \o}(x_{2})Q^{(\text{ref})-}_{r',\o}(y_{2}) \pmb{\langle} \hat n_{\underline{p}, \o}\,;  \hat{\psi}^{-}_{\underline{k}',{\o}} \hat{\psi}^{+}_{\underline{k}',\o} \pmb{\rangle}^{\text{(ref)}}_{\beta,  L} + \mathcal{E}^{(3)}_{2,1}(\underline{k}',\underline{p}; x_{2}, y_{2}, z_{2})\;.\nn
\eea
where $\mathcal{E}^{(3)}$ takes into account the trees of the reference model with at least one endpoint on scale $\geq 0$. The fast decay of $Z_{h_{\beta},\sharp,\m,e}(z_{2})$, together with the bound in Eq. (\ref{eq:speed}) and Proposition \ref{prp:rccref}, implies that $|Z^{(\text{ref})}_{\sharp,\m}(z_{2})|\leq C_{n}/(1 + |z_{2}|_{e}^{n})$. By the short-memory and the continuity of the GN trees, the error terms are bounded as: $|\mathcal{E}^{(\alpha)}_{2,1}(\underline{k},\underline{p}; x_{2}, y_{2}, z_{2})|\leq C_{n}\kappa^{-2 + \theta}/(1 + d(x_{2}, y_{2}, z_{2})^{n})$. Using also that $\pmb{\langle} \hat n_{\underline{p}, \o}\,;  \hat{\psi}^{-}_{\underline{k}',{\o}} \hat{\psi}^{+}_{\underline{k}',\o} \pmb{\rangle}^{\text{(ref)}}_{\beta,  L}\simeq 2^{-2h_{\underline{k}'}}$, we get:
\bea
&&\pmb{\langle} {\bf T} \hat j^{\sharp}_{\m, \underline{p}, z_{2}}\,; \hat a^{-}_{\underline{k}' + \underline{k}_{F}^{\o},x_{2}, r} \hat a^{+}_{\underline{k}' + \underline{p} + \underline{k}_{F}^{\o},y_{2}, r'} \pmb{\rangle}_{\beta, L }\\ &&= \pmb{\langle} \hat n_{\underline{p}, \o}\,;  \hat{\psi}^{-}_{\underline{k}',{\o}} \hat{\psi}^{+}_{\underline{k}',\o} \pmb{\rangle}^{\text{(ref)}}_{\beta,  L}\Big[ Z^{\text{(ref)}}_{\sharp,\m}(z_{2})Q^{(\text{ref})+}_{r, \o}(x_{2})Q^{(\text{ref})-}_{r',\o}(y_{2}) + O(\kappa^{\theta} d(x_{2}, y_{2}, z_{2})^{-n}) \Big]\;.\nn
\eea
This proves the second of (\ref{eq:latWI}). To conclude, let us prove the last one of Eq. (\ref{eq:latWI}). Proceeding as for the previous two Schwinger functions one gets, for $\m,\n=0,1$:
\be
\pmb{\langle} {\bf T} \hat j^{\sharp}_{\m, \underline{p}, x_{2}}\,; \hat j^{\sharp'}_{\n, -\underline{p}, y_{2}} \pmb{\rangle}_{\beta, L} = Z^{(\text{ref})}_{\sharp,\m}(x_{2})Z^{(\text{ref})}_{\sharp',\n}(y_{2}) \pmb{\langle} \hat n^{\sharp}_{\underline{p}, \o}\,; \hat n^{\sharp'}_{-\underline{p},\o}  \pmb{\rangle}_{\beta, L}^{(\text{ref})} + \sum_{\alpha = 1,2,3} \mathcal{E}^{(\alpha)}_{0,2}(\underline{p}; x_{2}, y_{2})\;,
\ee
where the error terms have the same meaning as in the previous discussions. We claim that, for some $\theta >0$:
\be\label{eq:bdE}
|\mathcal{E}^{(\alpha)}_{0,2}(\underline{p}; x_{2}, y_{2})| \leq \frac{C_{n}}{1+|x_{2} - y_{2}|^{n}}\;,\quad | \mathcal{E}^{(\alpha)}_{0,2}(\underline{p}; x_{2}, y_{2}) - \mathcal{E}^{(\alpha)}_{0,2}(\underline{0}; x_{2}, y_{2}) | \leq \frac{C_{n} |\underline{p}|^{\theta}}{1 + |x_{2} - y_{2}|^{n}}\;,\quad \forall n\in \mathbb{N}\;.
\ee
These bounds imply the estimates for the error term in the last one of Eq. (\ref{eq:latWI}). To prove them, we proceed as follows. Consider the error terms in coordinate space, $\widecheck{\mathcal{E}}^{(\alpha)}_{0,2}(\xx,\yy)$. Forgetting about the dimensional gains coming from the short-memory and the continuity properties of the GN trees, one gets, proceeding as in Section 2.2.7, of \cite{GMPcond}, Eqs. (2.79)--(2.84), $|\widecheck{\mathcal{E}}^{(\alpha)}_{0,2}(\xx,\yy)|\leq C_{n}\|\underline{x} - \underline{y}\|_{\beta, L}^{-2} |x_{2} - y_{2}|^{-n}$. This bound is not even enough to prove uniform boundedness in $\underline{p}$ of the Fourier transform. Instead, exploiting the dimensional gain $\sim 2^{\tilde \theta h}$ for the trees with root on scale $h$ contributing to $\widecheck{\mathcal{E}}^{(\alpha)}_{0,2}(\xx,\yy)$, implied by the short-memory and the continuity of the GN trees, one gets the improved bound: $|\widecheck{\mathcal{E}}^{(\alpha)}_{0,2}(\xx,\yy)|\leq C_{n}\|\underline{x} - \underline{y}\|_{\beta, L}^{-2-\tilde \theta} |x_{2} - y_{2}|^{-n}$. This bound can be used to prove the Lipshitz continuity in $\underline{p}$ of $|x_{2} - y_{2}|^{n}\mathcal{E}^{(\alpha)}_{0,2}(\underline{p}; x_{2}, y_{2})$ with exponent $0<\theta<\tilde \theta$. This proves the bounds (\ref{eq:bdE}). Finally, the existence of the $\beta, L\to\infty$ limit of the correlation function can be proven as in Lemma 2.6 of \cite{BFMhub1}. This concludes the proof of Proposition \ref{prp:relref}.\qed

\medskip

\noindent{\bf Acknowledgements.} The work of G. A. and M. P. has been carried out thanks to the support of the NCCR SwissMap. Furthermore, the work of G. A. and of M. P. has been supported by the Swiss National Science Foundation via the grant ``Mathematical Aspects of Many-Body Quantum Systems''. V. M. has received funding from the European Research Council (ERC) under the European Union's Horizon 2020 research and innovation programme (ERC CoG UniCoSM, grant agreement n.724939). V. M. also acknowledges support from the Gruppo Nazionale per la Fisica Matematica (GNFM).

\appendix

\section{Details on the integration of the massive modes}\label{app:1d}

\subsection{Bounds on the propagators}\label{app:bulk}

In this Appendix we shall discuss some details of the proof of Proposition \ref{prp:1d}. We will prove two key technical results, that allow to prove Proposition \ref{prp:1d} by repeating the analysis of Section 5.2 of \cite{GMP}. We start with the proof of Eq. (\ref{eq:gbound}).

\begin{prop}{\bf (Decay of the bulk propagator.)}\label{prp:decay} For any $n\in \mathbb{N}$ there exists $C_{n}>0$, independent of $\beta, L, N, \d$ and $c>0$, independent of $\beta, L, N, \d, n$ such that:
\be\label{eq:gbd}
| g^{(\text{bulk})}_{\beta,L,N}(\xx, r; \yy, r')| \leq \frac{C_{n} \d^{-2}}{1 + (\d \| \underline{x} - \underline{y} \|_{\beta, L})^{n}} e^{-c \d |x_{2} - y_{2}|}\qquad \forall\, \xx,\yy\;.
\ee
\end{prop}
\begin{proof} Recall Eq. (\ref{eq:gbulk}),
\bea\label{eq:gbd0}
\hat g^{\text{(bulk)}}_{\beta, L,N}(\underline{k}; x_{2}, r; y_{2},  r') &=& \sum_{e=1}^{n_{\text{edge}}} \Big(\frac{\chi_{N}(k_{0})\chi_{e}(k_{1})}{-ik_{0} + \hat H(k_{1}) - \m}P_{\perp}^{e}(k_{1})\Big)(x_{2},r; y_{2},  r')\nn\\
&& +  \Big(\frac{\chi_{N}(k_{0}) \chi_{\geq}(k_{1})}{-ik_{0} + \hat H(k_{1}) - \m}\Big)(x_{2},r; y_{2},  r')\;.
\eea
We claim that for all $\underline{k}\in \mathbb{D}_{\b, L}$:
\be\label{eq:res}
\big|\partial_{\underline{k}}^{n} \hat g^{\text{(bulk)}}_{\beta, L,N}(\underline{k}; x_{2}, r; y_{2},  r')\big| \leq \frac{C_{n}}{(k_{0}^2 + \d^{2})^{\frac{n+1}{2}}} e^{-c \d |x_{2} - y_{2}|}\;,
\ee
with $\partial_{\underline{k}}$ the discrete gradient: $\partial_{k_{0}} f(k_{0}, k_{1}) = (\beta/2\pi) (f(k_{0} + 2\pi/\beta, k_{1}) - f(k_{0}, k_{1}))$ and $\partial_{k_{1}} f(k_{0}, k_{1}) = (L/2\pi) (f(k_{0}, k_{1} + 2\pi/L) - f(k_{0}, k_{1}))$. If so, the bound (\ref{eq:gbd}) can be easily proven via integration by parts. Let $d_{L}(x_{1}) = \frac{L}{\pi} \sin \Big(\frac{\pi x_{1}}{L}\Big)$, $d_{\b}(x_{0}) = \frac{\b}{\pi} \sin \Big( \frac{\pi x_{0}}{\beta} \Big)$. Notice that $C^{-1}\|\underline{x}\|_{\beta, L}^{2}\leq d_{L}(x_{1})^{2} + d_{\b}(x_{0})^{2}\leq C\|\underline{x}\|_{\beta, L}^{2}$, for some $C>0$. Let $n_{0}, n_{1}\in \mathbb{N}$, and suppose that $n = n_{0} + n_{1}>0$. We have:
\bea\label{eq:res2}
&&\big|d_{L}(x_{1} - y_{1})^{n_{1}} d_{\b}(x_{0} - y_{0})^{n_{0}}  g^{(\text{bulk})}_{\beta,L,N}(\xx, r; \yy, r')\big|  \nn\\&&\qquad = \Big|\frac{1}{\beta L}\sum_{\underline{k}\in \mathbb{M}^{\text{F}}_{\beta}\times S^{1}_{L}} e^{-i \underline{k}\cdot (\underline{x} - \underline{y})} \partial_{k_{1}}^{n_{1}}\partial_{k_{0}}^{n_{0}} \hat g^{\text{(bulk)}}_{\beta, L,N}(\underline{k}; x_{2}, r; y_{2},  r')\Big| \leq C_{n} \d^{-n-2} e^{-c|x_{2} - y_{2}|}\;.
\eea
The final estimate easily follows from the bound (\ref{eq:res}). Suppose now that $n_{1} = n_{0} = 0$. In this case, the bound (\ref{eq:res}) is not enough to bound $| g^{(\text{bulk})}_{\beta,L,N}(\xx, r; \yy, r') |$. Instead, a uniform bound easily follows from the fact that, for large $k_{0}$, the resolvent is approximately odd in $k_{0}$. In fact, $\sum_{k_{0}} \chi_{N}(k_{0})(-ik_{0} + \hat H(k_{1}) - \m)^{-1} = \sum_{k_{0}} \chi_{N}(k_{0}) (-ik_{0} + \hat H(k_{1}) - \m)^{-1} (\hat H(k_{1}) - \m) (ik_{0} + \hat H(k_{1}) - \m)^{-1}$, which is finite. The validity of the bound (\ref{eq:res2}) for all $n_{0}, n_{1}\in \mathbb{N}$ immediately implies Eq. (\ref{eq:gbd}).

Therefore, the proof is reduced to checking (\ref{eq:res}). To begin, let $n=0$. Consider the second term in the right-hand side of Eq. (\ref{eq:gbd0}). Due to the presence of the cutoff function, the resolvent is gapped. Exponential decay follows from a standard Combes-Thomas estimate, see for instance Proposition 10.5 of \cite{AW}: 
\be\label{eq:A7}
\Big| \Big(\frac{ \chi_{N}(k_{0}) \chi_{\geq}(k_{1}) }{-ik_{0} + \hat H(k_{1}) - \m}\Big)(x_{2}, r; y_{2}, r') \Big| \leq \frac{C}{\sqrt{k_{0}^{2} + \d^{2}}} e^{-c\d|x_{2} - y_{2}|}\;.
\ee
Let us now consider the first term in Eq. (\ref{eq:gbd0}). We write:
\be
\Big(\frac{\chi_{N}(k_{0})\chi_{e}(k_{1})}{-ik_{0} + \hat H(k_{1}) - \m} P_{\perp}^{e}(k_{1})\Big)(x_{2}, r; y_{2}, r')\equiv\Big(\frac{\chi_{N}(k_{0})\chi_{e}(k_{1})}{-ik_{0} + \hat H^e_{\perp}(k_{1}) - \m} P_{\perp}^{e}(k_{1})\Big)(x_{2}, r; y_{2}, r')
\ee
with $\hat H^e_{\perp}(k_{1}) := P_{\perp}^{e}(k_{1}) \hat H(k_{1}) P_{\perp}^{e}(k_{1})$. By construction, $\hat H^{e}_{\perp}(k_{1})$ is gapped in the support of $\chi_{e}(k_{1})$. Also, the kernel $\hat H^{e}_{\perp}(k_{1}; x_{2}, y_{2})$ decays exponentially in $|x_{2} - y_{2}|$, as a consequence of the exponential decay of $P_{\perp}^{e}(k_{1}; x_{2}, y_{2})$ and of the short range of $\hat H(k_{1}; x_{2}, y_{2})$. Therefore, we can again apply the Combes-Thomas bound to get:
\bea\label{eq:gbd6}
\Big|\Big(\frac{\chi_{N}(k_{0})\chi_{e}(k_{1})}{-ik_{0} + \hat H^e_{\perp}(k_{1}) - \m} P_{\perp}^{e}(k_{1})\Big)(x_{2}, r; y_{2}, r')\Big| &\leq& \frac{C}{\sqrt{k_{0}^{2} + \d^{2}}}  e^{-c\d |x_{2} - y_{2}|}\;,
\eea
This concludes the proof of (\ref{eq:res}) for $n=0$. Let us now prove it for $n > 0$. Let us take $n$ derivatives of Eq. (\ref{eq:gbd0}). Consider the contribution coming from differentiation of the second term in (\ref{eq:gbd0}). We get a sum of terms of the form:
\be
A_{\underline{m}} = \partial^{m_{1}}_{k_1} (\chi_{N}(k_{0}) \chi_{\geq}(k_{1})) \frac{1}{-ik_{0} + \hat H(k_{1}) - \m} \partial_{k_{1}}^{m_{2}} \hat H(k_{1}) \frac{1}{-ik_{0} + \hat H(k_{1}) - \m} \cdots \partial_{k_{1}}^{m_{n}} \hat H(k_{1})
\ee
with $\sum_{i} m_{i} = n$. The resolvents are gapped in the support of $\chi_{\geq}(k_{1})$, hence we can apply again the Combes-Thomas bound. Moreover, by the assumption on the Hamiltonian, $\partial^{m}_{k_{1}} \hat H(k_{1}; x_{2}, y_{2})$ is compactly supported in $x_{2} - y_{2}$. Thus, one easily gets $|A_{\underline{m}}(x_{2}, r; y_{2},  r') |\leq C_{\underline{n}} (k_{0}^{2} + \d^{2})^{-\frac{n+1}{2}}e^{-c\d |x_{2} - y_{2}|}$. The terms obtained differentiating the second contribution in the right-hand side of Eq. (\ref{eq:gbd0}) can be bounded in a similar way, we omit the details. This concludes the proof. 
\end{proof}
\begin{rem}
The asymptotic oddness of the resolvent could be used to improve the bound (\ref{eq:gbd}) by replacing $\d^{-2}$ with $\d^{-1}$. We will not need such improvement.
\end{rem}

The next proposition proves that the bulk propagator admits a Gram representation. This allows to apply the Brydges-Battle-Federbush formula, as in Section 5.2 of \cite{GMP}, to prove the convergence of the integration of the bulk degrees of freedom, Eqs. (\ref{eq:Weff})--(\ref{eq:wedge}). The proof of Proposition \ref{prp:gram} is by direct inspection, and will be omitted.

\begin{prop}{\bf (Gram representation.)}\label{prp:gram} Let $g^{(\text{bulk})}_{\beta,L,N} = \sum_{e} g^{(\text{bulk})}_{1,e} + g^{(\text{bulk})}_{2}$, with:
\bea
g^{(\text{bulk})}_{1,e}(\xx, r; \yy, r') &:=& \int_{\beta, L} \frac{d\underline{k}}{(2\pi)^2}\, e^{-i\underline{k}\cdot (\underline{x} - \underline{y})}\Big(\frac{\chi_{N}(k_{0})\chi_{e}(k_{1})}{-ik_{0} + \hat H(k_{1}) - \m}P_{\perp}^{e}(k_{1})\Big)(x_{2},r; y_{2},  r')\nn\\
g^{(\text{bulk})}_{2}(\xx, r; \yy, r') &:=& \int_{\beta, L} \frac{d\underline{k}}{(2\pi)^2}\, e^{-i\underline{k}\cdot (\underline{x} - \underline{y})} \Big(\frac{\chi_{N}(k_{0}) \chi_{\geq}(k_{1})}{-ik_{0} + \hat H(k_{1}) - \m}\Big)(x_{2},r; y_{2},  r')\;.
\eea
Then, 
\bea
g^{(\text{bulk})}_{1,e}(\xx, r; \yy, r') &=& \langle A^e_{1,\xx, r}\, , B^e_{1,\yy, r'} \rangle \equiv \sum_{r''} \int_{\beta, L} d\zz\, \overline{A^e_{1,\xx, r}(\zz, r'')} B^e_{1,\yy, r'}(\zz, r'')\;,\nn\\
g^{(\text{bulk})}_{2}(\xx, r; \yy, r') &=& \langle A_{2,\xx, r}\, , B_{2,\yy, r'} \rangle \equiv \sum_{r''} \int_{\beta, L} d\zz\, \overline{A_{2,\xx, r}(\zz, r'')} B_{2,\yy, r'}(\zz, r'')\;,
\eea
where:
\bea
A^e_{1,\xx, r}(\zz, r'') &:=& \int_{\beta, L} \frac{d\underline{k}}{(2\pi)^2}\, e^{i\underline{k}\cdot (\underline{x} - \underline{z})} \Big(\frac{\chi_{N}(k_{0})^{1/2}\chi_{e}(k_{1})^{1/2}}{k_{0}^{2} + (\hat H(k_{1}) - \m)^2} P_{\perp}^{e}(k_{1})\Big)(z_{2},  r''; x_{2}, r)\\
B^e_{1,\yy, r'}(\zz, r'') &:=& \int_{\beta, L} \frac{d\underline{k}}{(2\pi)^2}\, e^{i\underline{k}\cdot (\underline{y} - \underline{z})} \Big(\chi_{N}(k_{0})^{1/2}\chi_{e}(k_{1})^{1/2}( ik_{0} + \hat H(k_{1}) - \m ) \Big)(z_{2},  r''; y_{2}, r)\nn\\
A_{2,\xx, r}(\zz, r'') &:=& \int_{\beta, L} \frac{d\underline{k}}{(2\pi)^2}\, e^{i\underline{k}\cdot (\underline{x} - \underline{z})} \Big(\frac{\chi_{N}(k_{0})^{1/2}\chi_{\geq }(k_{1})^{1/2}}{k_{0}^{2} + (\hat H(k_{1}) - \m)^2}\Big)(z_{2},  r''; x_{2}, r)\nn\\
B_{2,\yy, r'}(\zz, r'') &:=& \int_{\beta, L} \frac{d\underline{k}}{(2\pi)^2}\, e^{i\underline{k}\cdot (\underline{y} - \underline{z})} \Big(\chi_{N}(k_{0})^{1/2}\chi_{\geq}(k_{1})^{1/2}( ik_{0} + \hat H(k_{1}) - \m ) \Big)(z_{2},  r''; y_{2}, r)\;.\nn
\eea

\end{prop}

\subsection{Proof of Eq. (\ref{eq:2d1d})}\label{app:2d1d}

The next proposition is the key technical result behind Eq. (\ref{eq:2d1d}). Eq. (\ref{eq:2d1d}) follows by Taylor expanding the exponential in the integral, and applying Proposition \ref{prop:repr} to each term in the sum.

\begin{prop}{\bf (One-dimensional representation of the edge field.)}\label{prop:repr} Let $\check{\xi}^{e}_{x_{2}}(x_{1};r)$, $(\psi^{+} * \check{\xi})_{\xx, r}$, $(\psi^{-} * \check{\xi})_{\xx, r}$ as in Eqs. (\ref{eq:conv}), (\ref{eq:checkxi}). Let $\mathcal{E}^{T}_{\Psi}$, $\mathcal{E}^{T}_{\psi}$ be the truncated expectations with respect to the Gaussian Grassmann fields $\Psi^{(\text{edge})}$, $\psi$. Then:
\be\label{eq:equal}
\mathcal{E}^{T}_{\Psi}(\Psi^{\text{(edge)}}(P_{1})\;; \Psi^{\text{(edge)}}(P_{2})\;; \cdots ; \Psi^{\text{(edge)}}(P_{q})) = \mathcal{E}^{T}_{\psi}((\psi*\check{\xi})(P_{1})\;; (\psi*\check{\xi})(P_{2})\;; \cdots ; (\psi*\check{\xi})(P_{q}))\;,
\ee
where $P_{i} = \{f_{j}\}_{j=1}^{|P_{i}|}$ is a set of field labels, and we used the notations $\Psi(P_{i}) = \prod_{f\in P_{i}} \Psi^{\e(f)}_{\xx(f),r(f)}$, $(\psi*\check{\xi})(P_{i}) = \prod_{f\in P_{i}}^{|P_{i}|} (\psi^{\e(f)}*\check{\xi})_{\xx(f), r(f)}$.
\end{prop}

\begin{proof} The proof is based on a separate evaluation of the left-hand side and of the right-hand side of Eq. (\ref{eq:equal}), via the Brydges-Battle-Federbush formula (\ref{eq:BBF}). Let us start with the left-hand side. We have:
\be
\mathcal{E}^{T}_{\Psi}(\Psi^{\text{(edge)}}(P_{1})\;; \Psi^{\text{(edge)}}(P_{2})\;; \cdots \;; \Psi^{\text{(edge)}}(P_{q})) = \sum_{T\in {\bf T}} \alpha_{T} \prod_{\ell \in T} g^{(\text{edge})}_{\ell} \int dP_{T}({\bf t})\, \det G^{(\text{edge})}_{T}({\bf t})
\ee
where $[G^{(\text{edge})}_{T}({\bf t})]_{f,f'} = t_{i(f),i(f')} g^{(\text{edge})}(\xx(f),r(f); \yy(f'),  r'(f'))$. By the Leibniz formula we get:
\be
\det G^{(\text{edge})}_{T}({\bf t}) = \sum_{\pi \in S_{d}} \text{sgn}(\pi) [G^{(\text{edge})}_{T}({\bf t})]_{1,\pi(1)} [G^{(\text{edge})}_{T}({\bf t})]_{2,\pi(2)}\cdots [G^{(\text{edge})}_{T}({\bf t})]_{d,\pi(d)}\;,
\ee
where $d = n - q + 1$ is the dimension of the matrix, with $2n = \sum_{i=1}^{q} |P_{i}|$, and $S_{d}$ is the set of permutations of $\{1,\,\ldots,\, d\}$. It is useful to rewrite the propagator $g^{\text{(edge)}}$ as:
\bea\label{eq:edgeform}
g^{(\text{edge})}(\xx,r; \yy, r') &=& \sum_{e}\int_{\beta, L} \frac{d\underline{k}}{(2\pi)^{2}}\, e^{-i\underline{k}\cdot (\underline{x} - \underline{y})} \xi^{e}_{x_{2}}(k_{1};  r) \overline{\xi^{e}_{y_{2}}(k_{1};  r)} \frac{\chi_{N}(k_{0}) \chi_{e}(k_{1})}{-ik_{0} + \e_{e}(k_{1}) - \m}\nn\\
&\equiv& \sum_{e,e'}\sum_{w_{1}, z_{1}} \check{\xi}^{e}_{x_{2}}(w_{1}; r) \overline{\check{\xi}^{e'}_{y_{2}}(z_{1};  r')} \d_{e,e'} g^{(\text{1d})}_{e}(\underline{x} - \underline{y} - w_{1} + z_{1})\;,
\eea
with $g^{(\text{1d})}_{e}$ given by Eq. (\ref{eq:prop1d}). By using Eq. (\ref{eq:edgeform}), we can write the determinant as:
\be\label{eq:edge21d1}
\det G^{(\text{edge})}_{T}({\bf t}) =  \sum_{\substack{\underline{e}, \underline{e}' \\ \underline{z_1}, \underline{w_1}}}  \prod_{j =1}^{d} \check{\xi}^{e_{j}}_{x_{2,j}}(w_{1,j}; r_{j}) \overline{\check{\xi}^{e'_{j}}_{y_{2,j}}(z_{1,j};   r'_{j})} \sum_{\pi \in S_{d}} \text{sgn}(\pi) \prod_{f'=1}^{d} [G^{\text{(1d)}}_{T}({\bf t}; e, e', w, z)]_{f',\pi(f')}
\ee
where $[G^{\text{(1d)}}_{T}({\bf t}; e, e', w, z)]_{f,f'} := t_{i(f),i(f')} \d_{e(f), e'(f')}g_{e(f)}^{(\text{1d})}(\underline{x}_{f} - \underline{y}_{f'} - w_{1,f} + z_{1,f'})$. Similarly, we have:
\be\label{eq:edge21d2}
\prod_{\ell \in T} g^{(\text{edge})}_{\ell} = \sum_{\underline{e}, \underline{e}', \underline{z_1}, \underline{w_1}}\big[\prod_{\ell \in T} \check{\xi}^{e(\ell)}_{x_{2}(\ell)}(w_{1}(\ell);  r(\ell))\overline{\check{\xi}^{e'(\ell)}_{y_{2}(\ell)}(z_{1}(\ell);   r'(\ell))}\big] \prod_{\ell\in T} g^{\text{(1d)}}_{\ell}(e,e',w,z)\;,
\ee
where, if $\ell = (f,f')$, $g^{\text{(1d)}}_{\ell}(e,e',w,z) := \d_{e(f), e'(f')} g^{\text{(1d)}}_{e(f)}(\underline{x}(f) - \underline{y}(f') - w_{1}(f) + z_{1}(f'))$. Eqs. (\ref{eq:edge21d1}), (\ref{eq:edge21d2}) imply:
\bea\label{eq:edge21d3}
&&\mathcal{E}^{T}_{\Psi}(\Psi^{\text{(edge)}}(P_{1})\;; \Psi^{\text{(edge)}}(P_{2})\;; \cdots \;; \Psi^{\text{(edge)}}(P_{q})) = \sum_{T\in {\bf T}} \alpha_{T}  \sum_{\substack{\underline{e}, \underline{e}'\\ \underline{z_1}, \underline{w_1}}}\big[\prod_{i=1}^{n}  \check{\xi}^{e_{i}}_{x_{2,i}}(w_{1,i};  r_{i}) \overline{\check{\xi}^{e'_{i}}_{y_{2,i}}(z_{1,i};  r'_{i})}\big] \nn\\ &&\quad\qquad\qquad\qquad\qquad\qquad\qquad \cdot \prod_{\ell\in T} g^{\text{(1d)}}_{\ell}(e,e',w,z) \int dP_{T}({\bf t})\, \det G^{\text{(1d)}}_{T}({\bf t}; e, e', w, z)\;.
\eea
Now, it is not difficult to see that, by the multilinearity of the truncated expectation, (\ref{eq:edge21d3}) is exactly what one gets after applying the Brydges-Battle-Federbush formula to the right-hand side of Eq. (\ref{eq:equal}). This concludes the proof.
\end{proof}

\section{Wick rotation}\label{app:wick}

Here we prove the Wick rotation, Proposition \ref{prp:wick}. The proof is based on the combination of the ideas of \cite{GMP}, with the ideas of the proof of Theorem 5.4.12 of \cite{BR}. The starting point is the following lemma.

\begin{lemma}\label{lem:wick} Let $A,\, B$ be two operators on $\mathcal{F}_{L}$, such that\footnote{Notice that $\langle A(-it) A^{*} \rangle_{\beta, L}\geq 0$, $\langle B^{*}(-it) B \rangle_{\beta, L} \geq 0$. In fact, by cyclicity of the trace, $\langle A(-it) A^{*} \rangle_{\beta, L} = \langle A(-it/2) A(-it/2)^{*} \rangle_{\beta, L} \geq 0$, $\langle B^{*}(-it) B \rangle_{\beta, L} = \langle B^{*}(-it/2) B(-it/2) \rangle_{\beta, L} \geq 0$.}:
\bea\label{eq:AB}
&&\frac{1}{L}{\langle} A(-it) A^*{\rangle}_{\beta, L} \leq C\;,\qquad \frac{1}{L}{\langle} B^*(-it) B {\rangle}_{\beta,L}\leq C\;,\qquad \forall t\in [0,\beta]\nn\\
&& \int_{0}^{\beta} dt\, \frac{1}{L}\langle A(-it) A^{*} \rangle_{\beta, L} \leq C\;,\qquad \int_{0}^{\beta} dt\, \frac{1}{L}\langle B^{*}(-it) B \rangle_{\beta, L} \leq C\;,
\eea
for some constant $C>0$ independent of $\beta, L$. Let $T>0$, $\eta >0$. Then,
\be
\int_{-T}^{0} dt\, e^{\eta t} \frac{1}{L}{\langle} [A(t), B] {\rangle}_{\beta, L} = i\int_{0}^{\beta} dt\, e^{-i\eta_{\beta} t} \frac{1}{L}{\langle} A(-it) B {\rangle}_{\beta, L} + \mathcal{E}_{\beta, L}(T,\eta)\;,
\ee
where $\eta_{\beta} \in \frac{2\pi}{\beta}\mathbb{Z}$ is such that $|\eta - \eta_{\beta}| = \text{min}_{\eta'\in \frac{2\pi}{\beta}\mathbb{Z}} |\eta - \eta'|$, and where the error term satisfies the bound $\big| \mathcal{E}_{\beta, L}(T,\eta) \big| \leq C( 1/(\eta^2 \beta) + e^{-\eta T})$.
\end{lemma}
\begin{proof} We start by writing:
\bea\label{eq:W1}
\int_{-T}^{0} dt\, e^{\eta t} {\langle} [ A(t) ,B ] {\rangle}_{\beta, L} &=& \int_{-T}^{0} dt\, \big[ e^{\eta t} {\langle} A(t) B {\rangle}_{\beta, L} - e^{\eta t} {\langle} B A(t) {\rangle}_{\beta, L}\nn\\
&=& \int_{-T}^{0} dt\, \big[ e^{\eta t} {\langle} A(t) B {\rangle}_{\beta, L} - e^{\eta t} {\langle} A(t - i\beta) B {\rangle}_{\beta, L} \big]\;,
\eea
where in the last line we used that, by cyclicity of the trace $\Tr e^{-\beta \mathcal{H}} B A(t) = \Tr A(t) e^{-\beta \mathcal{H}} B = \Tr e^{-\beta \mathcal{H}} A(t - i\beta) B$. Let $\eta_{\beta}$ be the closest element of $\frac{2\pi}{\beta}\mathbb{Z}$ to $\eta$: $|\eta - \eta_{\beta}| = \min_{\eta'\in \frac{2\pi}{\beta}\mathbb{Z}} |\eta - \eta'|$. We rewrite (\ref{eq:W1}) as, using that $e^{i\eta_{\beta}\beta} = 1$:
\be\label{eq:W11}
\int_{-T}^{0} dt\, e^{\eta t} \frac{1}{L}{\langle} [ A(t) ,B ] {\rangle}_{\beta, L} = \int_{-T}^{0} dt\, \big[ e^{\eta_{\beta} t} \frac{1}{L}{\langle} A(t) B {\rangle}_{\beta, L} - e^{\eta_{\beta} (t - i\beta)} \frac{1}{L}{\langle} A(t - i\beta) B {\rangle}_{\beta, L} \big] + \mathcal{E}^{(1)}_{\beta, L}(T, \eta)\;,
\ee
where the error term can be estimated as, by Cauchy-Schwarz inequality and thanks to the assumption (\ref{eq:AB}), for some $C>0$:
\bea\label{eq:E1}
\big|\mathcal{E}^{(1)}_{\beta, L}(T, \eta)\big| &\leq& \frac{1}{L}\int_{-T}^{0}dt\, \big|e^{\eta t} - e^{\eta_{\beta} t}\big| \big[ {\langle} A A^{*} {\rangle}_{\beta , L}^{1/2}{\langle} B^* B {\rangle}_{\beta , L}^{1/2} + {\langle} A(-i\beta) A^* {\rangle}_{\beta , L}^{1/2}{\langle} B^*(-i\beta) B {\rangle}_{\beta , L}^{1/2}  \big]\nn\\
&\leq& \frac{C}{\eta^2 \beta}\;.
\eea
Now, we notice that the function $e^{\eta_{\beta} z} {\langle} A(z) B {\rangle}_{\beta,L}$ is {\it entire} in $z\in \mathbb{C}$. Therefore, by Cauchy theorem we conclude that the integral of $e^{\eta_{\beta} z} {\langle} A(z) B {\rangle}_{\beta,L}$ along the boundary of the complex rectangle $(0,0)\to (0,-i\beta)\to (-T,-i\beta) \to (-T,0)$ is equal to zero. The two term in Eq. (\ref{eq:W11}) correspond respectively to the paths $(-T, 0)\to (0,0)$ and $(0, -i\beta) \to (-T, -i\beta)$. Thus we can write:
\be
\int_{-T}^{0} dt\, \big[ e^{\eta_{\beta} t} \frac{1}{L}{\langle} A(t) B {\rangle}_{\beta, L} - e^{\eta_{\beta} (t - i\beta)} \frac{1}{L}{\langle} A(t - i\beta) B {\rangle}_{\beta, L} \big]  =  i\int_{0}^{\beta} dt\, e^{-i\eta_{\beta} t} \frac{1}{L}{\langle} A(-it) B {\rangle}_{\beta, L} + \mathcal{E}^{(2)}_{\beta, L}(T, \eta)
\ee
where the error term collects the contribution of the path $(-T,-i\beta)\to (-T,0)$,
\bea\label{eq:E2}
\mathcal{E}^{(2)}_{\beta, L}(T, \eta) &=& i\int_{0}^{\beta} dt\, e^{\eta_{\beta}(-T -it)} \frac{1}{L}{\langle} A(-T -it) B {\rangle}_{\beta,L}\nn\\
\big| \mathcal{E}^{(2)}_{\beta, L}(T, \eta) \big| &\leq& e^{-\eta T} \frac{1}{L}\int_{0}^{\beta}dt\, {\langle} A(-it)A^* {\rangle}_{\beta,L}^{1/2}{\langle} B^*(-it) B {\rangle}_{\beta,L}^{1/2}\\
&\leq& e^{-\eta T} \Big(\int_{0}^{\beta}dt\, \frac{1}{L} {\langle} A(-it)A^* {\rangle}_{\beta,L}\Big)^{1/2}\Big(\int_{0}^{\beta}dt\, \frac{1}{L} {\langle} B^{*}(-it)B {\rangle}_{\beta,L}\Big)^{1/2} \leq Ce^{-\eta T}\;.\nn
\eea
The first bound follows from Cauchy-Schwarz inequality for $\langle \cdot \rangle_{\beta, L}$, the second bound follows from Cauchy-Schwarz inequality for the time integration, while last follows from the assumptions Eq. (\ref{eq:AB}). All in all we have:
\be
\int_{-T}^{0} dt\, e^{\eta t} \frac{1}{L}{\langle} [ A(t) ,B ] {\rangle}_{\beta, L} = i\int_{0}^{\beta} e^{-i\eta_{\beta} t} \frac{1}{L}{\langle} A(-it) B {\rangle}_{\beta, L} + \sum_{i=1}^{2} \mathcal{E}^{(i)}_{\beta, L}(T, \eta)
\ee
where, by Eqs. (\ref{eq:E1}), (\ref{eq:E2}), $\big| \sum_{i=1}^{2} \mathcal{E}^{(i)}_{\beta, L}(T, \eta) \big| \leq C( 1/(\eta^2 \beta) + e^{-\eta T})$. This concludes the proof.
\end{proof}

We are now ready to prove Proposition \ref{prp:wick}. 

\medskip

{\it Proof of Proposition \ref{prp:wick}.} Let us start by studying the charge transport coefficients. Consider $G^{\underline{a}}_{01}(\underline{p})$. The proof for the other coefficients will be exactly the same, with the exception of $G^{\underline{a}}_{11}(\underline{p})$, that will be discussed below. Let us first check the assumptions of Lemma \ref{lem:wick}. Let $p_{1}\neq 0$. We set $A := \hat \rho_{p_{1}}^{\leq a} - \langle \hat \rho_{p_{1}}^{\leq a}\rangle_{\beta, L}$, $B := \hat j_{1,-p_{1}}^{\leq a'} -\langle \hat j_{1,-p_{1}}^{\leq a'}\rangle_{\beta, L}$ (a shift by a constant does not change the commutators). The first assumption in Eq. (\ref{eq:AB}) can be checked using the bound in Eq. (\ref{eq:JJbd}),
\bea
\frac{1}{L} \langle A(-it) A^{*} \rangle_{\beta, L} \equiv \frac{1}{L} \langle \hat \rho_{p_{1}}^{\leq a}(-it)\,; \hat \rho_{-p_{1}}^{\leq a}  \rangle_{\beta, L} &\leq& \frac{1}{L}\sum_{x_{2}, y_{2}\leq a} \sum_{x_{1}, y_{1} = 0}^{L} \frac{1}{1 + t^{2} + |x_{1} - y_{1}|_{L}^{2}} \frac{C_{n}}{1 + |x_{2} - y_{2}|^{n}}\nn\\ &=& Ca
\eea
A similar bound holds true for $B$. This proves the validity of the first two assumptions in Eq. (\ref{eq:AB}). Consider now the second assumptions in Eq. (\ref{eq:AB}). We notice that, for $\m=0,1$:
\be
\int_{0}^{\beta} dt\, \frac{1}{L}\langle\hat j_{\m,p_{1}}^{\leq a}(-it)\,;  \hat j_{\m, -p_{1}}^{\leq a} \rangle_{\beta, L} \equiv \frac{1}{\beta L}\langle {\bf T}\, \hat j_{\m,(0,p_{1})}^{\leq a}\,;  \hat j_{\m, (0,-p_{1})}^{\leq a} \rangle_{\beta, L}\;,
\ee
where $\hat j_{\m,(0,p_{1})}^{\leq a} = \int_{0}^{\beta} dx_{0} \sum_{x_{1} = 0}^{L}\sum_{x_{2} = 0}^{a} e^{i\underline{p}\cdot \underline{x}} j_{\m,\xx}$. By the existence of the $\beta, L\to \infty$ limit of the Euclidean correlations, Lemma 2.6 of \cite{BFMhub1}, we know that 
\be
\frac{1}{\beta L}\langle {\bf T}\, \hat j_{\m,(0,p_{1})}^{\leq a}\,;  \hat j_{\m, (0,-p_{1})}^{\leq a} \rangle_{\beta, L} = \pmb{\langle} {\bf T}\, \hat j_{\m,(0,p_{1})}^{\leq a}\,;  \hat j_{\m, (0,-p_{1})}^{\leq a} \pmb{\rangle}_{\infty} + O(a \beta^{-\theta})\qquad \text{for some $\theta > 0$}; 
\ee
thus, thanks to Proposition \ref{prp:relref}, this allows to prove that $|\frac{1}{\beta L}\langle {\bf T}\, \hat j_{\m,(0,p_{1})}^{\leq a}\,;  \hat j_{\m, (0,-p_{1})}^{\leq a} \rangle_{\beta, L}|\leq Ca$. A similar bound holds for $B$. This concludes the check of the last two assumptions in Eq. (\ref{eq:AB}). Hence, we can apply Lemma \ref{lem:wick} to analytically continue the real times transport coefficient to imaginary times, for finite $\beta, L$. The existence of the $\beta, L\to \infty$ limit can be proven as in \cite{GMP}, item $ii)$ of Proposition 6.1 (see also Proposition 5.2 of \cite{MP}); it follows from the existence of the infinite volume dynamics, that can be proven via Lieb-Robinson bounds, and from the existence of the zero temperature, infinite volume Gibbs state, that can be proven with the RG methods discussed in this paper. Taking also the $T\to\infty$ limit, Eq. (\ref{eq:Dkk'0}) follows. The proof for the spin transport coefficients $G^{\underline{a},s}_{\m\n}(\underline{p})$ is exactly the same. This concludes the proof of Proposition \ref{prp:wick}.\qed

\end{document}